\documentclass[format=acmsmall, screen, nonacm]{acmart}

\AtBeginDocument{%
  }

\usepackage[nameinlink]{cleveref}
\usepackage{enumitem}[shortlabels]

\makeatletter
\AtEndPreamble{%
  \theoremstyle{acmdefinition}
  \@ifundefined{rem}{%
    \newtheorem{rem}[theorem]{Remark}
  }{}%
  \crefname{enumi}{}{}
  \Crefname{enumi}{}{}
  \crefname{lemma}{Lemma}{Lemmas}
  \Crefname{lemma}{Lemma}{Lemmas}
  \crefname{definition}{Definition}{Definitions}
  \Crefname{definition}{Definition}{Definitions}
}
\makeatother

\usepackage{dsfont}
\usepackage{adjustbox}
\usepackage{dashbox}
\usepackage{xspace}

\usepackage{amsmath}
\usepackage{cleveref}

\usepackage{tikz}
\usetikzlibrary{cd, angles, shapes}
\usetikzlibrary{decorations.pathreplacing,angles,quotes}

\newenvironment{varitemize}
{
	\begin{list}{\labelitemi}
		{\setlength{\itemsep}{0pt}
			\setlength{\topsep}{0pt}
			\setlength{\parsep}{0pt}
			\setlength{\partopsep}{0pt}
			\setlength{\leftmargin}{15pt}
			\setlength{\rightmargin}{0pt}
			\setlength{\itemindent}{0pt}
			\setlength{\labelsep}{5pt}
			\setlength{\labelwidth}{10pt}
	}}
	{
	\end{list} 
}

\newcounter{numberone}

\newcommand{\Quipper}{\texttt{Quipper}}
\newcommand{\Cirq}{\texttt{Cirq}}
\newcommand{\Qiskit}{\texttt{Qiskit}}

\newcommand{\shortv}[1]{}

\newcommand{\qlambda}{\ensuremath{\lambda^Q}\xspace}
\newcommand{\Qlambda}{\ensuremath{\Lambda}^Q\xspace}
\newcommand{\qcirc}{\ensuremath{\mathcal{QC}}\xspace}
\newcommand{\pv}[2]{\langle #1,#2 \rangle}

\newcommand{\letv}[3]{{\tt let}\,{#1}={#2}~{\tt in}~{#3}}
\newcommand{\starc}{*}
\newcommand{\letstar}[2]{{\tt let}\,{\starc}={#1}~{\tt in}~{#2}}

\newcommand{\tc}{\mathtt{t}\!\mathtt{t}}
\newcommand{\termZero}{\ensuremath{\mathtt{zero}}}
\newcommand{\termOne}{\ensuremath{\mathtt{one}}}
\newcommand{\fc}{\mathtt{f}\!\mathtt{f}}
\newcommand{\ite}[3]{\ensuremath{\mathtt{ if }~#1~\mathtt{ then }~#2~\mathtt{ else }~#3}}

\newcommand{\meas}{\ensuremath{\mathtt{meas}}\xspace}
\newcommand{\discard}{\ensuremath{\mathtt{discard}}}
\newcommand{\new}{\ensuremath{\mathtt{new}}\xspace}
\newcommand{\qgate}[1]{\ensuremath{\mathtt{#1}}}
\newcommand{\evalCtx}{\ensuremath{\mathtt{E}}}

\newcommand{\DIST}[1]{\mathrm{eval}_{\lambda}[#1]}
\newcommand{\DISTM}[1]{\mathrm{eval}_{\mathrm{QMSIAM}}[#1]}

\newcommand{\LITE}{\mathcal L^{\mathrm{ite}}}
\newcommand{\LITES}{\mathcal L^{\mathrm{ite}^*}}
\newcommand{\LMEAS}{\mathcal L^{\mathrm{meas}}}
\newcommand{\ADD}{\mathcal G}
\newcommand{\ADDH}{\mathcal H}
\newcommand{\EXT}{\mathrm{Ext}}
\newcommand{\DISC}[2]{\discard^{#1}(#2)}

\newcommand{\init}{\ensuremath{\mathtt{Init}}\xspace}
\newcommand{\fin}{\ensuremath{\mathtt{Fin}}\xspace}
\newcommand{\exec}{\ensuremath{\mathtt{Exec}}\xspace}
\newcommand{\config}{\ensuremath{\mathcal{C}}\xspace}
\newcommand{\donfig}{\ensuremath{\mathcal{D}}\xspace}
\newcommand{\cat}[1]{\ensuremath{\mathbf{#1}}\xspace}
\newcommand{\circuittkm}{\ensuremath{\textsf{QCSIAM}^!}\xspace}
\newcommand{\circuittkmz}{\ensuremath{\textsf{QCSIAM}}\xspace}
\newcommand{\evaltkm}{\ensuremath{\textsf{QMSIAM}}\xspace}

\newcommand{\baseT}{\ensuremath{\mathbb{B}}}
\newcommand{\baseTalt}{\ensuremath{\mathbb{A}}}
\newcommand{\qbit}{\ensuremath{\mathtt{qbit}}}
\newcommand{\bit}{\ensuremath{\mathtt{bit}}}

\newcommand{\inputType}[1]{\ensuremath{\mathtt{in}(#1)}}
\newcommand{\outputType}[1]{\ensuremath{\mathtt{out}(#1)}}

\newcommand{\positive}{\ensuremath{\mathrm{PDATA}}}
\newcommand{\posones}{\ensuremath{\mathrm{PONES}}}
\newcommand{\negative}{\ensuremath{\mathrm{NDATA}}}
\newcommand{\negones}{\ensuremath{\mathrm{NONES}}}
\newcommand{\ones}{\ensuremath{\mathrm{ONES}}}
\newcommand{\guard}{\ensuremath{\mathrm{GUARD}}}

\newcommand{\one}{\ensuremath{\mathds{1}}}

\newcommand{\MM}{\ensuremath{\mathcal{M}}}
\newcommand{\NN}{\ensuremath{\mathcal{N}}}
\newcommand{\GG}{\ensuremath{\mathcal G}}
\newcommand{\complex}{\ensuremath{\mathbb{C}}}

\newcommand{\boolT}{\ensuremath{\mathbb{B}}}

\usepackage{mathtools}
\usepackage{etoolbox}
\usepackage{stmaryrd}
\DeclarePairedDelimiter\bbracket{\llbracket}{\rrbracket}
\DeclarePairedDelimiter\bparenthesis{\llparenthesis}{\rrparenthesis}
\newcommand{\opn}[1]{\ensuremath{\operatorname{#1}}}
\newcommand{\interp}[1]{\ensuremath{\bbracket{#1}}}

\newcommand{\pinterp}[1]{\ensuremath{\bparenthesis{#1}}}
\newcommand{\translate}{\ensuremath{\tau}}

\newcommand{\totkm}{\ensuremath{\to_{\mathtt{tkm}}}}
\newcommand{\set}[1]{\ensuremath{\{#1\}}}

\newcommand{\ket}[1]{\ensuremath{\left|  #1 \right\rangle}}
\newcommand{\bra}[1]{\ensuremath{\left\langle  #1 \right|}}

\newcommand{\QO}{\mathbb{Q}}

\newcommand{\QOrule}{\mathtt{Q}}

\newcommand{\Id}{\ensuremath{\mathtt{Id}}}

\newcommand{\alt}{~\mid~}

\newcommand{\eqdef}{\coloneqq}
\newcommand{\myparagraph}[1]{\emph{#1}. }
\newcommand{\tocpm}{\mathtt{mix}}

\newcommand{\initc}{\ensuremath{\init_{\evaltkm}}}
\newcommand{\initm}{\initc}
\newcommand{\finc}{\ensuremath{\fin_{\evaltkm}}}

\newcommand{\initcc}{\ensuremath{\init_{\circuittkm}}}
\newcommand{\initcz}{\ensuremath{\init_{\circuittkmz}}}
\newcommand{\fincc}{\ensuremath{\fin_{\circuittkm}}}

\newcommand{\mtocl}{\ensuremath{\rightrightarrows}}

\usepackage{proof}
\usepackage{todonotes}
\setuptodonotes{inline}
\usepackage{longtable}
\usepackage{multirow}
\usepackage{nicematrix}
\usepackage{quantikz2}
\usepackage{tikz}
\usetikzlibrary{circuits.ee.IEC}
\usepackage{subcaption}

\usetikzlibrary{calc}
\usetikzlibrary{arrows.meta}
\usetikzlibrary{shapes, patterns, decorations.pathreplacing, decorations.text}
\usetikzlibrary{trees}
\usetikzlibrary{positioning}
\usetikzlibrary{tikzmark, cd}

\tikzset{
  node rotated/.style = {rotate=180},
  border rotated/.style = {shape border rotate=180},
  downtriangle/.style = {fill=white, draw=black, regular polygon, regular polygon sides=3, border rotated},
  triangle/.style = {fill=white, draw=black, regular polygon, regular polygon sides=3}
}

\usetikzlibrary{decorations.markings,arrows,decorations.pathmorphing}
\usetikzlibrary{shapes.misc,matrix,backgrounds,folding,positioning}
\usetikzlibrary{shapes.geometric}
\usetikzlibrary{shapes.multipart}
\pgfdeclarelayer{edgelayer}
\pgfdeclarelayer{nodelayer}
\pgfsetlayers{background,edgelayer,nodelayer,main}
\tikzset{every path/.style={draw=black!80, line width=0.6pt}}
\tikzstyle{every picture}=[baseline=-0.25em]
\tikzstyle{none}=[inner sep=0mm]
\tikzstyle{box}=[fill=white, draw=black, shape=rectangle]
\tikzstyle{arrow}=[decoration={markings,mark=at position 1 with
    {\arrow[scale=1.2,>=stealth]{>}}},postaction={decorate}]

\tikzstyle{box}=[rectangle, draw=black, fill=white, inner sep=1pt, font=\scriptsize]
\tikzstyle{very thick}=[-, line width=1pt]

\pgfdeclarelayer{edgelayer}
\pgfdeclarelayer{nodelayer}
\pgfsetlayers{background,edgelayer,nodelayer,main}
\tikzstyle{every loop}=[]

\newcommand{\tikzfigpathValue}{./figs}

\newcommand{\tikzfig}[1]{%
\IfFileExists{#1.tikz}
  {\input{#1.tikz}}
  {%
    \IfFileExists{\tikzfigpathValue/#1.tikz}
      {\input{\tikzfigpathValue/#1.tikz}}
      {\tikz[baseline=-0.5em]{\node[draw=red,font=\color{red},fill=red!10!white] {\textit{\tikzfigpathValue/#1}};}}%
  }%
}

\tikzset{mymeter/.append style={draw, inner sep=10, rectangle,
 font=\vphantom{A}, minimum width=30, line width=.8, path
 picture={\draw[black] ([shift={(.1,.3)}]path picture bounding
 box.south west) to[bend left=50] ([shift={(-.1,.3)}]path picture
 bounding box.south east);\draw[black,-latex] ([shift={(0,.1)}]path
 picture bounding box.south) -- ([shift={(.3,-.1)}]path picture
 bounding box.north);}}}

\newcommand{\myground}
{
\begin{tikzpicture}[circuit ee IEC,yscale=0.9,xscale=0.8]	 	 
\draw[solid,arrows=-] (-1ex,0) to (0,0) node[anchor=center,ground,rotate=0,xshift=.80ex] {};
\end{tikzpicture}}

\usepackage{ebproof}

\usepackage{mathrsfs}

\newcommand{\leaf}[1]{\ensuremath{\opn{leaf}(#1)}}
\newcommand{\node}[2]{\ensuremath{\opn{node}(#1, #2)}}

\newcommand{\seq}{\mathbin{;}}

\usepackage[normalem]{ulem}

\usepackage[toc,page,header]{appendix}

\usepackage{titletoc}

\makeatletter
  \def\ttl@Hy@steplink#1{%
    \Hy@MakeCurrentHrefAuto{#1*}%
    \edef\ttl@Hy@saveanchor{%
      \noexpand\Hy@raisedlink{%
        \noexpand\hyper@anchorstart{\@currentHref}%
        \noexpand\hyper@anchorend
        \def\noexpand\ttl@Hy@SavedCurrentHref{\@currentHref}%
        \noexpand\ttl@Hy@PatchSaveWrite
      }%
    }%
  }%
  \def\ttl@Hy@PatchSaveWrite{%
    \begingroup
      \toks@\expandafter{\ttl@savewrite}%
      \edef\x{\endgroup
        \def\noexpand\ttl@savewrite{%
          \let\noexpand\@currentHref
              \noexpand\ttl@Hy@SavedCurrentHref
          \the\toks@
        }%
      }%
    \x
  }%
  \def\ttl@Hy@refstepcounter#1{%
    \let\ttl@b\Hy@raisedlink
    \def\Hy@raisedlink##1{%
      \def\ttl@Hy@saveanchor{\Hy@raisedlink{##1}}%
    }%
    \refstepcounter{#1}%
    \let\Hy@raisedlink\ttl@b
  }%
\def\ttl@gobblecontents#1#2#3#4{\ignorespaces}%
\makeatother

\newcommand\DoToC{%
\startcontents[sections]
\printcontents[sections]{l}{1}{\setcounter{tocdepth}{3}}
}

\begin{document}

\title{Compiling Quantum $\lambda$-Terms into Circuits\\ via the Geometry of Interaction}

\author{Kostia Chardonnet}
\email{kostia.chardonnet@pm.me}
\orcid{0009-0000-0671-6390}
\affiliation{%
  \institution{Université de Lorraine, CNRS, Inria, LORIA, F-54000}
  \city{Nancy}
  \country{France}
}

\author{Ugo Dal Lago}
\email{ugo.dallago@unibo.it}
\orcid{0000-0001-9200-070X}
\email{}
\affiliation{%
  \institution{University of Bologna, Italy and INRIA Sophia Antipolis, France}
  \city{Bologna}
  \country{Italy}
}

\author{Naohiko Hoshino}
\email{nhoshino@cis.sojo-u.ac.jp}
\orcid{0000-0003-2647-0310}
\affiliation{%
  \institution{Sojo University}
  \city{Sojo}
  \country{Japan}
}

\author{Paolo Pistone}
\orcid{0000-0003-4250-9051}
\email{paolo.pistone@ens-lyon.fr}
\affiliation{%
  \institution{Université Claude Bernard Lyon 1}
  \city{Lyon}
  \country{France}
}

\renewcommand{\shortauthors}{Chardonnet et al.}

\begin{abstract}
  We present an algorithm turning any term of a linear
	quantum $\lambda$-calculus into a quantum circuit.
	The essential ingredient behind the proposed algorithm is
	Girard's geometry of interaction, which, differently from its
	well-known uses from the literature, is here leveraged to
	perform as much of the \emph{classical} computation as possible,
	at the same time producing a circuit that, when evaluated, performs all the
	\emph{quantum} operations in the underlying $\lambda$-term.
	We identify higher-order control flow as the primary obstacle 
	towards efficient solutions to the problem at hand. Notably, geometry
	of interaction proves sufficiently flexible to enable efficient compilation
	in many cases, while still supporting a total compilation procedure.
	Finally, we characterize through a type system those $\lambda$-terms
	for which compilation can be performed efficiently.
\end{abstract}

\keywords{Quantum computation, lambda-calculus, compilation, geometry of interaction}

\settopmatter{printacmref=false}
\setcopyright{none}
\renewcommand\footnotetextcopyrightpermission[1]{}
\pagestyle{plain}

\maketitle

\section{Introduction}
\label{sect:intro}
Quantum computing holds immense potential to revolutionize various
areas of computer science by solving complex problems much faster than
classical computing~\cite{grover1996fast,shor1999polynomial}. This is
due to its ability to leverage the power of quantum bits, thanks to
\emph{superposition} and \emph{entanglement}. One of
the fascinating aspects of quantum computing is the diversity of model
architectures being explored. These include gate-based quantum
computing, but also quantum annealing~\cite{quantumannealing}, topological
quantum computing~\cite{topological}, and others.

On the side of high-level quantum programming languages, the QRAM
model of quantum computation~\cite{Knill96,WY2023} is among the
most successful ones. There, a classical computing machine interacts
with a quantum processor by instructing the latter to create new
qubits, apply some unitary transformations to the existing qubits, or
measure (some of) them. In other words, programs have classical 
control flow, but can also communicate with a quantum device that acts on quantum data by applying unitary
operations or measurements. This is conceptually similar  to what happens 
when programming in presence of an external storage device. 
Thanks to measurements, however, the
overall evolution (and the computation's final result) can be
\emph{probabilistic}. Moreover, the classical computer and the quantum
processor do not have the same nature: while the former is a purely
classical device, the latter's internal state consists of a finite
number of qubits, each of them manipulated following the postulates of
quantum mechanics. In particular, quantum bits cannot be erased nor
duplicated, and the quantum processor can only perform operations 
of a very specific shape.

Based on this model, several \emph{quantum programming languages} have
been developed, from assembly
code~\cite{cross2017openquantumassemblylanguage,Cross_2022}, to
imperative programming languages with loops and classical
tests~\cite{feng2021quantum}, to functional programming languages and
$\lambda$-calculi~\cite{selinger2009quantum,selinger2008fully,DLMZ2009}.
All these languages follow the QRAM model: programs can 
perform classical operations internally, but also interact with an external 
quantum device, instructing it to apply quantum operations to one or
more of the stored qubits. In particular, between any two interaction points, the underlying
classical machine can perform an arbitrary amount of classical work.  
A positive aspect of this programming style is that the results 
of measurements are made available in the form of ``first-class'' 
Booleans and can therefore influence the control flow. This arguably helps 
to make programs simple and readable.

However, a fracture exists between the aforementioned programming
paradigm, and the reality of quantum computing. Most quantum
architectures are currently very hard to be accessed interactively,
and instead require \emph{a whole quantum circuit} to be sent to them. The
latter is first created in its entirety by a classical algorithm, and
then processed and optimized for the specific architecture through
\emph{transpilation}. As a result, quantum programs are, in
practice, often written down either as programs representing
\emph{one} quantum
circuit~\cite{cross2017openquantumassemblylanguage}, or as programs in
so-called \emph{circuit description languages}, like
\Qiskit~\cite{qiskit2024}, \Cirq~\cite{cirq}, or
\Quipper~\cite{Quipper1,Quipper2}. These languages manipulate circuits
as ordinary data structures, often taking the form of libraries for
mainstream programming languages. This separates circuit
\emph{construction} from circuit \emph{execution}, enabling offline
optimization and transpilation. Programming is thus seen as the direct
construction of quantum circuits and the model is no longer the one of
QRAMs: the program's task consists in building a circuit, and not of
interacting with a quantum device.  This way, however, the possibility
of making partial results of a quantum computation available to the
programmer is lost: the only way to run a circuit is to execute it
\emph{in its entirety}, even if the so-called \emph{dynamic lifting}
techniques can partially alleviate this
burden~\cite{protoQuipperDynamicLifting}.

It is thus natural to ask whether it is possible to reconcile the QRAM model
with the need to produce complete circuits, thereby targeting existing hardware
architectures. In fact, this work provides an answer to the following question:
would it be possible to \emph{compile} programs written in QRAM languages, in
particular functional languages with higher-order
functions~\cite{selinger2009quantum,selinger2008fully,DLMZ2009}, down to quantum
circuits, along the way \emph{performing} as much of the classical work as
possible, but \emph{delegating} all the quantum operations to the target
circuit? Moreover, would it be possible to do that \emph{efficiently} in
presence of both measurements and conditionals? As detailed in
Section~\ref{sect:informal} below, this is not trivial and cannot be easily
solved by relying on the usual, rewriting based operational semantics of the
underlying language. Indeed, a conditional whose branches have a higher-order
type cannot be easily and efficiently compiled into a conditional of the target
circuit language: as soon as a branching is \emph{opened}, it is difficult to
unify the two branches, effectively \emph{closing} the branching. This can give
rise to an exponential blow-up in the size of the underlying circuit, a
well-known phenomenon which shows up, e.g., in the related problem of
dynamically
lifting~\cite{protoQuipperDynamicLifting,colledan_et_al:LIPIcs.TYPES.2022.3,dongho21}
the value of a Boolean in circuit description languages.

\myparagraph{Contributions}
In this work we propose a new compilation procedure for the linear
fragment of a quantum $\lambda$-calculus in the style of Selinger \&
Valiron's quantum $\lambda$-calculus~\cite{selinger2009quantum,
10.1007/11417170_26} onto a paradigmatic quantum circuit model
corresponding to a fragment of the QASM
language~\cite{cross2017openquantumassemblylanguage}. Our
procedure gets rid of {\em both} higher-order operations \emph{and}
classical control, providing an effective solution to the
following problem: given a well-typed term of first-order type, in
which higher-order functions and conditionals can possibly occur,
extract the underlying circuit, effectively anticipating the classical
work as much as possible.

The fundamental
ingredient of our compilation scheme is Girard's \emph{Geometry of
Interaction} (GoI for short)~\cite{goi0,goi1}, a semantic framework
for linear logic which is well-known to induce 
\emph{abstract machines}
\cite{Mackie1995,danos1999reversible,Accattoli2020} and circuit
\emph{synthesis} algorithms~\cite{Ghica2007}. The basic idea behind the
GoI translation can be described as follows: given a type derivation
$\pi$ with conclusion $\Gamma \vdash M:A$, one introduces a finite set
of tokens which can  travel along occurrences of base types  
throughout $\pi$; the paths followed by
such tokens then give rise to a circuit relating \emph{positive} and
\emph{negative} occurrences of base types in $\Gamma$ and $A$,
which is constructed incrementally, following the paths of tokens 
inside $\pi$. For instance, a
type derivation with conclusion
$x:\qbit, y:\qbit \vdash M: (\qbit\multimap\qbit)\multimap \qbit$
(where $M$ might indeed contain higher-order operations as well as
conditionals) will, after compilation, give rise to a quantum
circuit with three input qubits and two output qubits.

Our compilation procedure works in two steps: first, the typing
derivation is translated via GoI onto a language for quantum circuits
\emph{with} control flow. Then, a second compilation step takes
place, which translates the circuit onto a proper quantum circuit, notably
\emph{eliminating} the use of classical control flow. 
While the second sub-problem, a solution to which
is given in Section \ref{sect:circuits}, can be solved
efficiently regardless of the underlying circuit, the first phase is 
challenging, computationally speaking. In particular, what makes 
it non-trivial is the presence, in the source $\lambda$-calculus, 
of \emph{higher-order control flow} in the form of conditional constructs 
whose branches have higher-order types. As explained in the 
following section, this constitutes the conceptual and computational 
bottleneck of the problem. We address it as follows.
On the one hand, we design a somehow naïve, but correct, approach that  
compiles the two branches, in an \emph{asynchronous} way, onto two distinct circuits.
This approach, although sound, is potentially inefficient, as it may produce a circuit of size \emph{exponential} in that of the original program. On the other hand, 
Geometry of Interaction 
naturally suggests an alternative, \emph{synchronous}, way of handling higher-order 
control flow, which enables the production of compact circuits; 
however, this approach, although efficient, is not always applicable, as it may 
lead to deadlocks. For this reason, the strategy employed by our machine is to 
apply the synchronous method
whenever possible, and resort to the asynchronous, inefficient, method only 
when necessary, that is, in the presence of deadlocks.

After presenting the compilation procedure, we establish its soundness via a simulation result with respect to a different GoI interpretation
of quantum $\lambda$-calculi \cite{DalLago2017}, relating the two underlying machines via
the well-known  
 \emph{completely positive maps} semantics of quantum circuits
\cite{Selinger2004}.
 
Finally, we show that there is a way to isolate precisely those terms of the 
source calculus which can be compiled through the optimized machine, thus 
\emph{never} incurring in an exponential blow-up. For that, we introduce a type 
system ensuring that well-typed terms may not, by construction, give rise to 
deadlocks. Remarkably, the type system is not only sound but complete.

\paragraph{Outline}

In Section \ref{sect:informal}, we provide an informal overview over the 
challenge of compiling higher-order quantum programs with higher-order control 
flow onto quantum circuits. 
In Section \ref{sect:calculus} we introduce a linear quantum $\lambda$-calculus 
that will be the source language of our compilation procedure, together with an 
operational semantics for it. 
In Section \ref{sect:circuits} we introduce quantum circuits with conditionals, 
and we show that the latter can be efficiently eliminated from quantum circuits.
In Section \ref{sect:tkm} we present our machine, called the \circuittkm, first 
in its naïve, inefficient, form, and, then, in its more optimized form, proving 
its soundness and termination properties.
In Section \ref{sect:types} we present a type system capturing those programs 
which can be compiled efficiently without deadlocks. In the final two sections 
we discuss related and potential future work.

Due to space constraints, most proofs had to be elided, but are available in 
the Appendix.

\section{A Bird's Eye View on the Problem}
\label{sect:informal}
In this section we provide an overview on the problem addressed 
in this paper, focusing on the challenges it poses from a computational 
viewpoint.

Our goal is that of translating a term $M$ written in a linear quantum
$\lambda$-calculus~\cite{selinger2009quantum} onto a quantum circuit
$C$~\cite{Nielsen_Chuang_2010}, that is, a finite sequence of operations for
creating, modifying and measuring qubits, possibly carrying around the result of
previous measurements for the sake of conditioning the applications of
subsequent quantum operations. Observe that $M$ can manipulate higher-order
functions, while $C$ is a sequential composition of quantum gates, and has thus a genuinely
first-order nature.

At least for what concerns the
classical part of the underlying computation, a first obvious idea would be to try to compile $M$ via rewriting or abstract
machines~\cite{landin1964mechanical,Krivine2007,Accattoli2020},
adapting such approaches so that evaluating a quantum instruction
occurring in $M$ has the effect of modifying an incrementally-built circuit,
rather than explicitly performing some quantum operations. For example,
suppose that $M$ is the term
$\letv{x}{S}{Q}$, where $S$ is $\ite{N}{L}{P}$.
Suppose further that the branches $L$ and $P$ have type $\qbit\multimap\qbit$,
and that the guard $N$ is a term of Boolean type whose value is produced through
some form of quantum computation, e.g., $N$ is
$\meas(\qgate{H}~(\new\ \fc))$, while $L$ and $P$ are $\lambda x.x$ and $\lambda
x.\qgate{H}(x)$, respectively. (Here, $\qgate{H}$ stands for 
the so-called
\emph{Hadamard} gate, while \new\ and \meas\ initializes and measure a qubit,
respectively.) Now, suppose we want to compile $M$ into an 
equivalent quantum
circuit $C_M$. When evaluating $N$, such an abstract machine would (correctly) not perform any
quantum operations, but rather produce a circuit of the following form:
\begin{center}
	\begin{quantikz}[scale=.9]
		\lstick{\ket{0}} & \gate{H} & \meter{} & \setwiretype{c} &
	\end{quantikz}
\end{center}
As the Boolean produced by $N$ depends on the result of a measurement, its value
would not be known at circuit-\emph{building} time, but only at
circuit-\emph{execution} time. Therefore, \emph{both} branches
of $S=\ite{N}{L}{P}$ should be executed and compiled separately. Yet, the two
branches have \emph{functional} type, and compiling the continuation $Q$ only 
once would be non-trivial, precisely because of the presence
of a variable $x$ of function type. The situation would be even more complicated
if the type of $L$ and $P$ were of an even higher order. Since we want to keep
the abstract machine we are describing compliant with the underlying reduction
semantics, it is thus natural to proceed by compiling $Q[x\leftarrow L]$ and
$Q[x\leftarrow P]$ into two circuits $C_{Q,L}$ and $C_{Q,P}$. The overall
circuit constructed would then have the following form:
\begin{center}
	\begin{tikzpicture}
	\begin{pgfonlayer}{nodelayer}
		\node [style=none] (0) at (-1.375, -1) {$\ket{0}$};
		\node [style=none] (1) at (-1, -1) {};
		\node [style=none] (2) at (-0.5, -0.75) {};
		\node [style=none] (3) at (-0.5, -1.25) {};
		\node [style=none] (4) at (0, -1.25) {};
		\node [style=none] (5) at (0, -0.75) {};
		\node [style=none] (6) at (-0.25, -1) {$H$};
		\node [style=none] (11) at (1, -1) {
			\begin{quantikz}
				\meter{}
			\end{quantikz}
		};
		\node [style=none] (12) at (2, 0.25) {};
		\node [style=none] (13) at (2, -0.25) {};
		\node [style=none] (14) at (3.5, -0.25) {};
		\node [style=none] (15) at (3.5, 0.25) {};
		\node [style=none] (16) at (2.75, 0.25) {};
		\node [style=none] (17) at (2.75, -0.25) {};
		\node [style=none, font={\scriptsize}] (18) at (2.4, 0) {$C_{Q,L}$};
		\node [style=none, font={\scriptsize}] (19) at (3.125, 0) {$C_{Q,P}$};
		\node [style=none] (20) at (2.75, 0.2) {};
		\node [style=none] (21) at (2.75, -0.2) {};
		\node [style=none] (24) at (4, 0.175) {};
		\node [style=none] (25) at (4, -0.175) {};
		\node [style=none] (26) at (3.5, 0.175) {};
		\node [style=none] (27) at (3.5, -0.175) {};
		\node [style=none, font={\scriptsize}] (28) at (3.75, 0) {$\dots$};
		\node [style=none] (29) at (2, 0.175) {};
		\node [style=none] (30) at (2, -0.175) {};
		\node [style=none] (31) at (1.5, 0.175) {};
		\node [style=none] (32) at (1.5, -0.175) {};
		\node [style=none, font={\scriptsize}] (33) at (1.75, 0) {$\dots$};
		\node [style=none] (34) at (-0.5, -1) {};
		\node [style=none] (35) at (0.7, -1) {};
		\node [style=none] (36) at (0, -1) {};
		\node [style=none] (37) at (1.3, -0.975) {};
		\node [style=none] (38) at (1.3, -1.025) {};
		\node [style=none] (39) at (2.775, -0.975) {};
		\node [style=none] (40) at (2.725, -1.025) {};
		\node [style=none] (41) at (2.775, -0.25) {};
		\node [style=none] (42) at (2.7, -0.25) {};
	\end{pgfonlayer}
	\begin{pgfonlayer}{edgelayer}
		\draw (2.center) to (5.center);
		\draw (4.center) to (5.center);
		\draw (4.center) to (3.center);
		\draw (3.center) to (2.center);
		\draw (13.center) to (14.center);
		\draw (14.center) to (15.center);
		\draw (15.center) to (12.center);
		\draw (12.center) to (13.center);
		\draw [style=dashed] (20.center) to (21.center);
		\draw (26.center) to (24.center);
		\draw (27.center) to (25.center);
		\draw (31.center) to (29.center);
		\draw (32.center) to (30.center);
		\draw (1.center) to (34.center);
		\draw (36.center) to (35.center);
		\draw (39.center) to (41.center);
		\draw (40.center) to (42.center);
		\draw (37.center) to (39.center);
		\draw (40.center) to (38.center);
	\end{pgfonlayer}
\end{tikzpicture}
\end{center}
where the box at the top represents a conditional gate whose branches are 
$C_{Q,L}$ and $C_{Q,P}$.

By proceeding this way, however, we eventually run into a problem of
efficiency: sequentially composed programs with higher-order control flow
could possibly give rise to an exponential blow-up in the size of the produced
global circuit. Indeed, suppose we generalize the above example to a family
of terms $M_n=R_n^n$, where
\begin{align*}
R_{n+1}^m&\eqdef\letv{x_n}{S}{R_n^m};\\
R_0^m&\eqdef\lambda x.x_1(x_2(\ldots (x_m x))).
\end{align*}
It is clear that, by proceeding as above while compiling $M_n$, we would obtain a
circuit of size \emph{exponential} in $n$, consisting of $n$ levels of nested
conditionals. Still, a much simpler circuit that
captures the quantum operations performed by $M_n$ can easily be constructed in this case, namely the following:
\begin{center}
	\begin{quantikz}[scale=.9]
		\qw              &  \qw     & \qw      & \qw & \gate{H}   & \ \ldots\  &
		\gate{H} &  \\
		\lstick{\ket{0}} & \gate{H} & \meter{} & \setwiretype{c} & \vcw{-1} \\
		\lstick{\ket{0}} & \gate{H} & \meter{} & \setwiretype{c} &  & & \vcw{-2}
	\end{quantikz}
\end{center}
The compilation technique presented in this paper produces, on input
$M_n$, precisely the circuit described above. 
The key idea is to view compilation as a form of \emph{parallel data-flow 
analysis}, supported by the geometry of interaction. This 
proceeds by letting tokens follow occurrences of base types in 
the underlining type derivation. When applied to the term
$M_n$, this gives rise to a data-flow graph similar to the 
following one
\begin{center}
	\includegraphics[scale=0.65]{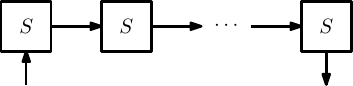}
\end{center}
There are here $n$ nodes, each corresponding to an occurrence 
of the conditional $S$. Edges here model data dependencies. 
Please observe that there are no 
circularities here. The compilation of
$M_n$ involves $n+1$ tokens. The first $n$ tokens each 
correspond to the $n$ occurrences of $N$, are local to one
of the $S$ nodes, and have a simple and short-lived 
behavior. The last token, instead, traverses the $n$ 
occurrences of $S$. The handling of conditionals is here done 
as follows: when \emph{all} the incoming tokens are available, 
the two 
branches are compiled in a recursive 
way, resulting in two circuits which then becomes the branches of a conditional 
at the level of the original circuit. We will call this way of handling 
conditionals the \emph{synchronous} rule. By tracing token paths 
\emph{hierarchically}, starting from the conditional branches placed at highest 
depth, and then moving upwards, the synchronous rule can turn such 
\emph{higher-order control flow} onto plain, \emph{circuit-level} conditionals, 
yielding in the end compact circuits like the one above, once circuit-level
conditionals have been compiled away.

Is this the end of the story? Unfortunately, the answer is negative: as anticipated, this synchronous method is not \emph{always} applicable.
Consider, for instance, the following term:
$$
R \eqdef \ite{N}{(\lambda f,g,x. f(gx))}{(\lambda f,g,x. g(fx))}.
$$
The two branches of the conditional are here of \emph{second-order} type 
$(\qbit \multimap \qbit)
\multimap (\qbit \multimap \qbit ) \multimap (\qbit \multimap \qbit )$.
Let $U$ and $W$ be quantum gates. If one were to apply the aforementioned
compilation scheme to the term $R~U~W$, one would obtain a data dependency 
graph like the following
one:
\begin{center}
	\includegraphics[scale=0.65]{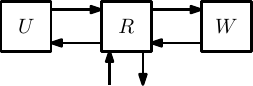}
\end{center}
A form of \emph{deadlock} arises, due to the presence of a circular dependency 
between the conditional $R$ and the two gates $U$ and $W$. Observe that 
deciding whether compiling $U$ \emph{before} $W$ or viceversa depends on 
how the conditional in $R$ evaluates. Moreover, compiling the circuits 
underlying the two branches of $R$ requires knowing, \emph{in advance}, to 
which gates and in which order the variables $f$ and $g$ will be applied, and 
thus to have already compiled both $U$ and $W$. 

It should be noted that this limitation is not as much due to 
our chosen architecture, as
it is, in fact, \emph{unavoidable}: it can be shown that no quantum circuit
containing only a single occurrence of both $U$ and $W$ is equivalent to the
term above. To be a bit more precise, there is no quantum circuit $C$ with three
holes $[-]$, $\{-\}$ and $\lfloor - \rfloor$ where each of $[-]$, $\{-\}$ and
$\lfloor - \rfloor$ appears exactly once such that for any Boolean value $x$ and
quantum gates $U$ and $W$, both $C[x]\{U\}\lfloor W \rfloor$ and
$\ite{x}{W;U}{U;W}$ produce the same quantum states. (For more detail and proof
of impossibility, see \cite{PhysRevA.88.022318}). In
Appendix~\ref{asec:swapp-unit-gates}, we give an elementary proof for quantum 
gates $U$ and $W$ acting on two qubits.

How can we overcome this \emph{impasse}? Is it possible to compile the family of terms
$M_n$ efficiently, while at the same time handling satisfactorily
deadlock phenomena such as those occurring in the term $R$? The answer
is affirmative, and consists in 
complementing the synchronous rule
mentioned above with another, \emph{asynchronous}, rule that essentially 
follows the natural (but inefficient) idea of breaking 
deadlocks through duplication.
The resulting
machine, described in Section~\ref{sect:tkm}, is on the one
hand capable of compiling any term of the source $\lambda$-calculus,
and, on the other hand, efficient in those cases in which circular dependencies like those just highlighted do not occur.

\section{The $\qlambda$-Calculus}
\label{sect:calculus}
In this section we describe the syntax and operational semantics of $\qlambda$,
the calculus which we take as source language, and which can be seen as a linear
version of Selinger and Valiron's quantum
$\lambda$-calculus,~\cite{selinger2009quantum, 10.1007/11417170_26}, almost
identical to the one considered in~\cite{lago2013wavestyletokenmachinesquantum}.
\paragraph{Syntax}

The language of \emph{terms} consists of the standard $\lambda$-calculus
constructions, together with pairs, a conditional operator, and basic 
operators on bits and qubits.

\begin{definition}  We define the set of $\qlambda$ terms, noted $\Qlambda$ by:
  {
    \begin{align*}
      \Qlambda \ni M, N, P & ::= x \mid \lambda x. M \mid M~N &\text{ (Functions)} \\
                           & \mid \starc \mid \letstar{M}{N}\ & \text{(Unit)}\\
                           & \mid \pv{M}{N} \mid \letv{\pv{x}{y}}{M}{N} & \text{(Pairs)} \\
                           & \mid \ite{M}{N}{P} &\text{(Conditionals)} \\
                           & \mid c & \text{(Operations)}
    \end{align*}} Here, $x$ ranges over an infinite set of variables, while $c$ ranges over a
    finite set $\QO$ of operators (see Definition~\ref{def:qo} below).
    We write $\letv{x}{M}{N}$ for $(\lambda x. N)~M$ and $\lambda
    x_1, \dots, x_n. M$ for $\lambda x_1. \lambda x_2. \dots \lambda
    x_n. M$.
\end{definition}

\paragraph{Type System}

The \qlambda-calculus comes equipped with a linear type 
system~\cite{girard1987linear}: every piece of
data has to be used \emph{exactly} once. This means that nothing can be
duplicated or erased, hence respecting the laws of quantum physics
regarding the non-duplication of qubits \cite{Nielsen_Chuang_2010}. The type system
contains two base types, namely the one of classical bits (denoted $\bit$)
and the one of quantum bits (denoted $\qbit$), which can be combined through
tensors or function types:

\begin{definition}
  Types are defined as follows:
  \[A, B ::= \bit \alt \qbit \alt 1 \alt A\multimap B \alt A \otimes
    B\] The meta-variable $\baseT$ ranges over base types namely
    $\bit$ and $\qbit$, while $\baseTalt$ stands for $\baseT$ or $1$.
\end{definition}

We are now ready to give a few more details about the set of possible quantum
operations. The set $\QO$ includes quantum operations acting on quantum
data. They are taken over a (universal) set of quantum gates, along with an
operation for creating and measuring qubits as well as creating and discarding
classical bits.
\begin{definition}  \label{def:qo}
  The set $\QO$ of predefined operations includes:
  \begin{itemize}
   \item Classical bit initialization $\termZero$ and $\termOne$, and classical bit discarding
   $\discard$.
    \item An operator $\new$ which initializes a new
    qubit based on the value of a classical bit. 
    \item The measurement operation $\meas$, which takes a
    qubit and returns a bit.
    \item A universal set of quantum gates such as Clifford + T set: $\set{\opn{H},
    \opn{S}, \opn{T}, \opn{CNOT}}$.
  \end{itemize}
Each operator $c \in \QO$ comes equipped with a type
$\inputType{c}\rightarrow\outputType{c}$ where $\inputType{c}$ 
and $\outputType{c}$ are possibly empty sequences of base types:
{
\begin{align*}
  \termZero,\termOne &:()\rightarrow(\bit)&
    \meas &:(\qbit)\rightarrow(\bit)
  \\
  \discard &:(\bit)\rightarrow()&
   \opn{H}, \opn{S}, \opn{T} &: (\qbit)\rightarrow(\qbit)
  \\
  \new &: (\bit)\rightarrow (\qbit)&
\opn{CNOT} &:(\qbit, \qbit)\rightarrow(\qbit,\qbit)
\end{align*}}
Given a sequence $s = (\baseT_1, \dots, \baseT_n)$ we denote by $|s|$
its length $n$.
\end{definition}

\begin{definition}
  Given an operator $c \in \QO$ such that $\inputType{c} = (\baseT_1, \dots,
  \baseT_n)$ and $\outputType{c} = (\baseT'_1, \dots, \baseT'_m)$, we define
  $\mathcal{T}(c)$ as $\baseT_1 \otimes \dots \otimes \baseT_n \multimap \baseT'_1
  \otimes \dots \otimes \baseT'_m$, where 
  $\baseT_1 \otimes \dots \otimes \baseT_p$ is $1$ when $p=0$.
\end{definition}

\begin{definition}	\emph{Typing contexts} $\Gamma,\Delta$ are sets of pairs of variables
	and their corresponding types.
    Typing rules are defined over typing judgments of the form
    $\Delta \vdash M : A$, which reads as ``the term $M$, under
    context $\Delta$, has type $A$''. The typing system is
    generated inductively from the following rules:
  
  {
  \[ \infer{x : A \vdash x : A}{} 
  \qquad
  \infer{\Delta \vdash \lambda x. M : A\multimap B}
  {\Delta, x : A\vdash M : B}
  \qquad
  \infer{\Delta, \Gamma \vdash M~N : B}
  {\Delta \vdash M : A\multimap B\qquad\Gamma\vdash N : A}
  \]
  \[
  \infer{\vdash \starc : 1}{}
  \qquad
   \infer{\Delta, \Gamma \vdash \letstar{M}{N} : A}
  {\Delta \vdash M : 1 \qquad\Gamma\vdash N : A}
  \]  
  \[
   \infer{\Delta, \Gamma\vdash \pv{M}{N} : A\otimes B}
   {\Delta\vdash M : A\qquad \Gamma\vdash N : B}
   \qquad
   \infer{\Delta, \Gamma \vdash \letv{\pv{x}{y}}{M}{N} : C}
   {\Delta \vdash M : A \otimes B\qquad\Gamma, x : A, y : B \vdash N : C}
  \]
  \[
  \infer{\Delta, \Gamma \vdash \ite{M}{N}{P} : A}
  {\Delta \vdash M : \bit \qquad \Gamma \vdash N : A \qquad \Gamma \vdash P : A}
  \qquad
  \infer{\vdash c : \mathcal{T}(c)}{}
  \]}
\end{definition}

\begin{example}
	\label{ex:coin-flip-term}
	The term $\mathit{COIN}$ implementing a perfectly fair \emph{coin-flip} can be
	written down as $\meas \, (\opn{H} \, (\new \, (\termOne \, \starc)))$ and
	correctly receives the type $\bit$ in the empty context.
\end{example}

\paragraph{Operational Semantics}

$\qlambda$ can be naturally endowed with a \emph{call-by-value} reduction
semantics. While linear functions and pairs can be treated in a completely
standard way, basic operations requires some care. First, it is convenient to
endow the calculus with two terms $\tc$ and $\fc$ of type $\bit$, to which
$\termOne \, \starc$ and $\termZero \, \starc$ evaluate, respectively. Creating
qubits is more delicate. In order to manipulate them, $\qlambda$ makes use of an
external quantum register in which the qubits the underlying program manipulates
are stored. The term can refer to the quantum registers through variables, each
of them referring to a qubit in the quantum memory. When applying a quantum
operation to these variables, the qubits corresponding to those variables is
modified accordingly. This is formalized through the notion of a \emph{quantum
closure}:

\newcommand{\qclosures}{\mathscr{Q}}
\newcommand{\distr}[1]{\mathcal{D}(#1)}
\begin{definition}[Quantum Closure~\cite{selinger2009quantum}]
  A quantum closure (or simply closure) is a triple $[Q, L, M]$ where
  $L = [x_1, \dots, x_n]$ is a list of variables, 
  $Q$ is a normalized vector of $\mathbb{C}^{2^n}$, and $M$ is a term.
  The set of quantum closures is indicated as $\qclosures$.
\end{definition}

The operational semantics of {\qlambda}\!,
following~\cite{selinger2009quantum}, is probabilistic, and defining it requires the notion of pseudo-distribution.
\begin{definition}[Pseudo-Distributions]
  Let $X$ be a set. A \emph{pseudo-distribution} over $X$ is a finite
  set of pairs, that we note $\{x_1^{\mu_1},\ldots,x_n^{\mu_n}\}$, 
  where $x_1,\ldots,x_n$ are pairwise distinct elements of $X$,
  and $\mu_1,\ldots,\mu_n$ are elements of $\mathbb{R}_{[0,1]}$
  such that $\sum_{i = 1}^{n} \mu_{i} = 1$.  
The set of
  pseudo-distributions on $X$ is denoted as $\distr{X}$. Given
  $\mu_{1}, \ldots,\mu_{n} \in \mathbb{R}_{[0,1]}$ and
  $\nu_{1},\ldots,\nu_{n} \in \distr{X}$ such that
  $\sum_{1 \leq i \leq n} \mu_{i} = 1$, we denote
  their convex combination by $\sum_{1 \leq i \leq n} \mu_{i} \nu_{i}$.
\end{definition}
We spell out the operational semantics of $\qlambda$ through
a relation $\to\subseteq\qclosures\times\distr{\qclosures}$. First, we need some auxiliary definitions.

\begin{definition}[Values, Evaluation Contexts, and Substitution]~
  The evaluation contexts $\evalCtx$ are defined as:
  \begin{align*}
  \evalCtx\; ::=\;& [\cdot] \mid \letv{\pv{x}{y}}{\evalCtx}{M} \mid
  \letstar{\evalCtx}{M} \mid \evalCtx~M \mid V~\evalCtx \mid\\
  &\pv{\evalCtx}{M} \mid \pv{V}{\evalCtx} \mid \ite{\evalCtx}{M}{N}
  \end{align*}
  As usual, $\evalCtx[M]$ is
  $\evalCtx$ where the hole $[\cdot]$ has been filled by the term $M$.
  Values are generated by the following grammar:
  \[
  V, W ::= x \mid \starc \mid c \mid \tc  \mid \fc \mid
  \pv{V}{W}\mid \lambda x. M
  \]
  We write $M[x \leftarrow N]$ for the capture-avoiding substitution of $x$ in $M$
  by the term $N$.
\end{definition}

\begin{definition}
	There are three kinds of reduction rules:
	\begin{itemize}
	\item
          \textbf{Classical Reduction Rules}. First,
          we have some (pretty standard) reduction
          rules which only influence the third component of the underlying
          quantum closures and which do not produce any probabilistic
          effect:
        {
            \begin{align*}
        [Q, L, (\lambda x. M) V] &\to \{[Q, L, M[x\leftarrow V]]^1\} \\
        [Q, L, \letv{\starc}{\starc}{M}] &\to \{[Q, L, M]^1\} \\
        [Q, L, \letv{\pv{x}{y}}{\pv{V}{W}}{M}] &\to \{[Q, L, M[x\leftarrow V, y\leftarrow W]]^1\} \\
        [Q, L, \ite{\tc }{M}{N}] &\to \{[Q, L, M]^1\} \\
        [Q, L, \ite{\fc}{M}{N}] &\to \{[Q, L, N]^1\} \\
    \end{align*}}
    \item \textbf{Operator Reduction Rules}.
    The operators, on the other hand, can very much influence the content of the 
    quantum register, possibly producing probabilistic effects:
    {
      \begin{align*}
        [Q, L, \termOne\;\starc] &\to \{[Q, L,\tc]^1\} \\    	
        [Q, L, \termZero\;\starc] &\to \{[Q, L,\fc]^1\} \\    	
        [Q, L, \discard\;b] &\to \{[Q, L,\starc]^1\} \text{ for } b \in \set{\tc, \fc} \\   	
        [Q, [x_{1},\ldots,x_{n}], U\langle x_{j_1}, \dots, x_{j_n}\rangle ]
                                 &\to \{[R, [x_{1},\ldots,x_{n}], \langle x_{j_1}, \dots, x_{j_n}\rangle ]^1\} \\
        [\alpha \ket{Q_0} + \beta\ket{Q_1}, L, \meas\ x] &\to  \{[\ket{R_0}, L \setminus \{x\}, \fc]^{\mid\alpha\mid^2}, \\
        &\qquad [\ket{R_1}, L \setminus \{x\}, \tc ]^{\mid\beta\mid^2}\} \\
        [Q, [x_1, \dots, x_n], \new\ \fc] &\to \{[Q \otimes \ket{0}, [x_1, \dots, x_n, y], y]^1\} \\
        [Q, [x_1, \dots, x_n], \new\ \tc ] &\to \{[Q \otimes \ket{1}, [x_1, \dots, x_n, y], y]^1\} 
      \end{align*}}Here $R$ is the state $Q$ where the gate $U$ has been applied to
    qubits $j_1, \dots, j_n$, and $\ket{Q_i}$ are
    normalized states of the form $\sum_j \alpha_j \ket{\psi_j^i}
    \otimes \ket{i} \otimes \ket{\phi_j^i}$ and $R_{i} = \sum_j \alpha_j \ket{\psi_j^i}
    \otimes \ket{\phi_j^i}$ for $i\in\set{0, 1}$
    respectively. Finally, in the rule for $\new$, the variable $y$ is
    fresh.
   \item \textbf{Congruence Rule}. As usual, reduction rules can be applied anywhere
    in an evaluation context.
    {
    \[
      \infer
      {[Q, L, \evalCtx[M]] \to \{[R_1, L_1,\evalCtx[N_1]]^{\mu_1},\ldots,[R_n, L_n,\evalCtx[N_n]]^{\mu_n}\}}
      {[Q, L, M] \to \{[R_, L_1,N_1]^{\mu_1}, \ldots, [R_n, L_n,N_n]^{\mu_n}\}}
    \]}
    \end{itemize}
    Whenever $C\to\{D^1\}$, we take the liberty of writing $C\to D$.
\end{definition}

\begin{example}
  \label{ex:bell}
  Consider the term $\vdash M : \qbit \otimes \qbit$ given by
  \[
    (\lambda f. \lambda x. \opn{CNOT} \, \pv{f \, x}{\new \, (\termZero \, \starc)}) \, H \, (\new \, (\termZero \, \starc)).
  \]
  This term corresponds to the circuit:
   \begin{center}
            \begin{quantikz}
              \lstick{\ket0} & \gate{H} & \ctrl{1} & \qw \\
              \lstick{\ket0} & \qw & \targ & \qw & \qw \\
            \end{quantikz}      
  \end{center}
 which computes the Bell's State (also call EPR pair). Notice that we used
higher-order functions to change the state of the second qubit by feeding
another argument instead of $\new\ \fc$. We illustrate the reduction of $M$ in
\Cref{fig:eval}, where $Q$ and $L$ are initially empty (as there are no free
variables in $M$). In \Cref{sect:tkm}, we will show how the compiling
procedure produces the circuit above.
  
  \begin{figure}
      \fbox{
            \begin{minipage}{.99\textwidth}
            {\small
              \begin{align*}
                &[-, -, (\lambda f. \lambda x. \opn{CNOT} \, \pv{f \, x}{\new \, (\termZero \, \starc)}) \, H \, (\new \, (\termZero \, \starc))] & \\
                &\to [-, -, (\lambda x. \opn{CNOT}\,\pv{H \, x}{\new \, (\termZero \, \starc)}) \, (\new \, (\termZero \, \starc))] & \beta\text{-reduction} \\
                &\to [-, -, (\lambda x. \opn{CNOT}\,\pv{H \, x}{\new \, (\termZero \, \starc)}) \, (\new \, \fc)] & \text{Evaluate } \termZero \, \starc \\
                &\to [\ket{0}, [y], (\lambda x. \opn{CNOT} \, \pv{H \, x}{\new \, (\termZero \, \starc)}) \, y] & \text{Evaluate } \new \, \fc \\
                &\to [\ket{0}, [y], \opn{CNOT} \, \pv{H \, y}{\new \, (\termZero \, \starc)}] & \beta\text{-reduction} \\
                &\to [\frac{\ket{0} + \ket{1}}{\sqrt{2}}, [y], \opn{CNOT}\,\pv{y}{\new \, (\termZero \, \starc)}] & \text{Evaluate } H\,y \\
                &\to [\frac{\ket{0} + \ket{1}}{\sqrt{2}}, [y], \opn{CNOT}\,\pv{y}{\new \, \fc}] & \text{Evaluate } \termZero\,\starc \\
                &\to [\frac{\ket{0} + \ket{1}}{\sqrt{2}} \otimes \ket{0}, [y, z], \opn{CNOT}\,\pv{y}{z}] & \text{Evaluate } \new \, \fc \\
                &\to [\frac{\ket{00} + \ket{11}}{\sqrt{2}}, [y, z], \pv{y}{z})] & \text{Evaluate } \opn{CNOT}
              \end{align*}
              }
            \end{minipage}
      }
      \caption{Example of evaluation of a closure.}
      \label{fig:eval}
  \end{figure}
\end{example}

$\qlambda$ enjoys the usual properties of typed programming languages,
namely, Subject Reduction, Progress and Termination. For more details,
we refer to either~\Cref{app:lambda} or~\cite{selinger2009quantum,
10.1007/11417170_26}.

\section{Quantum Circuits}
\label{sect:circuits}
In this section, we introduce a language for quantum circuits with
classical control flow, called \qcirc, and we describe its
interpretation in the symmetric monoidal category of \emph{completely
positive maps} \cite{Selinger2004,selinger2008fully}. Then, we show
that any circuit is equivalent to a plain one, i.e. without
if-then-else, by describing a procedure to eliminate the use of
classical control flow.

\paragraph{Quantum Circuits with Classical Control Flow}

Let $\mathbb{L}$ be an infinite set of \emph{labels}, indicated as
$l,r,s,\ldots$, to be used as names for the wires of circuits. 
Quantum circuits are defined over $\QO$, the same set of quantum gates
we used in the last section. 

\begin{definition}
We define \emph{types} and \emph{circuits}
as follows:
\begin{align*}
  \boolT & ::= \bit \mid \qbit &\text{(Base Types)} \\
  C, D   & ::= c_{\overline{r}}^{\overline{l}}
           \mid C \seq D
           \mid \ite{l}{D}{E} 
         & \text{(Circuits)} 
\end{align*}
where $c$ ranges over $\QO$, and $\overline{l}$, $\overline{r}$ are finite
sequences of labels. The expression $c_{\overline{r}}^{\overline{l}}$ indicates
the quantum operation $c$ that should be applied \emph{and} the wires it acts
on.
\end{definition}

A \emph{$\qcirc$-environment} (or, simply, an
environment) is a set of pairs of a label and its corresponding type. We require
that any label occurs at most once in an environment. 
A \emph{type judgment} for a circuit $C$ is an expression of the form $\Gamma
\rhd C \rhd \Delta$. Type judgments are inferred inductively via the rules in
Fig.~\ref{fig:qcircrules}, where $\Gamma, \Delta$ indicates the union of
environments $\Gamma$ and $\Delta$, with the proviso that they
have no label in common. 
The term $c_{\overline{r}}^{\overline{l}}$ can be typed with an additional
environment $\Theta$, which means that the corresponding circuit is the
tensorized gate $\mathrm{Id} \otimes c$. Notice that $\overline{l}$ (resp.
$\overline{r}$) is omitted when $\inputType c$ (resp. $\outputType c$) is empty.
This allows us not to consider an operator
for parallel composition in the spirit of premonoidal categories \cite{POWER_ROBINSON_1997}.
The rule $\mathtt{ite}$ also deserves some discussion:
the conditional constructor only accepts labels as its condition,
but standard conditional constructor can be recovered by
means of sequential composition. In fact, given $\Theta \rhd E \rhd l:\bit$, we
can always derive $\Gamma, \Theta \rhd E \rhd \Gamma, l:\bit$, and therefore, we
obtain $\Gamma, \Theta \rhd E \mathbin{;} \ite{l}{C}{D} \rhd \Delta$.

\begin{figure}
  \centering \fbox{
    \begin{minipage}{\textwidth}
      \begin{center}
        {\small
        \begin{math}
          \begin{array}{c}
            \infer[\QOrule]{
            \Theta,l_1:\baseT_1,\ldots,l_n:\baseT_n
            \rhd c^{\overline{l}}_{\overline{r}} \rhd
            \Theta, r_1:\baseT'_1,\ldots,r_m:\baseT'_m
            }{
            c \in \QO 
            \qquad
            \begin{array}{cc}
              \inputType{c} = (\baseT_1, \ldots, \baseT_n)
              &
                 \overline{l} = (l_{1},\ldots, l_{n})
              \\
              \outputType{c} = (\baseT'_1, \ldots, \baseT'_m)
              &
                 \overline{r} = (r_{1},\ldots, r_{m})
            \end{array}
            }
            \\[3.5mm]
            \infer[;]{
            \Gamma \rhd C ; D \rhd \Delta
            }{
            \Gamma \rhd C \rhd \Psi \qquad
            \Psi \rhd D \rhd \Delta}
            \qquad
            \infer[\mathtt{ite}]{
            \Gamma, l : \bit \rhd
            \ite{l}{C}{D} \rhd \Delta
            }{
            \Gamma \rhd C \rhd \Delta
            \qquad \Gamma \rhd D \rhd \Delta
            }
          \end{array}
        \end{math}
        }
      \end{center}
      \end{minipage}
    }
  \caption{Circuit Typing Rules.}
  \label{fig:qcircrules}
\end{figure}

\subsection{Semantics of Circuits}

We interpret circuits in terms of
completely positive maps, as
in~\cite{selinger2008fully}. The category
\cat{CPM} has for objects finite tuples of
positive integers $\sigma = (n_1, \dots, n_k)$ and
as morphisms completely positive maps
$V_\sigma \to V_{\sigma'}$ where
$V_{(n_1, \dots, n_k)} = \mathbb{C}^{n_1 \times
  n_1} \times \dots \times \mathbb{C}^{n_k\times
  n_k}$.
\cat{CPM} is symmetric monoidal closed (in fact,
dagger compact closed,
cf.~\cite{selinger2009quantum}), with tensor
product
$(n_1, \dots, n_k) \otimes (m_1, \dots, m_k) =
(n_1m_1,n_{2}m_{2}, \ldots, n_km_k)$ and a unit
$\one = (1)$.

We interpret each type of \qcirc as an object of \cat{CPM} via
$\interp{\bit} = (1, 1)$, $\interp{\qbit} = (2)$ and $\interp{A
\otimes B} = \interp{A} \otimes \interp{B}$. For the interpretation of
environments, we fix a linear order $\leq$ on $\mathbb{L}$. The
interpretation of an environment $\Gamma$ is given by
$\interp{\emptyset} = \one$ and
$\interp{\Gamma\cup\{l:\boolT\}}=\interp{\Gamma}\otimes
\interp{\boolT}$, where $l$ is the greatest label in $\Gamma \cup
\{l:\baseT\}$. For each element of $\QO$ we associate a
\cat{CPM}-morphisms defined as follows:
\begin{itemize}
\item for each quantum gate $U$ acting on
  $n$-qubits, $\pinterp{U}:\pinterp{\qbit^{n}}\to \pinterp{\qbit^{n}}$ is given by 
  \begin{math}
    \pinterp{U}(A) = U A U^{\dagger};
  \end{math}
\item 
  $\pinterp{\termOne},\pinterp{\termZero} \colon (1) \to \pinterp{\bit}$ are
given by $\pinterp{\termOne}(x)= (x,0)$ and $\pinterp{\termZero}(x)=
(0,x)$;
\item
  $\pinterp{\new} \colon \pinterp{\bit} \to
  \pinterp{\qbit}$ and
  $\pinterp{\meas} \colon \pinterp{\qbit} \to
  \pinterp{\bit}$ are given by 
$\pinterp{\new}(x,0) = x\ket{1}\bra{1}$,
$\pinterp{\new}(0,x) = x\ket{0}\bra{0}$, and 
      \begin{align*}
      \pinterp{\meas}
      \begin{pmatrix}
        a & b \\
        c & d \\
      \end{pmatrix}
      = (a, d);
  \end{align*}
\item $\pinterp{\discard} \colon \pinterp{\bit} \to (1)$ is given by $
  \pinterp{\discard}(a, b) = a + b$. 

\end{itemize}

Finally, for any typing judgment $\Gamma \rhd C \rhd \Delta$
we define a completely positive map $\interp{C} : \interp{\Gamma} \to \interp{\Delta}$ by induction on
$C$, as illustrated in Fig.~\ref{fig:circden}. For simplicity, in the definition of the interpretation of
$c^{\overline{l}}_{\overline{r}}$, we assume that labels in $\overline{l} = (l_{1},\ldots,l_{n})$
are larger than any labels in $\Gamma$ and $l_{1} < \cdots < l_{n}$. We also assume similar conditions on
$\overline{r}$ and the label $l$ in the definition of the interpretation of if-then-else. For general cases,
we need to compose $\interp{C}$ with appropriate symmetries indicated by labels.

\begin{figure}
  \centering
  \fbox{
    \begin{minipage}{\textwidth}
      {\small
      \begin{align*}
        \interp{\Gamma \rhd C ; D \rhd \Delta}
        &= \interp{\Psi \rhd D \rhd \Delta}
          \circ \interp{\Gamma \rhd C \rhd \Psi}
        \\
        \interp{\Theta, l_{1}:\baseT_{1},\ldots,l_{n}:\baseT_{n}
        \rhd c^{\overline{l}}_{\overline{l'}}
        \rhd \Theta, l'_{1}:\baseT'_{1},\ldots,l'_{m}:\baseT'_{m}}
        &= \interp{\Theta} \otimes \pinterp{c}
        \\
        \interp{\Gamma,l:\bit \rhd \ite{l}{C}{D}
        \rhd \Delta}
          (\vec A,\vec B)&=
        \interp{\Gamma \rhd C \rhd \Delta}
        (\vec A)
        +
        \interp{\Gamma \rhd D \rhd \Delta}
        (\vec B)
      \end{align*}
      }
    \end{minipage}
  }
  \caption{Interpretation of Circuits.}
  \label{fig:circden}
\end{figure}

\subsection{Eliminating Classical Control Flow}

We now show how to transform a \qcirc-circuit $C$ into a circuit without
conditionals, which we call \emph{plain circuit}, that computes the same completely positive map. 

\begin{theorem}
  \label{thm:if-then-else-removal}
  For any circuit $\Gamma\rhd C\rhd \Delta$ of size $n$, there exists a plain
  circuit $\Gamma\rhd D\rhd \Delta$ of size $\mathcal O(n)$ such that
  $\interp{C}=\interp{D}$.
\end{theorem}
\begin{proof}[Proof Sketch]
  For environments
  $\Gamma = (l_{1}:\baseT_{1},\ldots,l_{n}:\baseT_{n})$ and
  $\Delta = (r_{1}:\baseT_{1},\ldots,r_{n}:\baseT_{n})$, and for a
  label $l$, we can construct a plain circuit
  $\Gamma,\Delta, l:\bit \rhd \tau_{\Gamma,\Delta,l} \rhd
  \Gamma,\Delta,l:\bit$ that
  conditionally swaps environments:
  $\interp{\termOne_{l};\tau_{\Gamma,\Delta,l}}$ is the swapping on
  $X \otimes X$, and $\interp{\termZero_{l};\tau_{\Gamma,\Delta,l}}$
  is the identity on $X \otimes X$ where
  $X = \interp{\Gamma} = \interp{\Delta}$. We then replace every
  $\ite{l}{D}{E}$ in a given circuit $C$ with a circuit of the following
  form:
  \begin{equation*}
    \begin{tikzpicture}[baseline=(current bounding box.south)]
  \node[draw, rectangle, minimum width=0.5cm, minimum height=0.5cm] (0) {$0$};
  \node[draw, rectangle, minimum width=0.5cm, minimum height=0.5cm, right=0.5cm of 0] (new) {$\new$};
  \node[draw, rectangle, minimum width=0.5cm, minimum height=1.7cm, right=0.5cm of new, yshift=-0.6cm] (tau) {$\tau$};
  \node[draw, rectangle, minimum width=0.5cm, minimum height=0.5cm, right= 0.5cm of tau, yshift=0.6cm] (D) {$D'$};
  \node[draw, rectangle, minimum width=0.5cm, minimum height=0.5cm, right= 0.5cm of tau] (E) {$E$};
  \node[draw, rectangle, minimum width=0.5cm, minimum height=1.7cm, right=0.5cm of E] (tau2) {$\tau$};
  \node[draw, rectangle, minimum width=0.5cm, minimum height=0.5cm, right=0.5cm of tau2, yshift=0.6cm] (meas) {$\meas$};

  \draw (new.east) -- node[above]{$\scriptstyle \Xi$} (new.east-|tau.west);
  \draw (tau.west) -- (0.west |- tau.west) node[left] {$\Gamma$};
  \draw (0.east) -- (new.west);
  \draw ([yshift=-0.6cm]tau.west) -- node[above] {$\scriptstyle \bit$}
  ([yshift=-0.6cm]tau.west -| 0.west);
  \draw ([xshift=0.5cm]meas.east) -- ++(0,0.2);
  \draw ([xshift=0.5cm]meas.east) -- ++(0,-0.2);
  \draw ([xshift=0.55cm]meas.east) -- ++(0,0.15);
  \draw ([xshift=0.55cm]meas.east) -- ++(0,-0.15);
  \draw ([xshift=0.6cm]meas.east) -- ++(0,0.1);
  \draw ([xshift=0.6cm]meas.east) -- ++(0,-0.1);
  \draw (meas.east) -- ([xshift=0.5cm]meas.east);
  \draw (tau2.east) -- (tau2.east-|meas.east) -- ++(0.5,0)
  node[right] {$\Xi$};
  \draw ([xshift=0.5cm,yshift=-1.2cm]meas.east) -- ++(0,0.2);
  \draw ([xshift=0.5cm,yshift=-1.2cm]meas.east) -- ++(0,-0.2);
  \draw ([xshift=0.55cm,yshift=-1.2cm]meas.east) -- ++(0,0.15);
  \draw ([xshift=0.55cm,yshift=-1.2cm]meas.east) -- ++(0,-0.15);
  \draw ([xshift=0.6cm,yshift=-1.2cm]meas.east) -- ++(0,0.1);
  \draw ([xshift=0.6cm,yshift=-1.2cm]meas.east) -- ++(0,-0.1);
  \draw ([yshift=0.6cm]tau.east) -- (D.west);
  \draw (tau.east) -- (E.west);
  \draw ([yshift=0.6cm]tau2.west) -- (D.east);
  \draw (tau2.west) -- (E.east);
  \draw ([yshift=-0.6cm]tau.east) --
  node[below] {$\scriptstyle \bit$}
  ([yshift=-0.6cm]tau2.west);
  \draw ([yshift=0.6cm]tau2.east) -- (meas.west);
  \draw ([yshift=-0.6cm]tau2.east) --
  ([yshift=-0.6cm]tau2.east-|meas.east) --
  ++(0.5,0);
  ;

\end{tikzpicture}
  \end{equation*}
  where each box indicates a circuit, and
  computation proceeds from left to right:
  wires on the left (right) hand side of a box receive inputs (return
  outputs). The gate $\myground$ indicates a discarded bit. 
  Observe that, if we
  feed the bit $1$ to $l:\bit$, then the circuit above will compute
  $D$ and ``kill'' $E$, that is, feed it with $\ket 0$ and measure and
  discard its result; conversely, if we feed the bit $0$ to $l:\bit$,
  then the circuit above will similarly compute $E$ and ``kill'' $D$.
  In this way, the behavior of the above circuit is equivalent to
  that of $\ite{l}{D}{E}$.
  Because the size of $\tau_{\Gamma,\Delta,l}$ is proportional to the
  size of $\Gamma$ and $\Delta$, we see that the size of the result of
  if-then-else elimination is proportional to the size of the original circuit $C$.
\end{proof}

\section{The Quantum Circuit Token Machine}
\label{sect:tkm}

In this section we introduce the quantum circuit machine \circuittkm \
(Quantum Circuit Interaction Abstract Machine), a machine capable of
turning any quantum $\lambda$-term into a circuit in \qcirc. We
proceed in two steps: we first introduce a simpler, but less efficient, machine \circuittkmz\ which is proved sound; then we introduce a more efficient machine \circuittkm{}, along the lines of
Section~\ref{sect:informal}, to which we extend the soundness result.

\subsection{Extended Circuits}

Before embarking in the detailed description of the \circuittkmz, we need
to introduce a generalization of the concept of a circuit, whose purpose is to conveniently
represent a tree of ordinary circuits. These are called \emph{extended
circuit}, and are defined as follows:
$$
E,F    ::= C \mid  C \xmapsto l (E,F) 
$$
where $C$ ranges over ordinary circuits and $l$ over labels.
For any extended circuit $E$, let us fix an indexing $\alpha_1,\dots, \alpha_n$ of its $\meas$-gates. We define the following notions:
\begin{itemize}
\item an \emph{address for $E$}, noted $\ADD,\ADDH$, is a partial function from labels to $\{0,1\}$ that describes a unique path from the root of $E$ to one of its leaves (a more precise definition is contained in the Appendix);
\item a \emph{super-address for $E$}, noted $s,t$, is a partial function from labels and indexes to $\{0,1\}$ that describes a unique path from the root of $E$ to one of its leaves, together with a Boolean choice for each $\meas$-gate occurring along this path. Any super-address $s$ induces an address $\ADD_s$ by restricting it to labels.

\end{itemize}
We indicate as $@E$ the set of addresses of $E$, and as $@^{\mathsf s}E$ the set of its super-addresses.

Since an extended circuit $E$ represents a branching of
possible executions, when we add new gates to it, we need to know \emph{where} to add such
gates.
\begin{definition}[Substitution ``at address $\GG$'']
Given an extended circuit $E$ and an address $\GG\in @E$, we denote by $E@\GG$ the circuit $C$ of $E$ ``at address $\GG$''.
It is defined inductively on $E$ : 
\begin{itemize}
  \item $C@\emptyset = C$.
  \item $(C\xmapsto l (E, F))@\GG = \left\{ \begin{array}{l}
    E@(\GG\backslash \set{l \to 1}) \hfill \text{ when } (l \to 1) \in \GG \\
    F@(\GG\backslash \set{l \to 0}) \hfill \text{ when } (l \to 0) \in \GG \\
    \textbf{undefined otherwise}  \\
  \end{array}\right.$
\end{itemize}
We write $E@\GG[C]$ for the substitution of the circuit $E@\GG$ by $C$ in $E$.
\end{definition}

 The typing rules for extended circuits use
\emph{extended typing environments}, defined by
\begin{equation*}
  \mathcal J, \mathcal K ::= \leaf{\Gamma}
  \alt \node{\mathcal J}{\mathcal K}.
\end{equation*}
An extended typing environment is, intuitively, a binary tree whose leaves are standard circuit environments. 
$\mathcal J$ is said \emph{uniform} when all its leaves are of the form $\leaf{\Gamma}$, for some fixed $\Gamma$. 
A \emph{type judgment} for an extended circuit $E$ is an
expression of the form $\Gamma \rhd E \rhd \mathcal{J}$, and type
judgments are inferred inductively via the rules in
Fig.~\ref{fig:extended-circuit-typing}. The typing rule of $C
\xmapsto{l} (E,F)$ means that extended typing environments capture
branching information of extended circuits. Notice that the typing of
extended circuits is non-uniform, for instance we can build the
circuit $\Gamma \rhd C \xmapsto{l} (\opn{H}, \meas) \rhd \node{\leaf{\qbit}}{\leaf{\bit}}$.
A well-typed circuit $\Gamma\rhd E\rhd \mathcal J$ is  \emph{uniform} when $\mathcal J$ is uniform.
In this case, when all leaves of $\mathcal J$ are of the form $\leaf{\Delta}$, we may simply write 
$\Gamma\rhd E\rhd \Delta$.
\begin{figure}
  \centering \fbox{
	    \begin{minipage}{\textwidth}
	    {
		      $$
				          \infer{\Gamma \rhd C \rhd
					            \opn{leaf}(\Delta)}{\Gamma \rhd C \rhd
					            \Delta} 
					        \qquad\qquad\qquad
				          \infer{\Gamma \rhd C
					            \xmapsto{l}(F, G) \rhd
					            \node{\mathcal{J}}{\mathcal{V}}} 
					            {
					            \begin{array}{c}
					            	\Gamma\rhd C \rhd l : \bit, \Delta
					            	\\
                        \Delta \rhd F \rhd \mathcal{J}  \qquad
					            	\Delta \rhd G \rhd \mathcal{V}
					            \end{array}
					            }
			      $$
			      }
		    \end{minipage}}
  \caption{Extended Circuit Typing Rules.}
  \label{fig:extended-circuit-typing}
\end{figure}

It is always possible to recover a circuit from an extended circuit through the
map $\translate(-)$ defined by $\translate(C) = C$ and
$ \translate(C \xmapsto l (F, G))
 = C \mathbin{;} \ite l {\translate(F)} {\translate(G)}$. When $E$ is uniform, we obtain:
 
\begin{lemma}
 \label{lem:translation-EC-well-typed}
 For any uniform extended circuit $\Gamma \rhd E \rhd \mathcal{J}$, where all leaves of $\mathcal J$ are $\leaf{\Delta}$, 
 $\Gamma \rhd \translate(E) \rhd \Delta$.
\end{lemma}

We can define the semantics of a {uniform well-typed extended circuit} as the
semantics of its translation: $\interp{E} = \interp{\translate(E)}$. For an
arbitrary well-typed (non-uniform) extended circuit $\Gamma\rhd E\rhd \mathcal
J$, notice that any address $\ADD\in @^{\mathsf s}E$ induces a circuit
environment $\Gamma\rhd -\rhd \mathcal J_{\ADD}$ (where $\mathcal J_\ADD$
indicates the leaf of $\mathcal J$ at address $\ADD$). In fact, we can define
the semantics of $E$ \emph{per slice}, as $\interp{E}^{\sharp
s}:\interp{\Gamma}\to\interp{\mathcal J_{\ADD_s}}$, for any choice of a
super-address $s\in @^{\mathsf s}E$. This is given in details in the Appendix.

\subsection{The \circuittkmz{} Machine}

The \circuittkmz\ is a \emph{token machine}, in the style of Girard's geometry
of interaction~\cite{goi0,goi1,danos1999reversible,asperti1995paths}. The idea
is that, starting from a type derivation $\pi$ in \qlambda, we will let a finite
set of \emph{tokens} move through $\pi$, following the occurrences of \emph{base
types} $1$, $\bit$ or $\qbit$. Each token starts from a \emph{negative}
occurrence of some base type in the conclusion of $\pi$, corresponding to an
input, and eventually reaches a \emph{positive} occurrence, corresponding to an
output. As they travel along the derivation, the tokens perform parallel
data-flow analysis, producing along the way the circuit corresponding to the
underlying $\qlambda$-term. At every moment along the execution of the
\circuittkmz\ on $\pi$, each token lies in a \emph{position} inside $\pi$,
namely at an occurrence of some base type $\boolT$ in some type judgment
occurring in $\pi$. To be more precise, let us first introduce the notion of
\emph{position within a type $A$}. Any occurrence of some base type $\boolT$ in
a type $A$ is determined by a unique \emph{occurrence context} $\mathtt C$,
corresponding intuitively to a type containing a unique occurrence of the hole
$[-]$, so that $A=\mathtt C[\boolT]$.
More precisely, \emph{positive} and \emph{negative} occurrence contexts
are defined as follows, in a mutually recursive way:
  \begin{align*}
  \mathtt{P} &::= [-] \alt \mathtt{N} \multimap A \alt A \multimap
  \mathtt{P} \alt \mathtt{P} \otimes A \alt A \otimes \mathtt{P} \\
  \mathtt{N} &::= \mathtt{P} \multimap A \alt A \multimap
  \mathtt{N} \alt \mathtt{N} \otimes A \alt A \otimes \mathtt{N}
  \end{align*}
an occurrence context $\mathtt C$ is either a positive or a negative occurrence
context, and a  $\boolT$-\emph{position} within a type $A$ is defined as a pair 
$(\mathtt C,\baseT)$ such that $A=\mathtt C[\baseT]$.
A \emph{position within a judgment $\Gamma\vdash A$} is either a position in
$A$ or a position in some type $B$ occurring in $\Gamma$. If the position is
within $A$, then it is positive (resp.~negative) precisely when it is
positive (resp.~negative) as a position within $A$; if the position is within
$B$, then it is positive (resp.~negative) when it is negative
(resp.~positive) as a position within $B$. Finally, a \emph{position within a
type derivation $\pi$} is a position within some judgment occurring in
$\pi$, and its polarity corresponds to its polarity as a position within the
judgment.

There are certain sets of positions that will be crucial in the following, and
deserve to be given a formal status.
   
\begin{definition}[Position Sets]
	Given a derivation $\pi$, we define the following
	sets of positions:
	\begin{itemize}
	\item the set of negative positions of datatypes and 
	of the unit type $1$ in the conclusion
	of $\pi$, noted $\negative_\pi$ and $\negones_\pi$, 
	respectively;
	\item the set of positive positions of datatypes and
	of the unit type $1$ in the conclusion
	of $\pi$, noted $\positive_\pi$ and $\posones_\pi$,
	respectively;
	\item the set of positions lying at occurrences of
	the unit typing rule $\vdash * : 1$, noted
	$\ones_\pi$;
	\item the set of $\bit$-positions lying at the guards of conditionals,
	noted $\guard_\pi$.
\end{itemize}
\end{definition}
For every set of positions $\mathtt{X}$, we write $\mathtt{X}^\downarrow$ for
the subset of $\mathtt{X}$ consisting of all the positions in $\mathtt{X}$
occurring outside the scope of any conditional branches, i.e., such that the
path from the position to the conclusion of $\pi$ does not go through a then or
else branch in an instance of the conditional typing rule.

Let $\mathsf{PSN}_\pi$ indicate the set of positions in $\pi$ and $\mathsf{LAB}_\pi\subseteq \mathbb L$ a finite set of labels. 
We suppose given a bijection $\mathsf{lab} :\mathsf{PSN}_\pi \to \mathsf{LAB}_\pi$
associating each position in $\pi$ with a unique label.
For any set $\mathtt X=\{\sigma_1,\dots, \sigma_n\}$ of positions, with $\sigma_i$ of type $\mathbb B_i$, we define $\mathbf{ENV}(\mathtt X)=\{\mathsf{lab}(\sigma_1):\mathbb B_1,\dots, \mathsf{lab}(\sigma_n):\mathbb B_n\}$.

Tokens carry an \emph{address}, denoted $\GG$, which, as in the case of extended circuits, is 
a partial function from labels $\mathsf{LAB}_\pi$ to $\{0,1\}$. The address tells us
in which branch of a certain conditional construct the token should
go. Indeed, tokens might enter and exit the conditional, and have to
somehow \emph{remember} where they should go. 

We are finally ready to define what tokens are:

\begin{definition}[Token]
  \label{def:token}
  A \emph{token} in a type derivation $\pi$ is a pair $\mathfrak a=(\sigma, \ADD)$ formed by a position $\sigma$ in $\pi$ and an address $\ADD$.
  A token $\mathfrak a=(\sigma,\ADD)$ is called a \emph{unit} token if $\sigma $ is a $1$-position in $\pi$, otherwise it is called a \emph{data token}.
\end{definition}

The \circuittkmz\ is a \emph{multitoken} machine, and finite sets
of tokens (which are ranged over by $\MM$) will be essential components
of their configurations.
For any set of positions $\mathtt X$, let $\mathbf{TKN}(\mathtt X)$ be the set of tokens $(\sigma,\emptyset)$, such that $\sigma\in \mathtt X$. 
Given a set of tokens $\MM$, let $\mathbf{PSN}(\MM)$ indicate the set of positions occurring in $\MM$. 
Moreover, for any address $\ADD$, let $\MM|_\ADD\subseteq \MM$ be the set of tokens in $\MM$ at address $\ADD$, i.e.~of the form $\mathfrak a=(\sigma, \ADD)$, and $\GG(\MM)$ be the set
$\set{( \sigma, \GG) \alt ( \sigma, \GG') \in \MM}$, i.e. the set $\MM$
where all the addresses have been set to $\GG$.

It is convenient to isolate sets of tokens
$\MM$ covering every positive or negative position of a type occurrence in $\pi$: given a type occurrence $A$ and a set of tokens 
$\MM$, we say that:
  \begin{varitemize}
    \item $\MM$ \emph{positively saturates} $A$ if there is a token at position $\sigma$ in 
    $\MM$ for each positive position $\sigma$ in $A$.
    \item $\MM$ \emph{negatively saturate} $A$ if there is a token at position $\sigma$ in
    $\MM$ for each negative position $\sigma$ in $A$.
  \end{varitemize}

For any set of tokens $\MM$ there is a naturally defined (possibly non-uniform) extended type environment $\mathbf{ENV}(\MM)$ obtained by analysis of the addresses of the tokens: for any address $\ADD$ occurring in $\MM$, the circuit environment $\mathbf{ENV}(\MM)_\ADD$ is $\mathbf{ENV}(\mathbf{PSN}(\MM|_\ADD)-\mathcal A)$, where $\mathcal A=\negones_\pi\cup\posones_\pi\cup \ones_\pi\cup\guard_\pi$.
In other words, $\mathbf{ENV}(\MM)_\ADD$ contains assignments $l:\mathbb B$, where $l=\mathsf{lab}(\sigma)$, for each \emph{data} token $a=(\sigma,\ADD)$ of type $\mathbb B$ in $\MM|_\ADD$, with the sole exception of the tokens occurring in the guard of an if-then-else.

A \emph{configuration} of the \circuittkmz is given by the positions of the
tokens together with the circuit produced so far. More precisely, a
configuration     $\config$ is a tuple $(\pi,  \MM, E)$ where:
    \begin{varitemize}
        \item $\pi:\Gamma\vdash M:A$ is a type derivation of \qlambda; 
        \item $\MM$ is a set of tokens;
        \item  
        $\mathbf{ENV}(\negative_\pi)\rhd E \rhd \mathbf{ENV}(\MM)$ is an extended
        circuit.
    \end{varitemize}

  Notice that the definition above implies that the addresses appearing in the tokens $\MM$ coincide with those which are well-defined for $E$, i.e.~with those in $@E$. Moreover, each label $l=\mathsf{lab}(\sigma)$ of an input wire of $E$ is associated with some initial position of the tokens, i.e.~with some $\sigma\in\negative_\pi$; on the other hand, each label $l=\mathsf{lab}(\sigma)$ of an output wire of $E$ at address $\ADD$ is associated with a unique token $(\sigma,\ADD)\in \MM$. This is illustrated in Fig.~\ref{fig:tokenwires1}.
  
  \begin{figure}
  \fbox{
  \begin{minipage}{.99\textwidth}
  \begin{subfigure}{\textwidth}
    \dbox{
  \begin{minipage}{.97\textwidth}
   \begin{center}
  \resizebox{.325\textwidth}{!}{
\begin{tikzpicture}
\node(d) at (0,4) {Tokens traveling in the derivation $\pi$};
  
  \draw[dotted] (0,0.3) to (-1.5,3) to (1.5,3) to (0,0.3);
  \node(pi) at (0,1.5) {$\pi$};

  \node(a) at (0,0) { $x_1:\qbit,x_2,\qbit\vdash M:\qbit\otimes\qbit$};
  \node[red](tok1) at (-1.8,-0.4) {\tiny${ (\sigma_1,\emptyset)}$};
  \node[red](tok1) at (-0.4,-0.4) {\tiny${ (\sigma_2,\emptyset)}$};
  \node[red](tok) at (-1,-0.7) {\tiny initial positions $\negative_\pi$};

    \node[blue](tok1) at (-0.5,2) {\tiny$\stackrel{\bullet}{ (\sigma_3,\{l\to 0\})}$};
    \node[blue](tok1) at (0.7,2.5) {\tiny$\stackrel{\bullet}{ (\sigma_4,\{l\to 0\})}$};
    \node[blue](tok1) at (0.2,1) {\tiny$\stackrel{\bullet}{ (\sigma_5,\{l\to 1\})}$};
    \node[blue](tok1) at (-1,2.7) {\tiny$\stackrel{\bullet}{ (\sigma_6,\{l\to 1\})}$};
  \node[blue](tok) at (-2.1,1.7) {\tiny 
  \begin{tabular}{c}
 tokens in $\MM$\\
 may have addresses \\
 $\{l\to 0\}$ or $\{l\to 1\}$
 \end{tabular}
 };

 \end{tikzpicture}
 }
\hspace{1cm}
 \resizebox{.52\textwidth}{!}{
 \begin{tikzpicture}
 \node(d) at (2.5,3.8) {\small Circuit $\mathbf{ENV}(\negative_\pi)\rhd E\rhd \mathbf{ENV}(\MM)$};

  \draw (0,0) to (4.5,0) to (4.5,3) to (0,3) to (0,0);
    \node(e) at (-0.3, 3.2) {$E$};
 \draw (0,2.3) to (-1,2.3);
 \node(l1) at (-0.5,2.5) {\tiny$\mathsf{lab}({\color{red}\sigma_1})$};
  \draw (0,0.7) to (-1,0.7);
   \node(l2) at (-0.5,1) {\tiny$\mathsf{lab}({\color{red}\sigma_2})$};
   \draw[red, decoration={brace,mirror,raise=5pt},decorate]
  (-1,2.5) -- node[left=6pt] {\tiny$\mathbf{ENV}(\negative_\pi)$} (-1,0.4);

\draw[->, dotted] (0.3,1.5) to [bend left] node[above] {\tiny${\color{blue}l\to 0}$} (1.4,2.15);
\draw[->, dotted] (0.3,1.5) to [bend right] node[below] {\tiny${\color{blue}l\to 1}$} (1.4,0.65);

  \draw (2,1.6) to (3.6,1.6) to (3.6,2.9) to (2,2.9) to (2,1.6);
    \node(e1) at (2.8,2.25) {$C_1$};
       \draw (2,2.6) to (1.8,2.6);
   \draw (2,1.9) to (1.8,1.9);
       \draw (3.6,2.6) to (4.3,2.6);
        \node(l1) at (4.1,2.8) {\tiny$\mathsf{lab}({\color{blue}\sigma_3})$};
   \draw (3.6,1.9) to (4.3,1.9);
        \node(l1) at (4.1,2.1) {\tiny$\mathsf{lab}({\color{blue}\sigma_4})$};
   \draw[blue, decoration={brace,raise=5pt},decorate]
  (4.5,2.8) -- node[right=6pt] {\tiny$\mathbf{ENV}(\MM|_{l\to 0})$} (4.5,1.8);

  \draw (2,1.4) to (3.6,1.4) to (3.6,0.1) to (2,0.1) to (2,1.4);
    \node(e2) at (2.8,0.75) {$C_2$};
           \draw (2,1.1) to (1.8,1.1);
   \draw (2,0.4) to (1.8,0.4);
       \draw (3.6,1.1) to (4.3,1.1);
        \node(l1) at (4.1,1.3) {\tiny$\mathsf{lab}({\color{blue}\sigma_5})$};
   \draw (3.6,0.4) to (4.3,0.4);
        \node(l1) at (4.1,0.6) {\tiny$\mathsf{lab}({\color{blue}\sigma_6})$};
   \draw[blue, decoration={brace,raise=5pt},decorate]
  (4.5,1.3) -- node[right=6pt] {\tiny$\mathbf{ENV}(\MM|_{l\to 1})$} (4.5,0.3);

  \end{tikzpicture}
}
\end{center}
 
      \caption{Configuration $(\pi,\MM,E)$: in $\pi$ there are two initial positions $\negative_\pi$ (in red); there are four tokens in $\MM$ (in blue), two at address $l\to 0$, two at address $l\to 1$.
      In the extended circuit $\mathbf{ENV}(\negative_\pi)\rhd E\rhd \mathbf{ENV}(\MM)$ the input wires are labeled from $\sigma_1,\sigma_2\in\negative_\pi$; the circuit has two addresses $l\to 0,l\to1$ leading to circuits $C_1,C_2$, where the output wires of $C_b$ are labeled from the tokens in $\MM|_{l\to b}$.
      }
    \label{fig:tokenwires1}
  \end{minipage}
  }
  \end{subfigure}
  
  \medskip
  
  \begin{subfigure}{\textwidth}
    \dbox{
  \begin{minipage}{.97\textwidth}
   \begin{center}
 \resizebox{.325\textwidth}{!}{
\begin{tikzpicture}
\node(d) at (0,4) {Tokens traveling in the derivation $\pi$};
  
  \draw[dotted] (0,0.3) to (-1.5,3) to (1.5,3) to (0,0.3);
  \node(pi) at (0,2.2) {$\pi$};

  \node(a) at (0,0) { $x_1:\qbit,x_2,\qbit\vdash M:\qbit\otimes\qbit$};
  \node[red](tok1) at (-1.8,-0.4) {\tiny${ (\sigma_1,\emptyset)}$};
  \node[red](tok1) at (-0.4,-0.4) {\tiny${ (\sigma_2,\emptyset)}$};
  \node[red](tok) at (-1,-0.7) {\tiny initial positions $\negative_\pi$};

  \node[blue](tok1) at (2.4,1) {\tiny${ (\tau_1,\{l\to0\})}$};
  \node[blue](tok1) at (2.4,0.7) {\tiny${ (\tau_1,\{l\to1\})}$};
  \node[blue](tok1) at (2.4,0.4) {\tiny$\downarrow$};

  \node[blue](tok1) at (0.9,1) {\tiny${ (\tau_2,\{l\to 0\}}$};
    \node[blue](tok1) at (0.9,0.7) {\tiny${ (\tau_2,\{l\to1\})}$};
      \node[blue](tok1) at (1.2,0.4) {\tiny$\downarrow$};

  \node[blue](tok) at (1.5,1.4) {\tiny 
  \begin{tabular}{c}tokens in the\\
  final positions $\positive_\pi$
  \end{tabular}
  };

 \end{tikzpicture}
 }
\hspace{1cm}
 \resizebox{.52\textwidth}{!}{
 \begin{tikzpicture}
 \node(d) at (2.5,3.95) {\small Circuit $\mathbf{ENV}(\negative_\pi)\rhd E\rhd \mathbf{ENV}(\positive_\pi)$};

  \draw (0,0) to (4.5,0) to (4.5,3) to (0,3) to (0,0);
    \node(e) at (-0.3, 3.2) {$E$};
 \draw (0,2.3) to (-1,2.3);
 \node(l1) at (-0.5,2.5) {\tiny$\mathsf{lab}({\color{red}\sigma_1})$};
  \draw (0,0.7) to (-1,0.7);
   \node(l2) at (-0.5,1) {\tiny$\mathsf{lab}({\color{red}\sigma_2})$};
   \draw[red, decoration={brace,mirror,raise=5pt},decorate]
  (-1,2.5) -- node[left=6pt] {\tiny$\mathbf{ENV}(\negative_\pi)$} (-1,0.4);

\draw[->, dotted] (0.3,1.5) to [bend left] node[above] {\tiny${\color{blue}l\to 0}$} (1.4,2.15);
\draw[->, dotted] (0.3,1.5) to [bend right] node[below] {\tiny${\color{blue}l\to 1}$} (1.4,0.65);

  \draw (2,1.6) to (3.6,1.6) to (3.6,2.9) to (2,2.9) to (2,1.6);
    \node(e1) at (2.8,2.25) {$C_1$};
       \draw (2,2.6) to (1.8,2.6);
   \draw (2,1.9) to (1.8,1.9);
       \draw (3.6,2.6) to (4.3,2.6);
        \node(l1) at (4.1,2.8) {\tiny$\mathsf{lab}({\color{blue}\tau_1})$};
   \draw (3.6,1.9) to (4.3,1.9);
        \node(l1) at (4.1,2.1) {\tiny$\mathsf{lab}({\color{blue}\tau_2})$};
   \draw[blue, decoration={brace,raise=5pt},decorate]
  (4.5,2.8) -- node[right=6pt] {\tiny$\mathbf{ENV}(\positive_\pi)$} (4.5,1.8);

  \draw (2,1.4) to (3.6,1.4) to (3.6,0.1) to (2,0.1) to (2,1.4);
    \node(e2) at (2.8,0.75) {$C_2$};
           \draw (2,1.1) to (1.8,1.1);
   \draw (2,0.4) to (1.8,0.4);
       \draw (3.6,1.1) to (4.3,1.1);
        \node(l1) at (4.1,1.3) {\tiny$\mathsf{lab}({\color{blue}\tau_1})$};
   \draw (3.6,0.4) to (4.3,0.4);
        \node(l1) at (4.1,0.6) {\tiny$\mathsf{lab}({\color{blue}\tau_2})$};
   \draw[blue, decoration={brace,raise=5pt},decorate]
  (4.5,1.3) -- node[right=6pt] {\tiny$\mathbf{ENV}(\positive_\pi)$} (4.5,0.3);

  \end{tikzpicture}
}
\end{center}
\caption{Final configuration $(\pi,\MM,E)$: all tokens have reached a position in $\positive_\pi$; 
the extended circuit $\mathbf{ENV}(\negative_\pi)\rhd E\rhd \mathbf{ENV}(\positive_\pi)$ is uniform, the output wires in all its leaves are labeled from the tokens at the positions $\positive_\pi$.
}
    \label{fig:tokenwires2}
        \end{minipage}
  }  
  \end{subfigure}  
      \end{minipage}
    }
    \caption{Illustration of the configurations of the \circuittkmz.}
    \label{fig:tokenwires}
   \end{figure}
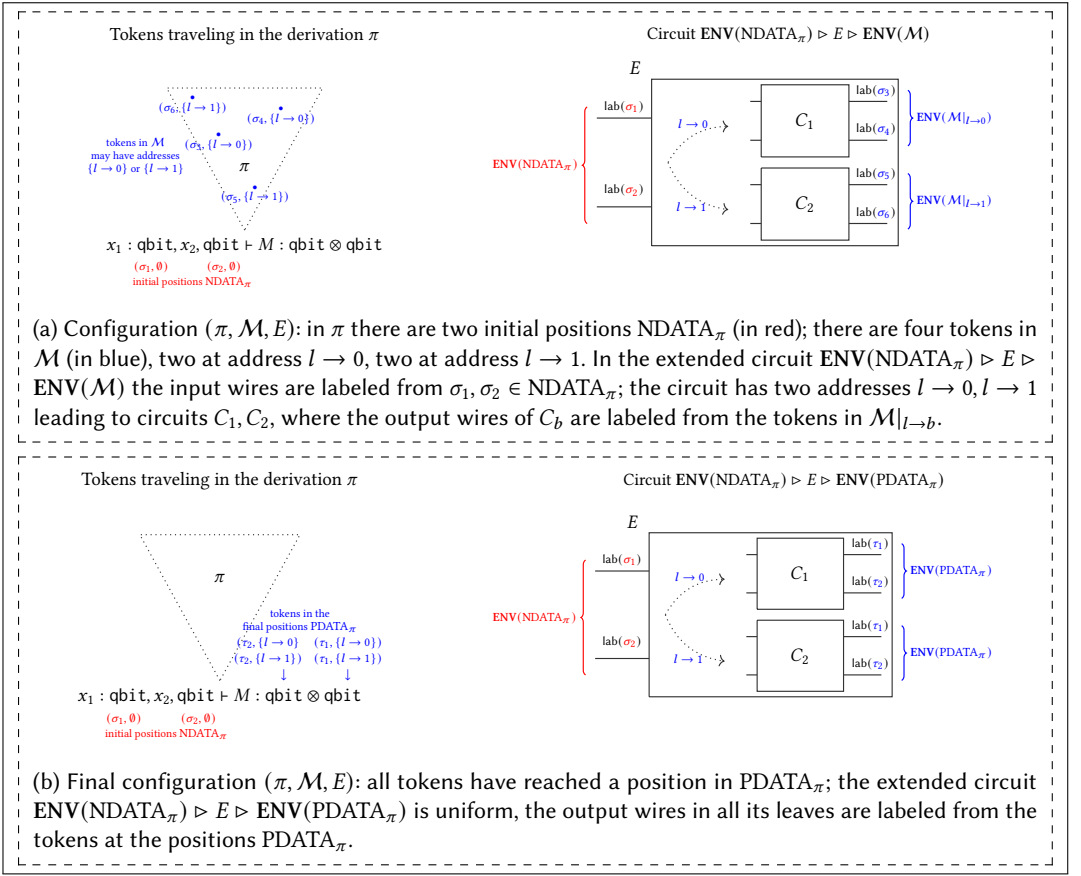

Let us define the \emph{initial} and \emph{final} configurations:
\begin{itemize}
\item 
The initial configuration for $\pi$ has the form 
$(\pi, \mathbf{TKN}(\negative_\pi\cup\negones_\pi\cup\ones_\pi^\downarrow), \mathrm{Id} )$ 
where $\mathrm{Id}$ is a circuit composed of an identity gate on each wire.
\item
A final configuration for $\pi$ has, instead the form
$(\pi,\MM,E)$ where $\mathbf{PSN}(\MM)=\positive_\pi\cup\posones_\pi\cup\guard_\pi$.
\end{itemize}

Observe that in a final configuration all tokens have reached some positive position in the conclusion of $\pi$ (but recall that tokens may still have different addresses). The circuit $E$ is then uniform with type environment $\mathbf{ENV}(\negative_\pi)\rhd E\rhd \mathbf{ENV}(\positive_\pi)$. This is illustrated in Fig.~\ref{fig:tokenwires2}.

\begin{figure*}

  \begin{subfigure}{\textwidth}
  \begin{minipage}{\textwidth}
  \footnotesize
  \begin{minipage}{.48\textwidth}

  \dbox{
  \begin{minipage}{.99\textwidth}
  \centering
          $
          \infer[\otimes]{\Gamma_2, \Delta_2 \vdash \pv{t_1}{t_2} :
          A_2 \otimes B_2}{\Gamma_1 \vdash t_1 : A_1 \qquad \Delta_1
          \vdash t_2 : B_1}
          $
          \resizebox{0.9\textwidth}{!}{
          \begin{minipage}{\textwidth}
          \begin{align*}
           (\set{\mathfrak a_1:  \Gamma_1^{+}} \cup \MM, E)& \totkm (\set{\mathfrak a_2 :  \Gamma_2^{+}} \cup \MM, E') \\
           (\set{\mathfrak a_1: \Delta_1^{+}} \cup \MM, E) &\totkm (\set{\mathfrak a_2: \Delta_2^{+}} \cup \MM, E') \\
           (\set{\mathfrak a_1 : A_1^{+}} \cup \MM, E)& \totkm (\set{\mathfrak a_2 : A_2^{+} \otimes B_2} \cup \MM, E') \\
           (\set{\mathfrak a_1 : B_1^{+}} \cup \MM, E)& \totkm (\set{\mathfrak a_2 : A_2 \otimes B_2^{+}} \cup \MM, E') 
         \end{align*}
         \end{minipage}
         }
  \end{minipage}
  }

  \vskip1mm

  \dbox{
    \begin{minipage}{.99\textwidth}
      \centering
      $   
       \infer{\Delta_2, \Gamma_2 \vdash \letstar{M}{N} : A_2}
      {\Delta_1 \vdash M : 1 \qquad\Gamma_1\vdash N : A_1}
      $
    \resizebox{0.9\textwidth}{!}{
          \begin{minipage}{\textwidth}
          \begin{align*}
           (\set{\mathfrak a_1: \Gamma_1^{+}} \cup \MM, E)& \totkm (\set{\mathfrak a_2: \Gamma_2^{+}} \cup \MM, E') \\
           (\set{\mathfrak a_1: \Delta_1^{+}} \cup \MM, E) &\totkm (\set{\mathfrak a_2: \Delta_2^{+}} \cup \MM, E') \\
           (\set{\mathfrak a_1: A_1^{+}} \cup \MM, E)& \totkm (\set{\mathfrak a_2 : A_2^{+}} \cup \MM, E') 
         \end{align*}
        \end{minipage}
    }
    \end{minipage}
  }

  \vskip1mm

  \dbox{
    \begin{minipage}{\textwidth}
      \vskip5mm
    \centering
    {$\infer{\Gamma_2, \Delta_2\vdash \letv{(x, y)}{M}{N} : C_2}{
            \Gamma_1\vdash M : A_1\otimes B_1 \qquad \Delta_2, x : A_2, y :
            B_2\vdash N : C_1}
            $}
            \resizebox{0.9\textwidth}{!}{
            \begin{minipage}{\textwidth}
           \begin{align*}
           (\set{\mathfrak a_1: \Gamma_1^{+}} \cup \MM, E)& \totkm (\set{\mathfrak a_2: \in \Gamma_2^{+}} \cup \MM, E') \\
           (\set{\mathfrak a_1 :  \Delta_1^{+}} \cup \MM, E) &\totkm (\set{\mathfrak a_2 : \Delta_2^{+}}\cup \MM, E') \\
           (\set{\mathfrak a_1 : C_1^{+}}\cup\MM, E) &\totkm (\set{\mathfrak a_2 : C_2^{+}}\cup\MM, E') \\
           (\set{\mathfrak a_1: A_1^{+} \otimes B_1}\cup\MM, E) &\totkm (\set{\mathfrak a_2 : A_2^{+}}\cup\MM, E') \\
           (\set{\mathfrak a_1 : A_1 \otimes B_1^{+}}\cup\MM, E) &\totkm (\set{\mathfrak a_2 : B_2^{+}}\cup\MM, E')
           \end{align*}
           \end{minipage}
           }
           \end{minipage}
    }

  \end{minipage}
  \ \ \ \ \ \ \ \ \ \ \
  \begin{minipage}{.48\textwidth}
  
  \dbox{
  \begin{minipage}{.99\textwidth}
  \vskip2mm
  \centering       
        $   
           \infer{\Gamma_2 \vdash \lambda x. M : A_2 \multimap B_2}
           {\Gamma_1, x : A_1 \vdash M : B_1}
        $
          \resizebox{0.89\textwidth}{!}{
          \begin{minipage}{\textwidth}
          \begin{align*}
           (\set{\mathfrak a_1: \Gamma_1^{+}}\cup\MM, E) & \totkm (\set{\mathfrak a_2 :  \Gamma_2^{+}} \cup \MM, E') \\
           (\set{\mathfrak a_1 : A_1^{+} \multimap B_2} \cup \MM, E)& \totkm (\set{\mathfrak a_2 : A_2^{+}} \cup \MM, E') \\
           (\set{\mathfrak a_1 : B_1^{+}} \cup \MM, E) &\totkm (\set{\mathfrak a_2 : A_2 \multimap B_2^{+}} \cup \MM, E') 
          \end{align*}
          \end{minipage}
  }
  \end{minipage}
  }

  \vskip 1mm

  \dbox{
  \begin{minipage}{\textwidth}
  \centering
        $
           \infer{\Gamma_2, \Delta_2\vdash MN : B_2}{\Gamma_1 \vdash M : A_1 \multimap B_1 \qquad \Delta_1 \vdash N : A_2}
           $
          \resizebox{0.9\textwidth}{!}{
          \begin{minipage}{\textwidth}
          \begin{align*}
           (\set{\mathfrak a_1 : \Gamma_1^{+}} \cup \MM, E)& \totkm (\set{\mathfrak a_2 :\Gamma_2^{+} }\cup\MM, E')\\
           (\set{\mathfrak a_1 : \Delta_1^{+}} \cup \MM, E)& \totkm (\set{\mathfrak a_2 :\Delta_2^{+}}\cup\MM, E')\\
           (\set{\mathfrak a_1 : A_2^{+}} \cup \MM, E) &\totkm (\set{\mathfrak a : A_1^{-} \multimap B_1} \cup \MM, E') \\
           (\set{\mathfrak a_1 : A_1 \multimap B_1^{+}} \cup \MM, E)& \totkm (\set{\mathfrak a_2 : B_2^{+}} \cup \MM, E')
          \end{align*}
          \end{minipage}
          }
        \end{minipage}
  }

  \vskip 1mm

  \dbox{
    \begin{minipage}{\textwidth}
    \centering
         $
         \infer{x : A_1 \vdash x : A_2}{}
         $
            \resizebox{0.8\textwidth}{!}{
            \begin{minipage}{\textwidth}
    \begin{align*}
              (\set{\mathfrak a_1 : A_1^{+}} \cup \MM, E) &\totkm (\set{\mathfrak a_2 : A_2^{-}} \cup \MM, E')
            \end{align*}
              \vskip4mm
            \end{minipage}
    }
    \end{minipage}
    }

  \end{minipage}
  
  \end{minipage}
    \caption{Structural rules. In all rules, the tokens $\mathfrak a_1,\mathfrak a_2$ are at some fixed address $\ADD$ and $E'=E[l_{\mathfrak a_1}\mapsto l_{\mathfrak a_2}]_\ADD$.}\label{fig:global-tkm-structural-rules}
  \end{subfigure}

   \medskip

  \begin{subfigure}{\textwidth}
   \dbox{
     \begin{minipage}{\textwidth}
    \begin{center}
  \footnotesize
            $ \infer{\vdash U : \baseT_1 \otimes \dots \otimes
            \baseT_n \multimap \baseT_{n+1} \otimes \dots \otimes \baseT_{2n}}{}$
            { 
            \resizebox{.99\textwidth}{!}{
              \begin{minipage}{\textwidth}
              \begin{align*}
          \begin{array}{c}
            (\set{\mathfrak a_1 : \baseT_1, \dots, \mathfrak a_n : \baseT_n} \cup \MM, 
             \Gamma \rhd E \rhd \mathcal{V})
          \end{array}
            \totkm
          \begin{array}{c}
            (\set{\mathfrak a_{n+1} : \baseT_{n+1}, \dots, \mathfrak a_{2n} : \baseT_{2n}} \cup \MM), 
            \Gamma \rhd E@\GG\left [E@\GG ; U^{l_{\mathfrak a_1},\dots, l_{\mathfrak a_n}}_{l_{\mathfrak a_{n+1}},\dots, l_{\mathfrak a_{2n}}}\right ] \rhd \mathcal{V'}
          \end{array}
          \end{align*}
        \end{minipage}
          }
          }
     \end{center}
   \end{minipage}} \caption{Circuit rule. We assume that all tokens $\mathfrak a_i$ are at address $\ADD$.}
   \label{fig:global-tkm-unitary-rule}
   \end{subfigure}

   \medskip

   \begin{subfigure}{\textwidth}
     \dbox{
     \begin{minipage}{.99\textwidth}
    \begin{center}
  \footnotesize
  $
  \infer{\Gamma_2, \Delta_3 \vdash \ite{M}{N}{P} : A_3}   
  { 
    \deduce{\Gamma_1 \vdash M : \bit}{}
    \qquad
    \deduce{\Delta_1 \vdash N : A_1}{\pi_2}
    \qquad
    \deduce{\Delta_2 \vdash P : A_2}{\pi_3}
         }  
   $
   \end{center}
   
  { 
  \tiny
  $\begin{array}{l}
  \textbf{Guard Rules} \\
  (\set{\mathfrak a_1 : \Gamma^+_1} \cup \MM, E) \totkm (\set{ \mathfrak a_2: 
  \Gamma^+_2} \cup \MM, E[l_{\mathfrak a_1}\mapsto l_{\mathfrak a_2}]_\ADD) \\
  (\set{\mathfrak a_2 : \Gamma^-_2} \cup \MM, E) \totkm (\set{\mathfrak a_1 : 
  \Gamma^-_1} \cup \MM, E[l_{\mathfrak a_2}\mapsto l_{\mathfrak a_1}]_\ADD) \\[1.5em]
  \textbf{Asynchronous Rules} \\
  (\set{(\sigma : \bit,\GG)} \cup \MM, E)   \totkm
  \left(\set{(\sigma : \bit,\GG_0),(\sigma : \bit,\GG_1)} \cup \MM', E@\GG \left[ E@\GG \xmapsto{l}(\Id, \Id)\right]\right), \\
 \ \\
    \text{provided that, for all tokens } (\sigma' : A, \GG') \in \MM, (l, 0),(l,1) \not\in \GG', \text{ and where }\\
  [0.3em]
  \MM' \text{ is } (\MM\backslash \MM\downarrow_\GG)
   \cup (\GG_1(\MM) \cup \GG_1(\ones_{\pi_2\downarrow}))
   \cup (\GG_0(\MM) \cup \GG_0((\ones_{\pi_3\downarrow}))), \   \GG_1 \text{ (resp. } \GG_0 \text{) is } \GG \cup (l \mapsto 1) \text{ (resp. } \GG\cup (l \mapsto 0)).
   \end{array}$
   }
   
   \medskip
   
   {
     \tiny
     
     In all of the following rules, the tokens $\mathfrak a_1,\mathfrak a_2,\mathfrak a_3$ are at address $\ADD$: 
     
    \medskip
    
  $\begin{array}{l}
  (\set{( \bit, \ADD)} \cup \set{\mathfrak a_3 :  \Delta_3^-} \cup \MM, E)
  \totkm
  (\set{( \bit,\ADD)} \cup \set{\mathfrak a_1: \Delta_1^+} \cup \MM, E[l_{\mathfrak a_3}\mapsto l_{\mathfrak a_1}]_\ADD), \text{provided }(l\to 0)\in \ADD \\
  (\set{( \bit, \ADD)} \cup \set{\mathfrak a_3 : \Delta_3^+} \cup \MM, E)
  \totkm
  (\set{( \bit, \ADD)} \cup \set{\mathfrak a_2 :  \Delta_2^+} \cup \MM, E[l_{\mathfrak a_3}\mapsto l_{\mathfrak a_2}]_\ADD),\text{provided }(l\to 1)\in \ADD \\
  (\set{(\bit, \ADD)} \cup \set{ \mathfrak a_3: A_3^-)} \cup \MM, E)
  \totkm
  (\set{(\bit, \ADD)} \cup \set{ \mathfrak a_1 : A_1^-)} \cup \MM, E[l_{\mathfrak a_3}\mapsto l_{\mathfrak a_1}]_\ADD),\text{provided }(l\to 0)\in \ADD \\
  (\set{(\bit, \ADD)} \cup \set{ \mathfrak a_3 : A_3^-} \cup \MM, E) 
  \totkm
  (\set{( \bit, \ADD)} \cup \set{ \mathfrak a_2 : A_2^-} \cup \MM, E[l_{\mathfrak a_3}\mapsto l_{\mathfrak a_2}]_\ADD),\text{provided }(l\to 1)\in \ADD.
  \end{array}$ 
  }
  \end{minipage}
  }

   \caption{Rule for the if-then-else.}
   \label{fig:global-tkm-ite-rule}
   \end{subfigure}

   \caption{Operational description of the rules of \circuittkmz.}
   \label{fig:global-tkm}
\end{figure*}

\begin{rem}
The situation we are interested in is that of a derivation $\pi:\ \Gamma \vdash
M:A$ where $\Gamma, A$ are \emph{first-order}, that is, they do not contain the linear arrow $\multimap$. In this case the initial positions
coincide with the base types in $\Gamma$, and the final positions coincide with
the base types in $A$. We call this a \emph{Boolean typing}. Observe that a
circuit typing is, roughly speaking, the type of a \qcirc-term with no inputs. The token
machine will then eventually produce a circuit $\Gamma \rhd E\rhd A$.
\end{rem}

To describe the token dynamics we use the following notations: 
first, we let $\mathfrak a: A^{\epsilon}$, with $\epsilon\in\{+,-\}$ indicate a token $\mathfrak a$ occurring in some either positive or negative position (depending on $\epsilon$) within the type $A$; moreover, 
given two distinct occurrences of the same type $A$ 
at different points of the derivation $\pi$, that we indicate with different indices, say $A_1,A_2$,
and a token $\mathfrak a_1: A_1^{\epsilon}$, we indicate as $\mathfrak a_2: A_2^{\epsilon'}$ a token having the same position, within $A_2$, that $\mathfrak a_1$ has within $A_1$, and with same address as $\mathfrak a_1$. Intuitively, when $\mathfrak a_1$ moves from $A_1$ to $A_2$, it ``becomes'' $\mathfrak a_2$. Notice that the polarity might change, e.g.~$A_1$ might be on the right-hand side of a sequent, and $A_2$ in the left-hand side.
Finally, for any token $\mathfrak a=(\sigma,\GG)$, we let $l_{\mathfrak a}$ be a shorthand for the associated label, i.e.~$\mathsf{lab}(\sigma)$, and for an extended circuit $E$, an address $\ADD\in @E$ and tokens $\mathfrak a,\mathfrak b$, we let $E[l_{\mathfrak{a}}\mapsto l_{\mathfrak b}]_{\ADD}$ indicate the circuit obtained by renaming $l_{\mathfrak{a}}$ as $l_{\mathfrak{b}}$ in the circuit $E@\ADD$.

The token dynamics is described by two kinds of rules. In the \emph{structural
rules}, a \emph{single} token moves upwards or downwards across the type
derivation, while the circuit is left essentially unchanged (actually, since wire labels are associated with positions, we rename one wire accordingly). These rules are essentially as in standard
GoI interpretations of the linear $\lambda$-calculus. Each transition rule is of the form
\[
(\pi,\set{\mathsf a_1:A_1^{\epsilon}}\cup\MM,E)\mapsto (\pi,\set{\mathsf a_2:A_2^{\epsilon'}}\cup\MM,E[l_{\mathfrak{a}_1}\mapsto l_{\mathfrak a_2}]_{\ADD}),
\]
corresponding to one token at address $\GG$ moving from one occurrence to another of the same formula in $\pi$, with the corresponding circuit wire being renamed accordingly.
We give half of the game
in~\Cref{fig:global-tkm-structural-rules}, while the other half is given by
symmetry: for each rule in which a positive token moves down from $A_1^{+}$ to $A_2^{+}$ (resp.~left from $A_2^+$ to $A_1^-$), there is a corresponding one in which a negative token moves up from $A_2^-$ to $A_1^-$ (resp.~right from $A_1^+$ to $A_2^-$). Notice that
there is no rule for $*$, this is because its rule (the initialization of a
token) is already taken care of when we initialize the token machine for the
first time. In the \emph{circuit rule}, 
given in \Cref{fig:global-tkm-unitary-rule}, the tokens move in a
\emph{synchronized} way: after \emph{all} necessary tokens reach the negative
occurrences of $\boolT_1,\dots, \boolT_n$, \emph{all} such tokens  move onto the
positive occurrences $\boolT_{n+1},\dots, \boolT_{2n}$; when this happens, the
gate $U$ on the corresponding wires is composed with the circuit $C$ at address
$\GG$ constructed so far, while keeping the other wires unchanged. Notice that
if $U = \termZero, \termOne$ or $\discard$ then one of the token have no label
since it is of type $1$ and does not correspond to a physical wire inside the
quantum circuit. Finally, the if-then-else rules, given in
\Cref{fig:global-tkm-ite-rule}, require a form of ``weak synchronization'': while
the tokens can enter the guard freely, this is not the case for the other
tokens. Once the guard has been evaluated and a token is present on the output
wire $l : \bit$ at address $\GG$, we create a new branching in the quantum
circuit being build. All the previous tokens at address $\GG$ are thus
duplicated, in one branch their address-choice becomes $\GG\cup\set{(l, 0)}$ while in the
other it becomes $\GG\cup\set{(l, 1)}$. Then, tokens can move up the and down the two
branches of the if-then-else, depending on the value of their adress on the
label of the guard. The side condition of that rule is to make sure we do not
apply this rule multiple times to the same if-then-else.

\begin{example}
  \label{ex:circuit1}
  Consider $M = x : \qbit, y : \qbit \vdash \opn{CNOT}(Hx, y) :
  \qbit\otimes \qbit$. We illustrate the run of the token machine over its type
  derivation in Fig.~\ref{fig:td}. If we replace the two variables $x,y$ by
  $\new\ \fc$ one can see that the produced circuit coincides with the one from
  Example \ref{ex:bell}.

   \begin{figure}
  \fbox{
\begin{minipage}{\textwidth}
  \centering
\begin{tikzpicture}
\node(a) at (0,0) {\tiny
 \begin{prooftree}

    \hypo{\vdash \opn{CNOT} : (\qbit \otimes {\qbit}) \multimap \qbit \otimes
    {\qbit}}

    \hypo{}\infer1{\vdash\opn H :\qbit\multimap \qbit}
    \hypo{x : \qbit \vdash x : \qbit}
\infer2{x:\qbit\vdash\opn Hx : \qbit}

    \hypo{y : {\qbit} \vdash y : {\qbit}}

    \infer2{x : \qbit, y : {\qbit }\vdash (\opn Hx, y) : \qbit \otimes {\qbit}}

    \infer2{x : \qbit, y : {\qbit}\vdash \opn{CNOT}(\opn H x, y) : \qbit \otimes
    {\qbit}}

  \end{prooftree}
};

\draw[->,red, rounded corners] (-2.2,-0.8) to (-2.2,-0.55) to (1.2,-0.55) to
(1.2, -0.15) to (0.55,-0.1) to (0.55,0.3) to (1.8, 0.3) to (1.8,0.7) to
(2.9,0.7) to (2.9, 0.4) to (-0.3,0.4) to (-0.3,1) to (0,1);

\draw[->,red, rounded corners] (0.4,1) to (0.6,1) to (0.6,0.55) to (1.8,0.3) to
(1.8,0) to (3.7, -0.1) to (3.7,-0.45) to (-4,-0.45) to (-4,0.35) to
(-3.25,0.35);

\draw[->,red, rounded corners] (-2.3,0.35) to (-2.05,0.35) to (-2.4,-0.4) to
(1.1,-0.4) to (1.1,-0.8);

\draw[->,blue, rounded corners]  (-1.3,-0.8) to (-1.3,-0.6) to (2.1,-0.6) to
(2.1,-0.2) to (3.5,-0.1) to (3.5,0.3) to (4.3,0.3) to (4.3, -0.12) to (4.5,
-0.12) to (4.5, -0.35) to (-3.4,-0.35) to (-3.4,0.15) to (-3.23,0.15);

\draw[->,blue, rounded corners]  (-2.3,0.15) to (-1.7,0.15) to (-1.7,-0.25) to
(1.8, -0.25) to (1.8,-0.8);

\node[rectangle, red, rounded corners, draw, thick, minimum height=0.5cm] (cnot)
at (0.2,1) {\tiny $\mathrm{H}$};

\node[rectangle, red, rounded corners, draw, thick, minimum height=0.5cm] (cnot)
at (-2.77,0.2) {\tiny $\mathrm{CNOT}$};

\end{tikzpicture}

  \end{minipage}
  }
  \caption{Example of execution of the token machine.}
  \label{fig:td}
  \end{figure}
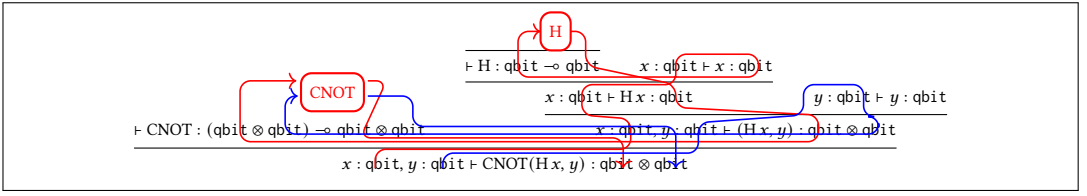
\end{example}

\subsection{Soundness}
\label{sect:soundness}

In this section we prove that the \circuittkmz is sound with respect to the
operational semantics of \qlambda\!, that is, that the circuits
produced by the machine perfectly match the quantum protocols described by the
corresponding \qlambda-terms. More details are given in the Appendix.

The situation we are interested in is the one in which we are given a
circuit typing $\Gamma\vdash M:A$ for some \qlambda-term $M$. Let us
first highlight an important difference between \qlambda\ and the
\circuittkmz. On the one hand, the evaluation of \qlambda-terms relies
on a quantum register and is intrinsically
probabilistic, due to the measure operator; in other words, the
different evaluations of a closure $m=[Q,L,M]$ produce a
\emph{distribution} $\DIST{m}$ of quantum closures. On the other
hand, the execution of the token machine does not rely on quantum
registers and is entirely deterministic: the measure operator results
in the introduction in the circuit of a measure gate that is
\emph{not} evaluated at compile time. In particular, the evaluation of
the \circuittkmz\ over $M$ does not produce a distribution of quantum
states, but a quantum circuit $C_M = \translate(E_M)$.

The soundness of the token machine is formulated as a result relating
the association $m\mapsto \DIST{m}$ between a quantum closure and
the distribution of quantum closures produced by
the evaluation of $m$, with the action of the quantum
circuit $C_M$ (or rather, of its $\cat{CPM}$-interpretation) over
mixed states. 

Let $\pi:\Gamma\vdash M:A$ be a Boolean typing derivation. 
Supposing $\Gamma=\{x_1:\mathbb B_1,\dots, x_m:\mathbb B_m\}$, we can identify
$\Gamma$ with the circuit environment $\mathbf{ENV}(\negative_\pi)=\{l_1:\mathbb B_1,\dots, l_n:\mathbb B_n\}$: notice that each term variable $x_i:\mathbb B_i$ is associated with the wire label $l_i$ corresponding to the position of $\mathbb B_i$ in $\pi$. 
Similarly, in any quantum closure $m=[Q,L,M]$, the term variable assignments in $L$ can be identified with $\mathbf{ENV}(\negative_\pi)$ (more precisely, 
 $|L|=\mathbf{ENV}(\negative_\pi)$, where for a list $L$ over some set $S$, we indicate as $|L|\subseteq S$ its underlying set).
 In this way, $m$ naturally induces a corresponding (pure) state $\tocpm(m)\in \interp{\mathbf{ENV}(\negative_\pi)}$ (notice that this state only depends on $Q$ and $L$). 
Analogously, at the end of the computation of $m$, in the produced distribution $ \DIST{m}=\{m_1^{\mu_1},\dots, m_k^{\mu_k}\}$, where $m_i=[Q_i,L_i,V_i]$, each $L_i$ can be identified with the circuit environment $\mathbf{ENV}(\positive_\pi)$, thus yielding a (mixed) state $\tocpm(\DIST{m})=\sum_i \mu_i\cdot \tocpm(m_i)\in \interp{\positive_\pi}$.

This leads to the following:

\begin{theorem}[Soundness of the
\circuittkmz]\label{thm:soundness_tkm}
Let $\pi :\ \Gamma \vdash M :A$ be a Boolean typing derivation and suppose
that the \circuittkmz{} produces, over $\pi$, the final configuration
$(\pi,\MM, E )$. Then, for any quantum closure $m=[Q,L,M]$, $\interp{E}(\tocpm(m)) =
\tocpm(\DIST{m})$.
\end{theorem}

We provide a sketch of the proof of Theorem \ref{thm:soundness_tkm}. The
fundamental ingredient is the introduction of yet another token machine, called
the \evaltkm\ (Quantum Memory-based Synchronous Interaction Abstract Machine)
whose execution is closer to the evaluation of \qlambda-terms, being
probabilistic. This machine is essentially the same as the $\textsf{MSIAM}$ from
\cite{lago2013wavestyletokenmachinesquantum,DalLago2017}. Similarly to the
\circuittkmz{}, the \evaltkm{} has multiple tokens that move from initial
positions to final positions inside the type derivation of $M$; while the tokens
are defined as for the \circuittkmz, the address information is now redundant
(although we keep it for uniformity): as a reduction of the machine only
traverses \emph{one} slice of $\pi$, all tokens have \emph{the same} address
$\GG$. Moreover, unlike the \circuittkmz, the machine does \emph{not} produce a
circuit; instead, the machine has access to a quantum register $[L,Q,c]$ that is
updated during execution. This register contains the current states of the
qubits (via $Q$) and bits (via $c$) in the positions occupied by the tokens
(accounted for by a list of label assignments $L$).

A configuration of the \evaltkm\ is thus of the form $\config=(\pi,\MM,m)$,
where $m=[L,Q,c]$, all tokens in $\MM$ are at the same address $\config_{\GG}$
and $|L|=\mathbf{ENV}(\MM)$. Initial and final configurations can be defined
similarly to the \circuittkmz (see the Appendix for more details).

The structural rules of the \evaltkm\ are as for the \circuittkmz: a single
token moves and the quantum and classical states $Q,c$ are not updated (yet, similarly to the \circuittkmz, the labels in $L$ are possibly renamed); instead, the behavior of
the \evaltkm\ differs when the tokens travel through a quantum gate $U\in \QO$: in this case the register is updated by directly applying the gate $U$ to the
quantum state. Observe that, when $U=\meas$, the update is inherently
probabilistic: the machine actually performs a measurement on its quantum
register.
 In the case of an if-then-else $\Gamma,\Delta\vdash\ite{N}{P_0}{P_1}:A$, the machine behaves as follows: 
 on the one hand, 
we have structural rules allowing tokens to
move upwards through the derivation of $N$; in this way a token may end up in
the positive position $N:\bit$; when this happens, the classical part of the register will contain a
value $b\in \{0,1\}$, which is added to the current address. We then have a second class of rules that allows a token
in a negative position in either $\Delta$ or $A$ to move upwards within the
sub-derivation of $P_b$: this may only happen once the Boolean $b$ has been
determined. 
Observe that, in any execution, only one among the sub-derivations of $P_0$ and
$P_1$ are actually visited by the tokens. Which one may depend on the execution
itself (as the Boolean $b$ produced might depend on previous measurements).

In an initial state for the \evaltkm\  one consider a quantum register
$m=[Q,L,c]$. By considering \emph{all} possible executions of the machine over
$m$, we obtain a distribution $\DISTM{m}$ of quantum states produced by the
machine. Similarly to what is done with the $\textsf{MSIAM}$
machine~\cite{lago2013wavestyletokenmachinesquantum,DalLago2017}, it is not
difficult to show that the distribution $\DISTM{m}$ perfectly matches the
distribution $\DIST{m'}$ produced by the \qlambda-evaluation of the closure
$m':=[Q,L,M]$. In this way, we are reduced to the problem of defining a
\emph{(bi)simulation} between the \evaltkm\ and  the \circuittkmz. 

Let now $\config=(\pi,\MM,E)$ be a configuration
of the \circuittkmz.
Rather than relating $\config$ with one similar configuration of the \evaltkm, we 
will relate it with a \emph{distribution} of such configurations $\mu=\{ \config_{s}^{\mu_{s}}\}_{{s}\in @^{\mathsf s}E}$, 
where $\config_s=(\pi,\MM_s,m_s)$, indexed by all super-addresses $s\in @^{\mathsf s}E$: the intuition is that any $s$, which provides choices for both the if-then-else and the $\meas$-gates of the underlying term, captures \emph{one possible evolution} of the \evaltkm. The distribution $\mu$ accounts then for all such possible evolutions, each with its own probability.

Let $m=[L,Q,c]$ be a register, to be used to initialize the \evaltkm. 
The relation $\mu \ \mathscr R^{m}_\pi\ \config$ holds when, for every $s\in @^{\mathsf s}E$, the following hold:
\begin{varitemize}
  \item the tokens in $\config_{s}$ coincide with those in $\config$ at address $\ADD_s$, i.e.~$\MM_{s}=\MM\vert_{\ADD_s}$;
  \item the application of the completely positive map $\interp{E}^{\sharp s}:\interp{\mathbf{ENV}(\negative_\pi)}\to\interp{\mathbf{ENV}(\MM|_{\ADD_s})}$ to the register $m$ produces \emph{the same} state as the one in $m_s$, taken with its own probability, i.e. $  \interp{E}^{\sharp s}( \tocpm(m))=\mu_{s}\cdot\tocpm(m_s)\in \interp{\mathbf{ENV}(\MM_s)}=\interp{\mathbf{ENV}(\MM_{\ADD_s})}$.
\end{varitemize}

\begin{figure}
\fbox{
\begin{minipage}{.99\textwidth}
\begin{center}
\medskip

\resizebox{.5\textwidth}{!}{
\begin{tikzpicture}

\draw (-3,-1.5) to (-2.4,-1.5) to (-2.4,1.5) to (-3,1.5) to (-3,-1.5);
\node(m) at (-2.7,0) {$m$};

\draw (-1.5,-1.5) to (0.5,-1.5) to (0.5,1.5) to (-1.5,1.5) to (-1.5,-1.5);
\node(m) at (-0.5,0) {$E$};

\draw (-2.4,1) to (-1.5,1);
\draw (-2.4,-1) to (-1.5,-1);
\node(d) at (-1.9,0.1) {$\vdots$};
\node[rotate=90](nd) at (-2.1,0) {\tiny$\negative_\pi$};

\draw (0.5,1) to (1.4,1);
\draw (0.5,-1) to (1.4,-1);
\node(d) at (0.8,0.1) {$\vdots$};
\node[rotate=90](nd) at (1.1,0) {\tiny$\positive_\pi$};

\node(eq) at (2,0) {\LARGE$=$};
\node(eq) at (3.7,0) {
\begin{minipage}{.2\textwidth}
\LARGE
$$
\sum_{s\in @^{\mathsf s}E} \mu_s\ \ \cdot
$$
\end{minipage}
};

\draw (5.8,-1.5) to (5.2,-1.5) to (5.2,1.5) to (5.8,1.5) to (5.8,-1.5);
\node(m) at (5.5,0) {$m_s$};

\draw (5.8,1) to (6.7,1);
\draw (5.8,-1) to (6.7,-1);
\node(d) at (6.1,0.1) {$\vdots$};
\node[rotate=90](nd) at (6.4,0) {\tiny$\positive_\pi$};

\end{tikzpicture}
}

\medskip

\end{center}
\end{minipage}
}
\caption{Illustration of the (bi)simulation equation $\interp{E}( \tocpm(m))=\sum_s\mu_s\cdot\tocpm(m_s)=\tocpm(\mu)$.
}
\label{fig:bisimulation}
\end{figure}

\begin{theorem}[the \evaltkm\ simulates the \circuittkmz]
  \label{thm:bisimulation}
Let $\config,\config'$ be configurations of the \circuittkmz\ and $\mu$
be a distribution of configurations for the \evaltkm. 
For all registers $m$, 
If $\mu \
\mathscr R^{m}_\pi \ \config $ and $\config\to \config'$ then there
exists a distribution $\mu'$ such that $\mu' \ \mathscr R^{m}_\pi \ \config'$
and such that multiple parallel executions of the \evaltkm \ lead from $\mu$ to
$\mu'$.
\end{theorem}

We can in fact show that the opposite direction is also true, i.e. we
have a \emph{bisimulation} between the two machines.
The bisimulation allows us to
transport rewrite results (confluence, termination and deadlock
freeness) of the
\evaltkm{}~\cite{lago2013wavestyletokenmachinesquantum,DalLago2017}
onto the \circuittkmz{}.

\begin{corollary}[Termination]
  Given a typing derivation $\pi : \Gamma \vdash M : A$, any run of the \circuittkmz{}
  over the initial configuration for $\pi$ terminates over a final
  configuration.
\end{corollary}

\begin{corollary}[Confluence up to circuit equivalence]
Given configurations $\config, \config_1,\config_2$ such that $\config\to\config_1$ and $\config\to\config_2$, there exists configurations $\donfig_1=(\pi,\mathcal M,E_1)$ and
$\donfig_2=(\pi,\mathcal M,E_2)$, such that $\config_1\to^*\donfig_1$, $\config_2\to^*\donfig_2$ and $\interp{E_1}=\interp{E_2}$.
\end{corollary}
Theorem~\ref{thm:soundness_tkm} follows from Theorem \ref{thm:bisimulation} and its corollaries: if both $\config$ and $\mu$ are final configurations, so that $\mathbf{ENV}(\negative_\pi)\rhd E\rhd \mathbf{ENV}(\positive_\pi)$ is uniform and all registers $m_s$ contain a state for $\mathbf{ENV}(\positive_\pi)$, 
 from the validity of  $\interp{E}^{\sharp s}( \tocpm(m))=\mu_{s}\cdot\tocpm(m_s)$ for all super-address $s\in @^{\mathsf s}E$, we deduce the validity of the following equation (illustrated in Fig.~\ref{fig:bisimulation}):
 \[
 \interp{E}( \tocpm(m))=\sum_s\mu_s\cdot\tocpm(m_s)=\tocpm(\mu).
 \] 
Observing that, when $m$ is constructed from a closure $m'=[Q,L,M]$, $\tocpm(\mu)$ must coincide with $\DISTM{m}=\DIST{m'}$, the equation above leads then immediately to Theorem \ref{thm:soundness_tkm}.

\subsection{The Full Machine}
\label{sect:full-machine}

The machine we have presented thus far follows the naïve approach that
we sketched in Section~\ref{sect:informal} and is thus not very
efficient. For instance, take the program $\pv{\lambda x.
\ite{M}{U_1~x}{U_2~x}}{V~y}$ where $U_1, U_2$ and $V$ are quantum
gates. Assume that the \circuittkmz first compiles the if-then-else. We
end up with an extended circuit whose two branches both contain a copy of the gate
$V$. However, this duplication of $V$ should not to happen, as we
could instead try to compile as if it were one big unitary gate to
which we then compose $V$ at the end, only once. This session is
dedicated to implementing this observation via a different, \emph{synchronous}, compilation rule.

The basic idea is to treat if-then-elses like unitary gates: the
tokens will synchronize at the negative positions in the conclusion of the if-then-else
and, once they are all present, two subroutines will be executed,
compiling the two branches of the if-then-else into two circuits $C_1$
and $C_2$. The resulting tokens and circuits will then be \emph{merged},
without having to change and duplicate the tokens in the rest of the
derivation. As discussed in Section 2, this approach is, however, not always possible, since we cannot assume that the relevant tokens will reach, in all situations, the negative conclusion of an if-then-else: a deadlock may arise as soon as different synchronization points (conditionals or quantum gates) wait each for tokens that should exit from the other.
We therefore have two
situations: either all the tokens are present in the negative
conclusion of an if-then-else, in which case we trigger the
synchronous rule, or we do not and the machine is stuck somewhere, in
which case we trigger an asynchronous-rule as described in the previous subsection. 

The machine \circuittkm is defined as follows: the structural rules 
(Figure~\ref{fig:global-tkm-structural-rules}) as well as the unitary rule
(Figure~\ref{fig:global-tkm-unitary-rule}) do not change; for
if-then-else instead we now have the following rules:
\begin{itemize}
  \item The \emph{synchronous rule} asks that the conclusion of the if-then-else
  is negatively saturated and that the guard has been evaluated with some label
  $l$ on the output bit of the guard. In that case, the two branches are run as
  subroutines over their sub-derivations $\pi_1,\pi_2$, until termination. This way, two extended circuit $E_1$ and $E_2$ are
  produced, \emph{with the same environment}, since $\mathbf{ENV}(\negative_{\pi_1})=\mathbf{ENV}(\negative_{\pi_2})$ and $\mathbf{ENV}(\positive_{\pi_1})=\mathbf{ENV}(\positive_{\pi_2})$. Then, the plain circuits $\tau(E_1)$ and $\tau(E_2)$ are put inside
  a plain if-then-else and the tokens are merged back on the positive
  occurrences of the conclusion of the if-then-else. 
  \item The \emph{guard rule} and the \emph{asynchronous} rules are left
  unchanged, aside from the side condition that \emph{no other rewriting rule
  applies}: i.e. the asynchronous rule only triggers when we don't have the
  opportunity to apply any synchronous rules.  
\end{itemize}

\begin{figure}
    \dbox{
      \begin{minipage}{\textwidth}
     \begin{center}
   \footnotesize
   $
   \infer{\Gamma_2, \Delta_3 \vdash \ite{M}{N}{P} : A_3}   
   { \deduce{\Gamma_1 \vdash M : \bit}{}
   \qquad
   \deduce{ \Delta_1 \vdash N : A_1}{\pi_2}
   \qquad
   \deduce{ \Delta_2 \vdash P : A_2}{\pi_3}
          }  
    $
    \end{center}
    
   { \tiny $\begin{array}{l}
         \textbf{Synchronous Rule} \\
         (\set{\bit,\GG} \cup \MM, E) \totkm \left (\MM \cup \GG(\MM_1), E@\GG\Big[E@\GG ;
         \ite{l}{\tau(E^*_1)}{\tau(E^*_2)}\Big]\right) \text{ with } l=\mathsf{lab}(\bit), \text{ and where: }\\
         - \MM|_{\GG} \textbf{ negatively saturates } \Delta_3, A_3, \\
         - \mathtt{Exec}(\pi_2) = (\pi_2, \MM_1, E_1), \\
         - \mathtt{Exec}(\pi_3) = (\pi_3, \MM_2, E_2), \\
         - E^*_b:= E_b\otimes\mathrm{id}_\Gamma.
         \end{array}$ 
         }
    \end{minipage}
   }
 
     \caption{Synchronous rule of the \circuittkm.}
    \label{fig:tkm-rules-ite-optimized}
 \end{figure}

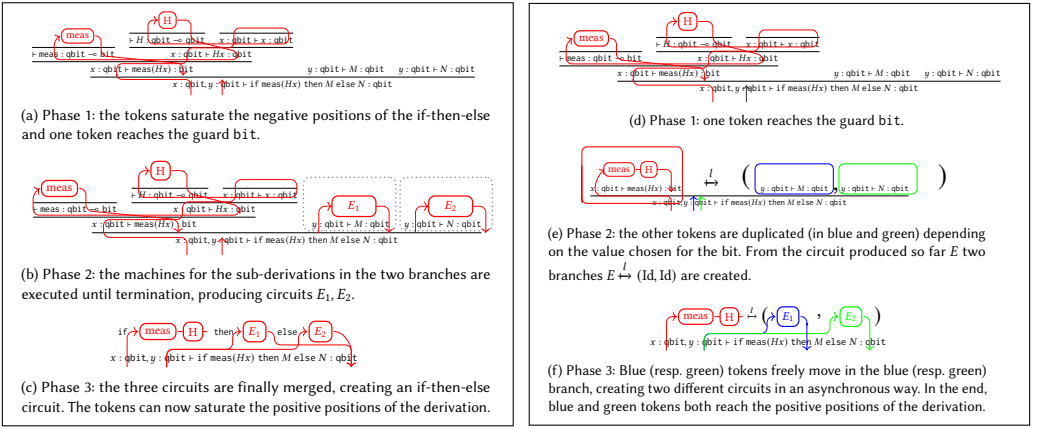
\begin{figure*}

\fbox{
  \begin{minipage}{.46\textwidth}
    \adjustbox{scale=0.65}{
      \begin{minipage}{\textwidth} 
      \begin{subfigure}{1.5\textwidth}
      \adjustbox{scale=0.8}{
      \begin{tikzpicture}
      \node(a) at (0,0) { \adjustbox{scale=0.7}{
      \begin{prooftree}
      \infer0{\vdash \mathtt{meas}: \qbit\multimap \bit} \infer0{\vdash H:
        \qbit\multimap \qbit} \infer0{x:\qbit\vdash x: \qbit} \infer2{x:\qbit\vdash
        Hx:\qbit} \infer2{x:\qbit\vdash \mathtt{meas}(Hx):\bit}
    
    \hypo{y : \qbit \vdash M : \qbit}
    \hypo{y : \qbit \vdash N : \qbit}
        \infer3{x : \qbit, y:\qbit \vdash \ite{\mathtt{meas}(Hx)}{M}{N} : \qbit}
      \end{prooftree}
    } };
    
    \draw[->,red, rounded corners] (-3.5,-0.4) to (-3.5,-0.1) to (-1.7,-0.1) to
    (-1.7,0.3) to (-0.5,0.3) to (-0.5,0.8) to (0.9,0.8) to (0.9,0.4) to (-2.7,0.4)
    to (-2.7,1) to (-2.45,1);
    
    \node[red, rectangle, rounded corners, draw, thick] at (-2.2,1)
    {\footnotesize$\mathrm{H}$};

    \draw[->,red, rounded corners] (-1.95,1) to (-1.6,1) to (-1.6,0.5) to (-0.4,0.2)
    to (-0.4,0) to (-5.1,0) to (-5.1,0.5) to (-4.95,0.6);
    
    \node[red, rectangle, rounded corners, draw, thick] at (-4.5,0.6)
    {\footnotesize$\mathrm{meas}$};
    
    \draw[->,red, rounded corners] (-4.05,0.6) to (-3.9,0.6) to (-3.9,0.1) to
    (-1.8,-0.1) to (-1.8,-0.4);
    
    \draw[red, rounded corners] (-1.6,-0.9) to (-1.6,-0.5) to (-1.95, -0.5) to (-3.5, -0.5) to (-3.5, -0.4);
    
    \draw[red, ->] (-0.8,-0.9) to (-0.8,-0.5);
    
    \end{tikzpicture}
    }
    
    \caption{Phase 1: the tokens saturate the negative positions of the if-then-else and one token reaches the guard $\bit$.}
    \end{subfigure}
    
    \bigskip

    \begin{subfigure}{1.5\textwidth}
      \adjustbox{scale=0.9}{
      \begin{tikzpicture}
      \node(a) at (0,0) { \adjustbox{scale=0.63}{
      \begin{prooftree}
      \infer0{\vdash \mathtt{meas}: \qbit\multimap \bit} \infer0{\vdash H:
        \qbit\multimap \qbit} \infer0{x:\qbit\vdash x: \qbit} \infer2{x:\qbit\vdash
        Hx:\qbit} \infer2{x:\qbit\vdash \mathtt{meas}(Hx):\bit}
    
    \hypo{y : \qbit \vdash M : \qbit}
    \hypo{y : \qbit \vdash N : \qbit}
        \infer3{x : \qbit, y:\qbit \vdash \ite{\mathtt{meas}(Hx)}{M}{N} : \qbit}
      \end{prooftree}
    } };
    
    \draw[->,red, rounded corners] (-3.5,-0.4) to (-3.5,-0.1) to (-1.7,-0.1) to
    (-1.7,0.3) to (-0.5,0.3) to (-0.5,0.8) to (0.9,0.8) to (0.9,0.4) to (-2.7,0.4)
    to (-2.7,1) to (-2.45,1);
    
    \node[red, rectangle, rounded corners, draw, thick] at (-2.2,1)
    {\footnotesize$\mathrm{H}$};

    \draw[->,red, rounded corners] (-1.95,1) to (-1.6,1) to (-1.6,0.5) to (-0.4,0.2)
    to (-0.4,0) to (-5.1,0) to (-5.1,0.5) to (-4.95,0.6);
    
    \node[red, rectangle, rounded corners, draw, thick] at (-4.5,0.6)
    {\footnotesize$\mathrm{meas}$};
    
    \draw[->,red, rounded corners] (-4.05,0.6) to (-3.9,0.6) to (-3.9,0.1) to
    (-1.8,-0.1) to (-1.8,-0.4);
    
    \draw[red, rounded corners] (-1.6,-0.9) to (-1.6,-0.5) to (-1.95, -0.5) to (-3.5, -0.5) to (-3.5, -0.4);
    
    
    \node[rectangle, draw, rounded corners, dotted, thick, gray,minimum width=2.1cm,
    minimum height=1.2cm] (rect) at (2.1,0.25) {};
    
    \node[red, rectangle, rounded corners, draw, minimum width=1cm, thick] at (2.2,0.2) {\footnotesize$E_1$};
    
    \draw[->,red, rounded corners] (1.4,-0.4) to (1.4,-0.2) to (1.4,0.2) to (1.7,0.2);
    \draw[<-,red, rounded corners] (3,-0.4) to (3,-0.2) to (3,0.2) to (2.7,0.2);

    \node[rectangle, draw, rounded corners, dotted, thick,gray, minimum width=2.1cm,
    minimum height=1.2cm] (rect) at (4.3,0.25) {}; 
    
    \node[red, rectangle, rounded corners, draw, minimum width=1cm, thick] at (4.4,0.2) {\footnotesize$E_2$};
    
    \draw[->,red, rounded corners] (3.6,-0.4) to (3.6,-0.2) to (3.6,0.2)
    to (3.9,0.2); 
    \draw[<-,red, rounded corners] (5.2,-0.4) to (5.2,-0.2) to (5.2,0.2)
    to (4.9,0.2);
    
    \draw[red, ->] (-0.8,-0.9) to (-0.8,-0.5);
    
    \end{tikzpicture}
    }
    
    \caption{Phase 2: the machines for the sub-derivations in the two branches are executed until termination, producing circuits $E_1,E_2$.}
    \end{subfigure}
    
    \bigskip
    
    \begin{subfigure}{1.5\textwidth}
      \adjustbox{center, scale=0.9}{
     \begin{tikzpicture}
      \node(a) at (-0.3,-0.1) { \adjustbox{scale=0.7}{ $x : \qbit, y:\qbit \vdash
      \ite{\mathtt{meas}(Hx)}{M}{N} : \qbit$} };

    \node(if) at (-2.8,0.5) {\tiny$\mathtt{if}$};
    
    \node[rectangle, thick, red, draw, rounded corners](meas) at (-2,0.5)
    {\footnotesize$\mathrm{meas}$}; \node[rectangle, thick, red, draw, rounded
    corners](meas) at (-1.2,0.5) {\footnotesize$\mathrm{H}$}; \draw[red] (-1.55,
    0.5) to (-1.42,0.5); \draw[red] (-2.45, 0.5) to (-2.55,0.5); \draw[red] (-0.98,
    0.5) to (-0.88,0.5);
    
    \node(then) at (-0.5,0.5) {\tiny$\mathtt{then}$};
    
    \node[rectangle, thick, red, draw, rounded corners](meas) at (0.2,0.5)
    {\footnotesize$E_1$};
    
    \node(then) at (0.95,0.5) {\tiny$\mathtt{else}$};
    
    \node[rectangle, thick, red, draw, rounded corners](meas) at (1.7,0.5)
    {\footnotesize$E_2$};

    \draw[->, rectangle, rounded corners, red] (-2.7,-0.3) to (-2.7,0.5) to
    (-2.45,0.5);

    \draw[->, rectangle, rounded corners, red] (-1.8,-0.3) to (-1.8,0.1) to
    (-0.3,0.1) to (-0.3,0.5) to (-0.1,0.5); \draw[->, rectangle, rounded corners,
    red] (-1.8,-0.3) to (-1.8,0.1) to (1.2,0.1) to (1.2,0.5) to (1.4,0.5);
    
    \draw[->, rectangle, rounded corners, red] (0.5,0.5) to (0.65,0.5) to (0.65,0.2)
    to (2.4,0.2) to (2.4,-0.3); \draw[->, rectangle, rounded corners, red] (2,0.5)
    to (2.15,0.5) to (2.4,0.2) to (2.4,-0.3);

    \end{tikzpicture}

    }
    
    \caption{Phase 3: the three circuits are finally merged, creating an if-then-else circuit. The tokens can now saturate the positive positions of the derivation.     }
    \end{subfigure}
    \end{minipage}
    }
\end{minipage}
}
\hfil
\fbox{
  \begin{minipage}{.45\textwidth}
    \adjustbox{scale=0.62}{
      \begin{minipage}{\textwidth}
      \begin{subfigure}{1.5\textwidth}
      \adjustbox{scale=0.9}{
      \begin{tikzpicture}
      \node(a) at (0,0) { \adjustbox{scale=0.65}{
      \begin{prooftree}
      \infer0{\vdash \mathtt{meas}: \qbit\multimap \bit} \infer0{\vdash H:
        \qbit\multimap \qbit} \infer0{x:\qbit\vdash x: \qbit} \infer2{x:\qbit\vdash
        Hx:\qbit} \infer2{x:\qbit\vdash \mathtt{meas}(Hx):\bit}
    
    \hypo{y : \qbit \vdash M : \qbit}
    \hypo{y : \qbit \vdash N : \qbit}
        \infer3{x : \qbit, y:\qbit \vdash \ite{\mathtt{meas}(Hx)}{M}{N} : \qbit}
      \end{prooftree}
    } };
    
    \draw[->,red, rounded corners] (-3.5,-0.4) to (-3.5,-0.1) to (-1.7,-0.1) to
    (-1.7,0.3) to (-0.5,0.3) to (-0.5,0.8) to (0.9,0.8) to (0.9,0.4) to (-2.7,0.4)
    to (-2.7,1) to (-2.45,1);
    
    \node[red, rectangle, rounded corners, draw, thick] at (-2.2,1)
    {\footnotesize$\mathrm{H}$};

    \draw[->,red, rounded corners] (-1.95,1) to (-1.6,1) to (-1.6,0.5) to (-0.4,0.2)
    to (-0.4,0) to (-5.1,0) to (-5.1,0.5) to (-4.95,0.6);
    
    \node[red, rectangle, rounded corners, draw, thick] at (-4.5,0.6)
    {\footnotesize$\mathrm{meas}$};
    
    \draw[->,red, rounded corners] (-4.05,0.6) to (-3.9,0.6) to (-3.9,0.1) to
    (-1.8,-0.1) to (-1.8,-0.4);
    
    \draw[red, rounded corners] (-1.6,-0.9) to (-1.6,-0.5) to (-1.95, -0.5) to (-3.5, -0.5) to (-3.5, -0.4);
    
    \draw[->] (-0.8,-0.9) to (-0.8,-0.5);

    \end{tikzpicture}
    }
    
    \caption{Phase 1: one token reaches the guard $\bit$.}
    \end{subfigure}
    
    \bigskip

    \begin{subfigure}{1.5\textwidth}
    \begin{center}
      \adjustbox{scale=0.8}{
      \begin{tikzpicture}
      \node(a) at (0.5,-0.5) { 
      \adjustbox{scale=0.65}{
      \begin{prooftree}
    \hypo{x:\qbit\vdash \mathtt{meas}(Hx):\bit}
    \hypo{\ \ \ \ \ \ \ \ \ \ \ \ \ \ \  \ \ \ \ \ \ \ \ \ \ }
    \hypo{y : \qbit \vdash M : \qbit}
    \hypo{y : \qbit \vdash N : \qbit}
        \infer4{x : \qbit, y:\qbit \vdash \ite{\mathtt{meas}(Hx)}{M}{N} : \qbit}
      \end{prooftree}
    } };
    
      \node[rectangle, thick, red, draw, rounded corners](meas) at (-3,0.2)
    {\footnotesize$\mathrm{meas}$}; 
    \node[rectangle, thick, red, draw, rounded
    corners](meas) at (-2.2,0.2) {\footnotesize$\mathrm{H}$}; 
    \draw[red] (-2.55,  0.2) to (-2.42,0.2); 
    \draw[red, rounded corners, ->] (-1.98,  0.2) to (-1.6,0.2) to (-1.6,-0.4);
    \draw[red, rounded corners, thick] (-4,-0.7) to (-4,0.8) to (-1.3,0.8) to (-1.3,-0.7) to (-4,-0.7) to (-4,0.3);

    \draw[red, rounded corners, ->] (-1.6,-0.9) to (-1.6,-0.5) to (-1.95, -0.5) to (-3.6, -0.5) to [bend left]
    (-3.4,0.2);
    
    \node(x) at (-0.5,0) {\Large $\xmapsto{l}$};

    
\node(par) at (0.4,0) {\Huge $($};
\node(par) at (5.7,0) {\Huge $)$};
\node(par) at (2.83,-0.35) {\Huge $,$};

    \node[rectangle, draw, rounded corners,  thick, blue,minimum width=2.1cm,
    minimum height=0.8cm] (rect) at (1.7,-0.05) {};
       
    \node[rectangle, draw, rounded corners,  thick,green, minimum width=2.1cm,
    minimum height=0.8cm] (rect) at (3.95,-0.05) {}; 
    
    \draw[blue, ->] (-1,-0.9) to (-1,-0.5);
    
    \draw[green, ->] (-0.8,-0.9) to (-0.8,-0.5);

    \end{tikzpicture}
    }
    \end{center}
    \caption{Phase 2: the other tokens are duplicated (in blue and green) depending on the value chosen for the bit. From the circuit produced so far $E$ two branches $E\xmapsto{l}(\mathrm{Id},\mathrm{Id})$ are created.}
    \end{subfigure}
    
    \bigskip
    
    \begin{subfigure}{1.5\textwidth}
      \adjustbox{scale=0.9, center}{
    
     \begin{tikzpicture}
      \node(a) at (-0.3,-0.1) { \adjustbox{scale=0.7}{ $x : \qbit, y:\qbit \vdash
      \ite{\mathtt{meas}(Hx)}{M}{N} : \qbit$} };

    \node[rectangle, thick, red, draw, rounded corners](meas) at (-2,0.5)
    {\footnotesize$\mathrm{meas}$}; \node[rectangle, thick, red, draw, rounded
    corners](meas) at (-1.2,0.5) {\footnotesize$\mathrm{H}$}; \draw[red] (-1.55,
    0.5) to (-1.42,0.5); \draw[red] (-2.45, 0.5) to (-2.55,0.5); \draw[red] (-0.98,
    0.5) to (-0.88,0.5);
    
    \node(then) at (-0.65,0.6) {\tiny$\xmapsto{l}$};
    \node(then) at (-0.35,0.5) {\LARGE$($};
    
    \node[rectangle, thick, blue, draw, rounded corners](meas) at (0.2,0.5)
    {\footnotesize$E_1$};
    
    \node(then) at (0.95,0.5) {\LARGE$,$};
    
    \node[rectangle, thick, green, draw, rounded corners](meas) at (1.7,0.5)
    {\footnotesize$E_2$};
        \node(then) at (2.35,0.5) {\LARGE$)$};

    \draw[->, rectangle, rounded corners, red] (-2.7,-0.3) to (-2.7,0.5) to
    (-2.45,0.5);
    
    \draw[->, rectangle, rounded corners, blue] (-1.8,-0.3) to (-1.8,0.1) to
    (-0.3,0.1) to (-0.3,0.5) to (-0.1,0.5); 
    
    \draw[->, rectangle, rounded corners,
    green] (-1.8,-0.3) to (-1.8,0.1) to (1.2,0.1) to (1.2,0.5) to (1.4,0.5);
    
    \draw[->, rectangle, rounded corners, blue] (0.5,0.5) to (0.65,0.5) to (0.65,0.2)
    to (0.65,0.2) to (0.65,-0.3);
    
    \draw[->, rectangle, rounded corners, green] (2,0.5)
    to (2.15,0.5) to (2.15,0.2) to (2.15,-0.3);

    \end{tikzpicture}
    }
    
    \caption{Phase 3: Blue (resp.~green) tokens freely move in the blue (resp.~green) branch, creating two different circuits in an asynchronous way. In the end, blue and green tokens both reach the positive positions of the derivation.}
    \end{subfigure}
    \end{minipage}
    }
  \end{minipage}
}

\caption{Synchronous (left) vs Asynchronous (right) execution of the machine.}
\label{fig:tkm_ex2}
\end{figure*}

To define the synchronous rule in a more precise way let us introduce the
following notion:

\begin{definition}[Execution]
  \label{def:execution}
  For any type derivation $\pi:\Gamma\vdash M:A$, if the \circuittkm
   terminates on a final configuration $\config$, let
   $\exec(\pi)=\config$; otherwise, let $\exec(\pi)$ be undefined.

\end{definition}

The synchronous rule is defined in Figure~\ref{fig:tkm-rules-ite-optimized} and illustrated in the following example.

\begin{example}
  \label{ex:circuit2}
  Let $x:\qbit\vdash M,N:\qbit$ be two terms and consider the term $T
  = \ite{P}{M}{N}$, where $P=\mathtt{meas}(\mathtt{H}x)$,
  corresponding to a fair flip coin. The difference between compiling $T$ with the
  synchronous or the asynchronous rule is illustrated
  in~\Cref{fig:tkm_ex2}. On the left, the synchronous rule waits for
  all the tokens to be present before compiling the two terms and
  merging the resulting tokens. On the right, the asynchronous
  rule, which does not allow for tokens to be merged, duplicates the tokens via a new address choice (in the figure, one black token yields one blue and one green token). Notice that every position, outside of the two branches, that is visited afterwards will be actually visited \emph{twice}: once in blue, once in green. This is the main source of inefficiency of the asynchronous rule.

\end{example}

For the machine to be well-defined, we need to ensure that
$\exec(-)$ is always well-defined in all application of the synchronous rule. As before, this can be
shown by a simulation between the \circuittkm{} machine and the \evaltkm{}
machine. This new simulation is a bit more subtle since the \circuittkm{} now
has two kinds of rules for the if-then-else, but it can be done following the same
methodology as described for \circuittkmz{}.

\begin{theorem}[the \evaltkm\ simulates the \circuittkmz]
  \label{thm:simulation-opti}
Let $\config,\config'$ be configurations of the \circuittkmz\ and $\mu$
be a distribution of configurations for the \evaltkm. For all registers $m$, if $\mu \
\mathscr R^{m}_\pi \ \config $ and $\config\to \config'$ then there
exists a distribution $\mu'$ such that $\mu' \ \mathscr R^{m}_\pi \ \config'$
and such that multiple parallel executions of the \evaltkm \ lead from $\mu$ to
$\mu'$.
\end{theorem}

This time we only have simulation and not a bisimulation. This is however enough
to argue for the soundness of the machine as we did for Theorem \ref{thm:soundness_tkm}, and in particular to show that the \circuittkm{} terminates:

\begin{corollary}[Termination]
  Given a typing derivation $\pi : \Gamma \vdash M : A$, 
   any run of the \circuittkm{} over the initial
  configuration of $\pi$ terminates in a final configuration.
\end{corollary}

Due to the lack of bisimulation, we cannot prove confluence as we did for the \circuittkm{}. Still, we can prove the
unicity, modulo circuit semantics, of the produced normal form:

\begin{corollary}
    Given a derivation $\pi : \Gamma \vdash M : A$ and two runs from
    the initial configuration that terminate in the two final
    configurations with extended circuits $E_1$ and $E_2$, then
    $\interp{E_1} = \interp{E_2}$.
\end{corollary}

\section{On Efficient Compilation and Types}
\label{sect:types}

In this section we introduce a type system $C\qlambda$ that captures the terms for which synchronous compilation always terminates producing a quantum circuit of linear size.

Going back to the term $RUW$ at the end of Section \ref{sect:informal}, where
$$
R \eqdef \ite{N}{\lambda f,g,x. f(gx)}{\lambda f,g,x. g(fx)},
$$
it is clear that the \circuittkmz cannot compile this term synchronously: 
this term contains an if-then-else of type $A=(\qbit \multimap \qbit) \multimap (\qbit \multimap \qbit) \multimap (\qbit \multimap \qbit)$; as illustrated in Fig.~\ref{fig:deadlock1}, 
and discussed in Section 2, the three synchronization points constituted by the if-then-else and the two gates $U,W$ are engaged in a \emph{circular} dependency, and the tokens are correspondingly stuck.

Another example is the term $PQ$, where $P:(\qbit \multimap \qbit) \multimap (\qbit \multimap \qbit) $ and $Q: (\qbit \multimap \qbit)$ are the terms 
\begin{align*}
P&\eqdef \ite{N}{\lambda f. f}{\lambda f,x.H(fx)}\\ 
Q&\eqdef\ite{N}{U}{V},
\end{align*}
where $N=\meas(\qgate{H}~(\new\ \fc))$ and $U,V:\qbit\multimap \qbit$ are two arbitrary gates. Here the two if-then-else are engaged in a similar circular dependency, and so the synchronous rule cannot be applied.

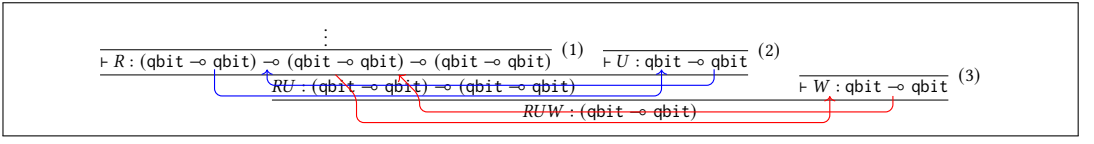
\begin{figure}
\fbox{
\begin{minipage}{\textwidth}
\adjustbox{scale=0.7, center}{
$
\begin{tikzpicture}

\node(proof) at (-1,0.3) {
\infer{RUW:(\qbit \multimap \qbit)}
	{
	\infer{RU:(\qbit \multimap \qbit) \multimap (\qbit \multimap \qbit)}
		{
		\infer[(1)]{\vdash R: (\qbit \multimap \qbit) \multimap (\qbit \multimap \qbit) \multimap (\qbit \multimap \qbit)}
			{\vdots
			}
		&\infer[(2)]{\vdash U:\qbit\multimap \qbit}{}
		}
	& \infer[(3)]{\vdash W:\qbit\multimap \qbit}{}
	}
};

\draw[<-, blue, rounded corners] (1.3,0.3) to (1.3,-0.2) to (-7.2,-0.2) to (-7.2,0.3);

\draw[->, blue, rounded corners] (2.3,0.3) to (2.3,0) to (-6.2,0) to (-6.2,0.3);

\draw[<-, red, rounded corners] (4.5,-0.2) to (4.5,-0.7) to (-4.5,-0.7) to (-4.5,-0.2) to (-4.9,0.2);

\draw[->, red, rounded corners] (5.7,-0.2) to (5.7,-0.5) to (-3.3,-0.5) to (-3.3,-0.2) to (-3.7,0.2);

\end{tikzpicture}
$
}
\end{minipage}
}
\caption{Deadlock in the term $RUW: \qbit\multimap\qbit$.
 }
\label{fig:deadlock1}
\end{figure}

To capture such cyclic dependencies, we associate any type derivation $\pi:\Gamma\vdash M:A$ with a directed graph $\mathscr G_\pi$ defined as follows. First, let us define the \emph{synchronization points} of $\pi$: these are either the gate axioms or the if-then-else rules, or any axiom of the form $\mathbb A_i\vdash \mathbb A_i$, for $i=1,\dots, p$, such
that $\pi$ contains an identity axiom of conclusion $x:B\vdash x:B$ and the list of atoms in $B$, read from left to right, is $\mathbb A_1,\dots,\mathbb A_p$.
Recall that tokens may traverse the identity axiom $x:B\vdash x:B$ in precisely $p$ ways, moving from the negatively occurring atom $\mathbb A_i$ to its positively occurring copy. For example, the identity axiom $x:\qbit_1\multimap \qbit_2\vdash x:\qbit_1\multimap\qbit_2$ yields the two synchronization points $\qbit_1\vdash \qbit_1$ and $\qbit_2\vdash \qbit_2$, corresponding to the two token paths through the axiom. 

Now, assign a distinct index $i\in \mathbb N$ to any {synchronization point} of $\pi$; these indices form the vertices of $\mathscr G_\pi$; an edge $i\to j$ exists whenever there exists a \emph{simple} token path (i.e.~a token path that only passes through structural rules) in $\pi$ going from a \emph{positive} position of the conclusion of $i$ to a \emph{negative} position of the conclusion of $j$. 

In the two examples considered we can easily see that the corresponding graphs contain a directed cycle relating the considered synchronization points. For instance, in the example from Fig.~\ref{fig:deadlock1} the graph $\mathscr G_\pi$ contains the cycles $1\to 2\to 1$ and $1\to 3\to 1$. 
By contrast, if we consider the terms $M_n$ from Section \ref{sect:informal}, one can check that the corresponding graphs are all acyclic, and indeed no deadlock actually occurs. 
In fact, the following holds:
\begin{proposition}\label{prop:graph1}
For any derivation $\pi:\Gamma\vdash M:A$, the synchronous compilation terminates on $\pi$ if and only if $\mathscr G_\pi$ is acyclic.
\end{proposition}

The graph $\mathscr G_\pi$ can be constructed directly and easily by induction on $\pi$.
We do this by introducing a \emph{colored} calculus $C\qlambda$.

\begin{definition}
\emph{Colored formulas} are defined by:
$$
A\eqdef  \baseTalt\mid\baseTalt^i \mid A\otimes A\mid A\multimap A,
$$
were $i\in \mathbb N$. For each colored formula $A$, let $|A|$ indicate the formula obtained by deleting labels.
For each uncolored formula $A$ containing the atoms $\mathbb A_1,\dots,\mathbb A_p$, and 
given indexes $i\in \mathbb N$ and $\vec i=i_1,\dots, i_p\in \mathbb N^p$, let $A^i$ (resp.~$A^{\vec i}$) be the colored formula in which all atoms have label $i$ (resp.~the colored formula with atoms $\mathbb A^{i_1}_1,\dots,\mathbb A^{i_p}_p$).
\end{definition}

We now show how type derivations can be labelled, inductively, with indices and a graph $\mathscr G$ describing their dependencies. For this we need the following construction:

\begin{definition}[match graph]
Given two colored formulas $A,B$, with $|A|=|B|$, their \emph{match graph}, noted $\mathscr M_{A,B}$, has the labels of $A$ and $B$ as vertices, and $i\to j\in \mathscr M_{A,B}$ if one of the following holds:
\begin{itemize}
\item for some negative atom $\mathbb A^i$ of $A$, the corresponding negative atom of $B$ is $\mathbb A^j$;
\item for some positive atom $\mathbb A^i$ of $B$, the corresponding positive atom of $A$ is $\mathbb A^j$.
\end{itemize}
\end{definition}

For example, take $A=(\qbit^1\multimap \qbit^2)\multimap \bit$ and $B=
(\qbit^2\multimap \qbit^3)\multimap \bit$, so that $|A|=|B|$. Their match graph has then vertices $1,2,3$ and edges $2\to 1$ (pos.~atom of $B$ to pos.~atom of $A$) and $2\to 3$ (neg.~atom of $A$ to neg.~atom of $B$).

The typing system $C\qlambda$ is generated from the following rules:

{\footnotesize
 \[ \infer{x : A^{\vec i} \vdash^{(\{i_1,\dots, i_p\},\emptyset)} x : A^{\vec i}}{} 
   \qquad
    \infer[(i)]{\vdash^{(\{i\},\emptyset)} c : A^i}{\mathcal T(c)=A}
 \qquad
   \infer{\Delta \vdash^{\mathscr G} \lambda x. M : A\multimap B}
  {\Delta, x : A\vdash^{\mathscr G} M : B}
  \]
  \[
  \infer{\Delta, \Gamma \vdash^{\mathscr G_1\cup\mathscr G_2\cup\mathscr M_{A,A'}} M~N : B}
  {\Delta \vdash^{\mathscr G_1} M : A\multimap B\qquad\Gamma\vdash^{\mathscr G_2} N : A'\qquad |A|=|A'|}
\]
\[  
   \infer{\vdash^{(\emptyset,\emptyset)} \starc  :1 }{}
  \qquad
 \infer{\Delta, \Gamma \vdash^{\mathscr G_1\cup\mathscr G_2} \letstar{M}{N} :A}
  {\Delta \vdash^{\mathscr G_1} M : 1\qquad\Gamma\vdash^{\mathscr G_2} N : A}
\]
\[
  \infer[(i)]{\Delta, |\Gamma|^i \vdash^{\mathscr G_1\cup\mathscr G_2\cup\mathscr G_3\cup\{ i\}} \ite{M}{N}{P} : |A|^i}
  {\Delta \vdash^{\mathscr G_1} M : \bit^j \qquad \Gamma^{\mathscr G_2} \vdash N : A \qquad \Gamma' \vdash^{\mathscr G_3} P : A'\qquad |\Gamma|=|\Gamma'|,|A|=|A'|}
    \]
  }

Each rule is a colored version of a corresponding rule of \qlambda: formulas are colored, and the judgments are labelled with a directed graph $\mathscr G$.
Gate and if-then-else rules are marked with a (unique) index, while identity axioms are marked with $p$ pairwise distinct indices; these indexes are used to inductively construct the graph $\mathscr G$. 
The application and let rules add the corresponding match graphs to $\mathscr G$.

For a colored derivation $\pi:\Gamma\vdash^{\mathscr G}M:A$, let $|\pi|:|\Gamma|\vdash M:|A|$ be its uncoloring, i.e.~the corresponding $\qlambda$-derivation obtained by deleting all colors and graphs. 
By inspecting, by induction on a derivation, all possible token paths we obtain the following:

\begin{theorem}\label{thm:graph2}
For any $C\qlambda$-derivation $\pi:\Gamma\vdash^{\mathscr G}M:A$, $\mathscr G\simeq\mathscr{G}_{|\pi|}$. As a consequence, $|\pi|$ can be compiled through the synchronous rule if and only if the produced graph $\mathscr G$ is acyclic.
\end{theorem}

\section{Beyond Linear Typing}\label{sect:discussion}

In this section, we discuss potential extensions of the compilation scheme we introduced and proved correct in
this paper.

On the one hand, it should certainly be noted that the type system
considered in this work is purely linear, not allowing for any form of
duplication. Furthermore, the types are multiplicative, namely
$\otimes$ and $\multimap$, and additives are present in disguise only
through the type $\bit$. Considering a more general type system with
additive and exponential connectives can be done, but requires great
care. On the one hand, it is in fact well known that Girard's Geometry
of Interaction is robust enough to allow for a faithful interpretation
of additives and exponentials~\cite{Girard1995,Laurent2001}, even in the presence of
multiple tokens~\cite{DalLago2017}. On the other hand, the way
conditionals are managed here and the special role they have means
that the typing of both branches of any conditional needs to be
restricted. Any such type should uniquely determine the (finite)
\emph{number} and the \emph{shape} of the tokens entering and exiting
the conditional. This of course cannot hold in presence of additives,
and requires exponentials to be \emph{bounded}, in the style of graded
comonads~\cite{GSS1992}. Similarly, a feature that can be added to our
type system without too much trouble is a form of \emph{structural
recursion} or \emph{iteration}. However, this would require the
introduction of an inductive type, such as that of natural numbers.
The question then becomes the following: should such a type be treated
similarly to $\bit$, that is, does it indicate the value that can
travel on a wire of the constructed circuit, or must its value be
known at compile time? In the second case, which we think corresponds
to the use of natural numbers in quantum algorithms, it is clear that
such an inductive type should itself not occur in the type of the
branches of a conditional construct. Summing up, while for the sake of
simplicity we have considered a very simple type system here, the approach can be adapted to more expressive
type systems. The only proviso is that of taking good care of how
conditionals can be typed, in particular guaranteeing finiteness and
determinacy.

\section{Related Work}
\label{sect:relatedwork}

The literature offers an extensive body of work on quantum
compilation, including optimization and so-called transpilation
techniques. Techniques specifically tailored to the compilation of
higher-order functional languages to quantum circuits are much
sparser. One notable contribution in this direction is the language
Qunity~\cite{voichick2023qunity}, which offers a higher-order quantum
programming language with classical control together with a full
compilation scheme towards
OpenQASM~\cite{cross2017openquantumassemblylanguage}. However, Qunity
does not feature a rewriting system and hence it cannot be executed.
This greatly limits the higher-level reasoning that can be done on this
language.

Girard's Geometry of Interaction~\cite{goi0,goi1} has already been
applied to quantum computing \cite{HasuoHoshino,DalLago2017}. In all
the aforementioned work, however, the underlying machine follows the
QRAM model, i.e., the token(s) perform some classical computation
while moving around the program, from time to time interacting with a
quantum register. The token trajectories are \emph{probabilistic}, and
provide an alternative way to fully execute the term on a quantum
input. By contrast, our approach is fully deterministic, and, instead
of executing the term, it produces a (yet to be executed) quantum
circuit, thus anticipating the quantum work.

Quantum $\lambda$-calculi come in many
flavours~\cite{selinger2009quantum,selinger2008fully,AG2005,DLMZ2009},
and semantics frameworks modelling them can themselves be built
following distinct lines~\cite{selinger2008fully,PaganiApplying,ClairambaultFullAbs}. 
Many programming languages following the QRAM model  have been introduced
since the mid nineties.
More recently, a somewhat different family of programming languages
\cite{lemonnier2024semanticseffectscentralityquantum,linearLambda,Arrighi_2017,chardonnet:tel-03959403}
have been designed focusing on what happens \emph{inside} the quantum
memory, this way allowing the user to write custom-made quantum operations
through so-called \emph{quantum conditionals}. This line of work is however 
outside the scope of this paper.
One should also mention the extensive body of literature on quantum circuit description languages~\cite{qwire,quipper}. The kind of challenges we faced here when compiling
conditionals are in fact similar to those encountered when endowing
Quipper with so-called dynamic lifting~\cite{dongho21,ColledanDalLago}.

\section{Conclusion}

We introduced a compilation scheme turning any term in a linear quantum
$\lambda$-calculus into an equivalent quantum circuit. Beyond proving the correctness of the translation, we have defined a type system which characterizes the
terms that can be compiled efficiently. Topics for future work include the
generalization of this compilation procedure to more expressive
calculi and its implementation in concrete programming languages, e.g., Linear
Haskell.
\bibliographystyle{ACM-Reference-Format}
\bibliography{biblio.bib}

\appendix
\appendixpage
\DoToC

\section{Additional Details of $\qlambda$}
\label{app:lambda}

In this section we give the formal results of $\qlambda$. As $\qlambda$ is a
sub-language of the original calculus by Selinger \& Valiron. The proofs can be
found in the original work~\cite{selinger2009quantum, 10.1007/11417170_26}.

\begin{definition}[Typed Quantum Closure]
    \label{def:typed-quantum-closure}
  We say that a quantum closure $[Q, L, M]$ is \emph{of type $A$ under
    context $\Gamma$}, denoted $\Gamma \vdash [Q, L, M] : A$ if
  $\Gamma, x_1 : \qbit, \dots, x_n : \qbit \vdash M : A$, where
  $x_1, \dots, x_n$ are in $L$.
\end{definition}

\begin{definition}[Value State]
  A quantum closure $[Q, L, M]$ is a \emph{value state} if $M$ is a
  value.
\end{definition}

\begin{lemma}[Substitution]
  \label{lemma:substitution}
  If $V$ is a value such that $\Delta \vdash V : A$ and $\Gamma, x : A
  \vdash M : B$ then $\Gamma, \Delta \vdash M[x\leftarrow V] : B$
\end{lemma}

\begin{proposition}[Subject Reduction]
  \label{prop:subject-reduction}
  Given a typed quantum closure $\Gamma \vdash [Q, L, M] : A$ and a reduction
  $[Q, L, M] \to \mu$ for a 
  distribution $\mu = \{[Q_{1}, L_{1}, M_{1}]^{\mu_{1}}, \ldots, [Q_{n}, L_{n},
  M_{n}]^{\mu_{n}}\}$, then
  $\Gamma \vdash [Q_{i},L_{i},M_{i}] : A$ for all $i \in \{1,\ldots,n\}$.
\end{proposition}

\begin{lemma}
\label{lemma:types-values}
A well-typed value $V$ is either a constant, a variable or one of the
following cases occurs: if it is of type $A\otimes B$ then the value is
of the shape $\pv{V_1}{V_2}$ where $V_1$ and $V_2$ are values of types
$A$ and $B$ respectively, if it is of type $\bit$ then it is either
$\tc $ or $\fc$, and if it is of type $A\multimap B$ it is a
$\lambda$-abstraction.
\end{lemma}

\begin{proposition}[Progress]
  Given a well-typed quantum closure $\vdash [Q, L, M] : A$, either $M$ is a
  value or there is $\mu \in \distr{\qclosures}$ such that $[Q, L, M] \to \mu$.
\end{proposition}

We inductively define binary relations
$\to_{n} \subseteq \qclosures \times \distr{\qclosures}$ as
follows: we first define $\to_{1}$ by $[Q,L,M] \to_{1} \mu$ if and only
if $[Q,L,M] \to \mu$ or $M$ is a value $V$ and $\mu = \{[Q,L,V]^{1}\}$;
and we define $\to_{n+1}$ by: $[Q,L,M] \to_{n+1} \mu$ if and only if
there are distributions
$\nu_{1},\ldots,\nu_{m} \in \distr{\qclosures}$ and
$\mu_{1},\ldots,\mu_{m} \in \mathbb{R}_{[0,1]}$ such that
$\mu = \sum_{1 \leq i \leq m} \mu_{i} \nu_{i}$ and
\begin{align*}
  [Q,L,M]
  &\to \{[Q_{1},L_{1},M_{1}]^{\mu_{1}},
    \ldots,[Q_{k},L_{k},M_{k}]^{\mu_{k}}\} \\
  [Q_{i},L_{i},M_{i}] &\to_{n} \nu_{i}
\end{align*}
for all $i \in \{1,\ldots,k\}$. 

\begin{theorem}[Termination]
  For any typed quantum closure $\vdash [Q,L,M] : A$, there is a
  unique distribution $\mu \in \distr{\qclosures}$ consisting of
  value states such that $[Q,L,M] \to_{n} \mu$ for any sufficiently
  large $n \in \mathbb{N}$.
\end{theorem}
\section{On Swapping Unitary Gates}
\label{asec:swapp-unit-gates}

Let $S$ be the following $\lambda$-term:
\begin{equation*}
  x:\bit,
  f:A,
  g:A
  \vdash
  \ite{x}{f \circ g}{g \circ f}
  : A
\end{equation*}
where
$A = (\qbit \otimes \qbit) \multimap (\qbit
\otimes \qbit)$ and $f \circ g$ is an abbreviation
of $\lambda x.\ f \, (g \, x)$. It is reasonable to expect
any compositional compilation algorithm associates
$S$ with a quantum circuit with holes
\begin{equation*}
  l:\bit, l_{1}:\qbit, l_{2}:\qbit
  \rhd C[-]\{-\} \rhd l_{1}:\qbit, l_{2}:\qbit
\end{equation*}
such that for any unitary gates
$U$ and $V$ acting on $2$ qubits,
\begin{equation*}
  \interp{C[U]\{V\}} =
  \interp{\ite{l}{U;V}{V;U}}
\end{equation*}
where $C[U]\{V\}$ is the quantum circuit obtained
by substituting $[-]$ with $U^{}$ and $\{-\}$ with
$V$. The goal of this section is to show that
either $[-]$ or $\{-\}$ must appear more than once
in $C[-]\{-\}$.

Towards this goal, let us formally introduce
quantum circuits with holes.
We define quantum circuits \emph{with holes} by
the following BNF:
\begin{equation*}
  \mathscr{C},\mathscr{D},\mathscr{E}
  ::=
  [-] \mid
  \{-\} \mid
  c_{\overline{r}}^{\overline{l}}
  \mid \mathscr{C} ; \mathscr{D}
  \mid \ite{\mathscr{C}}{\mathscr{D}}{\mathscr{E}}
\end{equation*}
where $c$ varies over $\mathbb{Q}$. Let
$\mathscr{C}$ be a quantum circuits with holes.
For a quantum circuit $C$ and $D$, we denote
$\mathscr{C}(C,D)$ for the quantum circuit
obtained by replacing $[-]$ and $\{-\}$ in
$\mathscr{C}$ with $C$ and $D$ respectively.
Throughout this section, we denote
$\{l_{1} : \qbit, l_{2} : \qbit\}$ by $\Theta$ and
denote $(l_{1},l_{2})$ by $\theta$. For
environments $\Gamma$ and $\Delta$, we write
$\Gamma \rhd \mathscr{C} \rhd \Delta$ when we can
derive the judgment from the rules in
Figure~\ref{fig:qcircrules} and the following axioms:
\begin{equation*}
  \infer{
    \Theta \rhd [-] \rhd \Theta
  }{},
  \qquad
  \infer{
    \Theta \rhd \{-\} \rhd \Theta
  }{}.
\end{equation*}
It is easy to see that if
$\Gamma \rhd \mathscr{C} \rhd \Delta$, then for any quantum
circuits $\Theta \rhd C \rhd \Theta$ and
$\Theta \rhd D \rhd \Theta$, we have
$\Gamma \rhd \mathscr{C}(C,D) \rhd \Delta$.

Given quantum circuits with holes
$\Gamma \rhd \mathscr{C} \rhd \Delta$ and
$\Gamma \rhd \mathscr{D} \rhd \Delta$, we write
\begin{equation*}
  \mathscr{C} \approx \mathscr{D}
\end{equation*}
when for any unitary gates $U$ and $V$ acting on 2
qubits, we have
\begin{equation*}
  \interp{\mathscr{C}(U^{\theta}_{\theta},V^{\theta}_{\theta})}
  =
  \interp{\mathscr{D}(U^{\theta}_{\theta},V^{\theta}_{\theta})}.
\end{equation*}
The following proposition is the goal of this
section.
\begin{proposition}\label{aprop:gateswap}
  Let $\mathscr{C}$ be a quantum circuit with
  holes such that
  \begin{equation*}
    l : \bit, \Theta
    \rhd \mathscr{C}(X,Y) \rhd \Theta,
  \end{equation*}
  and we define a quantum circuit with holes
  $\mathscr{D}$ to be
  \begin{equation*}
    \ite{l}{
      [-] ; \{-\}
    }{
      \{-\} ; [-]
    }.
  \end{equation*}
  If $\mathscr{C} \approx \mathscr{D}$, then
  either $[-]$ or $\{-\}$ appears more than once in
  $\mathscr{C}$.
\end{proposition}
\begin{proof}
  We suppose that both $[-]$ and $\{-\}$ appear
  exactly once in $\mathscr{C}$ and derive a
  contradiction. By applying if-then-else removal
  to $\mathscr{C}$, we obtain an if-then-else free
  quantum circuit with holes $\mathscr{E}$ such
  that $[-]$ and $\{-\}$ appear exactly once and
  $\mathscr{C} \approx \mathscr{E}$. We only give
  a proof for the case when $[-]$ appears left-hand side of $\{-\}$ in $\mathscr{E}$, i.e.,
  there is a subcircuit
  $\mathscr{E}_{0};\mathscr{E}_{1}$ of
  $\mathscr{E}$ such that $[-]$ appears in
  $\mathscr{E}_{0}$ and $\{-\}$ appears in
  $\mathscr{E}_{1}$. We can prove the other case
  in the same manner. Let us prepare some
  notations. We define quantum circuits
  $l:\bit \rhd \bra{00} \rhd l:\bit, \Theta$ and
  $l:\bit \rhd \mathscr{F} \rhd \Theta$ and
  $l:\bit \rhd \mathscr{G} \rhd \Theta$ by
  \begin{align*}
    \bra{00}
    &=
      0_{l_{0}} ;
      0_{l_{1}} ;
      \new_{l_{0}}^{l_{0}} ;
      \new_{l_{1}}^{l_{1}}, \\
    \mathscr{F}
    &= \bra{00} ;
      \mathscr{D}, \\
    \mathscr{G}
    &= \bra{00} ;
      \mathscr{E},
  \end{align*}
  We write $\mathbb{S}$ for the unit circle
  \begin{equation*}
    \{(z,w) \in \mathbb{C} \times \mathbb{C} \mid
    |z|^{2} + |w|^{2} = 1\}.
  \end{equation*}
  For $(z,w) \in \mathbb{S}$, we define a
  $4 \times 4$ unitary matrix $X_{z,w}$ by
  \begin{equation*}
    X_{z,w} =
    \begin{pmatrix}
      z & -w & 0 & 0 \\
      w^{\ast} & z^{\ast} & 0 & 0 \\
      0 & 0 & z & -w \\
      0 & 0 & w^{\ast} & z^{\ast} \\
    \end{pmatrix}.
  \end{equation*}
  For a $4 \times 4$ unitary matrix $U$, we simply
  write $\Theta \rhd U \rhd \Theta$ for the
  corresponding quantum circuit. For
  $n \in \mathbb{N}$, a $4 \times 4$ unitary
  matrix $U$ and a $4n \times 4n$ unitary matrix
  \begin{equation*}
    V =
    \begin{pmatrix}
      V_{1,1} & \cdots & V_{1,n} \\
      \vdots & \ddots & \vdots \\
      V_{n,1} & \cdots & V_{n,n} \\
    \end{pmatrix}
  \end{equation*}
  consisting of $4 \times 4$ submatrices
  $V_{i,j}$, we define a $4n \times 4n$ matrix
  $U \bullet V$ by
  \begin{equation*}
    U \bullet V =
    \begin{pmatrix}
      UV_{1,1}U^{\dagger}
      & \cdots
      & UV_{1,n}U^{\dagger} \\
      \vdots & \ddots & \vdots \\
      UV_{n,1}U^{\dagger}
      & \cdots
      & UV_{n,n}U^{\dagger} \\
    \end{pmatrix}.
  \end{equation*}
  We now proceed to prove the claim. Since
  $\mathscr{E}$ has no if-then-else constructor,
  and $[-]$ appears left-hand side of $\{-\}$ in
  $\mathscr{E}$, there exist a natural number
  $n \in \mathbb{N}$, a positive matrix
  $P \in \mathbb{C}^{4n \times 4n}$ and a family
  of trace-preserving completely positive maps
  \begin{equation*}
    \left\{
      f_{z,w} \colon
      \mathbb{C}^{4n \times 4n}
      \to
      \mathbb{C}^{4 \times 4}
    \right\}_{(z,w) \in \mathbb{S}}
  \end{equation*}
  such that for any $4 \times 4$ unitary matrix
  $U$ and any $(z,w) \in \mathbb{S}$,
  \begin{equation}\label{eq:E}
    \interp{1_{l:\bit};\mathscr{G}
      (U, X_{z,w})} =
    f_{z,w} (U \bullet P).
  \end{equation}
  This can be checked by induction on the
  structure of $\mathscr{E}$. Since
  $\mathscr{G} \approx \mathscr{F}$, it follows
  from \eqref{eq:E} that for any $4 \times 4$
  unitary matrix $U$ and any
  $(z,w) \in \mathbb{S}$, we have
  \begin{equation} \label{eq:E=D} %
    f_{z,w}(U \bullet P) = U
    \begin{pmatrix}
      |z|^{2} & zw & 0 & 0 \\
      z^{\ast} w^{\ast} & |w|^{2} & 0 & 0 \\
      0 & 0 & 0 & 0 \\
      0 & 0 & 0 & 0 \\
    \end{pmatrix}
    U^{\dagger}.
  \end{equation}
  Let us divide $P$ into the following $4 \times 4$
  submatrices $P_{i,j}$:
  \begin{equation*}
    \begin{pmatrix}
      P_{1,1} & \cdots & P_{1,n} \\
      \vdots & \ddots & \vdots \\
      P_{n,1} & \cdots & P_{n,n} \\
    \end{pmatrix},
  \end{equation*}
  and we denote the $(k,l)$th entry of $P_{i,j}$
  by $a_{i,j,k,l}$:
  \begin{equation*}
    P_{i,j} =
    \begin{pmatrix}
      a_{i,j,1,1}
      & a_{i,j,1,2}
      & a_{i,j,1,3}
      & a_{i,j,1,4} \\
      a_{i,j,2,1}
      & a_{i,j,2,2}
      & a_{i,j,2,3}
      & a_{i,j,2,4} \\
      a_{i,j,3,1}
      & a_{i,j,3,2}
      & a_{i,j,3,3}
      & a_{i,j,3,4} \\
      a_{i,j,4,1}
      & a_{i,j,4,2}
      & a_{i,j,4,3}
      & a_{i,j,4,4} \\
    \end{pmatrix}.
  \end{equation*}
  We define $4n \times 4n$ positive matrices $K$
  and $L$ by
  \begin{align*}
    K
    &=
      \begin{pmatrix}
        K_{1,1} & \cdots & K_{1,n} \\
        \vdots & \ddots & \vdots \\
        K_{n,1} & \cdots & K_{n,n} \\
      \end{pmatrix},
    &
      L
    &=
      \begin{pmatrix}
        L_{1,1} & \cdots & L_{1,n} \\
        \vdots & \ddots & \vdots \\
        L_{n,1} & \cdots & L_{n,n} \\
      \end{pmatrix}
  \end{align*}
  where $K_{i,j}$ and $L_{i,j}$ are given by
  \begin{align*}
    K_{i,j}
    &=
      \begin{pmatrix}
        a_{i,j,1,1}
        & a_{i,j,1,2}
        & a_{i,j,1,3}
        & 0 \\
        a_{i,j,2,1}
        & a_{i,j,2,2}
        & a_{i,j,2,3}
        & 0 \\
        a_{i,j,3,1}
        & a_{i,j,3,2}
        & a_{i,j,3,3}
        & 0 \\
        0
        & 0
        & 0
        & 0 \\
      \end{pmatrix},
    &
      L_{i,j}
    &=
      \begin{pmatrix}
        0
        & 0
        & 0
        & 0 \\
        0
        & 0
        & 0
        & 0 \\
        0
        & 0
        & 0
        & 0 \\
        0
        & 0
        & 0
        & a_{i,j,4,4} \\
      \end{pmatrix}.
  \end{align*}
  Since \eqref{eq:E=D} holds for any $U$
  and $(z,w) \in \mathbb{S}$, by replacing
  $U$ in \eqref{eq:E=D} with
  \begin{equation*}
    U
    \begin{pmatrix}
      1 & & & \\
        & 1 & & \\
        & & 1 & \\
        & & & -1 \\
    \end{pmatrix},
  \end{equation*}
  we see that for any $4 \times 4$ matrix $U$ and
  $(z,w) \in \mathbb{S}$, we also have
  \begin{equation} \label{eq:+++-}%
    f_{z,w}\left( U \bullet \left(
        \begin{pmatrix}
          1 & & & \\
            & 1 & & \\
            & & 1 & \\
            & & & -1 \\
        \end{pmatrix}
        \bullet P\right)
    \right)
    = U
    \begin{pmatrix}
      |z|^{2} & zw & 0 & 0 \\
      z^{\ast}w^{\ast} & |w|^{2} & 0 & 0 \\
      0 & 0 & 0 & 0 \\
      0 & 0 & 0 & 0 \\
    \end{pmatrix}
    U^{\dagger}.
  \end{equation}
  By adding each side of \eqref{eq:E=D} and
  \eqref{eq:+++-}, we see that for any
  $4 \times 4$ unitary matrix $U$ and any
  $(z,w) \in \mathbb{S}$,
  \begin{equation} \label{eq:fuq} %
    f_{z,w}(U \bullet (K + L)) = U
    \begin{pmatrix}
      |z|^{2} & zw & 0 & 0 \\
      z^{\ast}w^{\ast} & |w|^{2} & 0 & 0 \\
      0 & 0 & 0 & 0 \\
      0 & 0 & 0 & 0 \\
    \end{pmatrix}
    U^{\dagger}.
  \end{equation}
  By repeating similar arguments, we can show that
  for any $4 \times 4$ unitary matrix $U$ and any
  $(z,w) \in \mathbb{S}$, we have
  \begin{equation} \label{eq:QRST} %
    f_{z,w}(U \bullet (Q + R + S + T)) = U
    \begin{pmatrix}
      |z|^{2} & 0 & 0 & 0 \\
      0 & |w|^{2} & 0 & 0 \\
      0 & 0 & 0 & 0 \\
      0 & 0 & 0 & 0 \\
    \end{pmatrix}
    U^{\dagger}.
  \end{equation}
  Where the $4n \times 4n$ positive matrices $Q$, $R$, $S$ and $T$ are given by
  \begin{align*}
    Q
    &=
      \begin{pmatrix}
        Q_{1,1} & \cdots & Q_{1,n} \\
        \vdots & \ddots & \vdots \\
        Q_{n,1} & \cdots & Q_{n,n} \\
      \end{pmatrix},
    &
      R
    &=
      \begin{pmatrix}
        R_{1,1} & \cdots & R_{1,n} \\
        \vdots & \ddots & \vdots \\
        R_{n,1} & \cdots & R_{n,n} \\
      \end{pmatrix},
    \\
    S
    &=
      \begin{pmatrix}
        S_{1,1} & \cdots & S_{1,n} \\
        \vdots & \ddots & \vdots \\
        S_{n,1} & \cdots & S_{n,n} \\
      \end{pmatrix},
    &
      T
    &=
      \begin{pmatrix}
        T_{1,1} & \cdots & T_{1,n} \\
        \vdots & \ddots & \vdots \\
        T_{n,1} & \cdots & T_{n,n} \\
      \end{pmatrix}
  \end{align*}
  consisting of $4 \times 4$ submatrices
  \begin{align*}
    Q_{i,j}
    &=
      \begin{pmatrix}
        a_{i,j,1,1}
        & 0
        & 0
        & 0 \\
        0
        & 0
        & 0
        & 0 \\
        0
        & 0
        & 0
        & 0 \\
        0
        & 0
        & 0
        & 0 \\
      \end{pmatrix},
    &
      R_{i,j}
    &=
      \begin{pmatrix}
        0
        & 0
        & 0
        & 0 \\
        0
        & a_{i,j,2,2}
        & 0
        & 0 \\
        0
        & 0
        & 0
        & 0 \\
        0
        & 0
        & 0
        & 0 \\
      \end{pmatrix},
    \\
    S_{i,j}
    &=
      \begin{pmatrix}
        0
        & 0
        & 0
        & 0 \\
        0
        & 0
        & 0
        & 0 \\
        0
        & 0
        & a_{i,j,3,3}
        & 0 \\
        0
        & 0
        & 0
        & 0 \\
      \end{pmatrix},
    &
      T_{i,j}
    &=
      \begin{pmatrix}
        0
        & 0
        & 0
        & 0 \\
        0
        & 0
        & 0
        & 0 \\
        0
        & 0
        & 0
        & 0 \\
        0
        & 0
        & 0
        & a_{i,j,4,4} \\
      \end{pmatrix}.
  \end{align*}
  When $U$ in \eqref{eq:QRST} is the identity
  matrix, we have
  \begin{equation} \label{eq:id} %
    f_{z,w}(Q) + f_{z,w}(R) + f_{z,w}(S) +
    f_{z,w}(T) =
    \begin{pmatrix}
      |z|^{2} & 0 & 0 & 0 \\
      0 & |w|^{2} & 0 & 0 \\
      0 & 0 & 0 & 0 \\
      0 & 0 & 0 & 0 \\
    \end{pmatrix}.
  \end{equation}
  We define
  \begin{equation*}
    (q_{1},q_{2},q_{3},q_{4}),
    \quad (r_{1},r_{2},r_{3},r_{4}),
    \quad (s_{1},s_{2},s_{3},s_{4}),
    \quad (t_{1},t_{2},t_{3},t_{4})
  \end{equation*}
  to be the main diagonals of $f_{z,w}(Q)$,
  $f_{z,w}(R)$, $f_{z,w}(S)$ and $f_{z,w}(T)$
  respectively. Since $f_{z,w}(Q)$, $f_{z,w}(R)$,
  $f_{z,w}(S)$ and $f_{z,w}(T)$ are positive, all
  elements in the main diagonals are non-negative
  real numbers. By \eqref{eq:id}, we have
  \begin{align}
    q_{1} + r_{1} + s_{1} + t_{1} &= |z|^{2},
                                    \label{eq:qrst1} \\
    q_{2} + r_{2} + s_{2} + t_{2} &= |w|^{2},
                                    \label{eq:qrst2} \\
    q_{3} + r_{3} + s_{3} + t_{3} &= 0,
                                    \label{eq:qrst3} \\
    q_{4} + r_{4} + s_{4} + t_{4} &= 0.
                                    \label{eq:qrst4}
  \end{align}
  From \eqref{eq:qrst3} and \eqref{eq:qrst4}, we
  obtain
  \begin{equation*}
    q_{3} = r_{3} = s_{3} = t_{3} =
    q_{4} = r_{4} = s_{4} = t_{4} = 0.
  \end{equation*}
  By replacing $U$ in \eqref{eq:QRST} with
  \begin{equation*}
    \begin{pmatrix}
      1 & 0 & 0 & 0 \\
      0 & 0 & 1 & 0 \\
      0 & 1 & 0 & 0 \\
      0 & 0 & 0 & 1 \\
    \end{pmatrix},
  \end{equation*}
  we obtain
  \begin{equation*}
    f_{z,w}(Q) + f_{z,w}(U \bullet R) + f_{z,w}(U
    \bullet S) + f_{z,w}(T) =
    \begin{pmatrix}
      |z|^{2} & 0 & 0 & 0 \\
      0 & 0 & 0 & 0 \\
      0 & 0 & |w|^{2} & 0 \\
      0 & 0 & 0 & 0 \\
    \end{pmatrix}.
  \end{equation*}
  It follows from positivity of $f_{z,w}(Q)$,
  $f_{z,w}(U \bullet R)$, $f_{z,w}(U \bullet S)$
  and $f_{z,w}(T)$ that we have
  \begin{equation*}
    q_{2} = t_{2} = 0.
  \end{equation*}
  In the same manner, we can prove
  \begin{equation*}
    r_{1} = s_{1} = s_{2} = t_{1} = t_{2} = 0.
  \end{equation*}
  Therefore, by \eqref{eq:qrst1} and \eqref{eq:qrst2},
  we obtain
  \begin{equation*}
    q_{1} = |z|^{2},
    \qquad
    r_{2} = |w|^{2}.
  \end{equation*}
  Since $f_{z,w}$ preserves traces,
  \begin{equation*}
    \mathrm{tr}(Q) = |z|^{2},
    \qquad
    \mathrm{tr}(R) = |w|^{2}.
  \end{equation*}
  These equalities hold for any
  $(z,w) \in \mathbb{S}$, a contradiction.
\end{proof}

\section{Proof of Theorem~\ref{thm:if-then-else-removal}}
\label{asec:proof-if-then-else-removal}
\begin{theorem}
  \label{athm:if-then-else-removal}
  For any
  circuit $\Gamma\rhd C\rhd \Delta$ of size $n$
  there exists a plain circuit
  $\Gamma\rhd D\rhd \Delta$ of size
  $\mathcal O(n)$ such that
  $\interp{C}=\interp{D}$.
\end{theorem}
\begin{proof}
  We prove the claim by induction on the structure
  of $C$. We first prepare a few circuits. Let
  $\boolT$ be a base type, and let
  $L = (l_{1},l_{2},l_{3})$ be a triple of labels.
  For an environment $\Delta$ such that
  $\{l_{1}:\boolT,l_{2}:\boolT,l_{3}:\bit\}
  \subseteq \Delta$, we define a circuit
  \begin{equation*}
    \Delta
    \rhd \sigma_{L,\boolT}
    \rhd
    \Delta
  \end{equation*}
  as follows:
  \begin{align*}
    \sigma_{L,\qbit}
    &=
      \new^{l_{3}}_{l_{3}};
      \mathtt{CSWAP}^{L}_{L};
      \meas^{l_{3}}_{l_{3}},
    \\
    \sigma_{L,\bit}
    &=
      \new^{l_{1}}_{l_{1}};
      \new^{l_{2}}_{l_{2}};
      \sigma_{L,\qbit};
      \meas^{l_{1}}_{l_{1}};
      \meas^{l_{2}}_{l_{2}}
  \end{align*}
  where $\mathtt{CSWAP}$ is the controlled-swap gate (also called the
  \emph{Fredkin gate}) given by
  \begin{equation*}
    \mathtt{CSWAP} \ket{ijk} =
    \begin{cases}
      \ket{jik}, & \text{if $k = 1$}, \\
      \ket{ijk}, & \text{if $k = 0$}.
    \end{cases}
  \end{equation*}
  We note that we can construct $\mathtt{CSWAP}$
  using the Toffoli gate, and the Toffoli gate is
  definable in Clifford$+T$
  \cite{Nielsen_Chuang_2010}. To illustrate the
  constructions of the plain circuits,
  we adopt a graphical presentation of quantum
  circuits in an informal way. We illustrate the
  circuit $\sigma_{L,\boolT}$ as follows:
  \[
    \begin{tikzpicture}[baseline=(current bounding box.south)]
  \node[draw, rectangle, minimum width=0.75cm, minimum height=1cm]
  at (0,0) (s) {$\sigma_{L,\boolT}$};
  \draw (s.west) -- ++(-0.5,0) node[left] {$\boolT$};
  \draw ([yshift=0.4cm]s.west) -- ++(-0.5,0) node[left] {$\boolT$};
  \draw ([yshift=-0.4cm]s.west) -- ++(-0.5,0) node[left] {$\bit$};
  \draw (s.east) -- ++(0.5,0) node[right] {$\boolT$};
  \draw ([yshift=0.4cm]s.east) -- ++(0.5,0) node[right] {$\boolT$};
  \draw ([yshift=-0.4cm]s.east) -- ++(0.5,0) node[right] {$\bit$};
\end{tikzpicture}.
  \]
  Here, wires on the left (right) hand side
  receive inputs (return outputs). It is not
  difficult to check that when $\boolT = \qbit$,
  \begin{align*}
    \interp{\Gamma \rhd
    \termOne_{l_{3}} ; \sigma_{L,\qbit} \rhd \Delta
    }
    &=
      \interp{
      \Gamma
      \rhd
      \termOne_{l_{3}} ;
      \mathtt{SWAP}^{l_{1},l_{2}}_{l_{1},l_{2}} \rhd \Delta},
    \\
    \interp{
    \Gamma
    \rhd
    \termZero_{l_{3}} ; \sigma_{L,\qbit}
    \rhd
    \Delta
    }
    &=
      \interp{
      \Gamma
      \rhd
      \termZero_{l_{3}}
      \rhd
      \Delta
      }
  \end{align*}
  where $\Gamma = \Delta \setminus \{l_{3}:\bit\}$ and
  $\mathtt{SWAP}$ is given by
  $\mathtt{SWAP} \ket{ij} = \ket{ji}$. These
  equations mean that $\sigma_{L,\qbit}$ and $\sigma_{L,\bit}$
  conditionally swap its inputs as illustrated as
  follows:
  \begin{equation*}
    \begin{gathered}
      \begin{tikzpicture}[baseline=(current bounding box.south)]
  \node[draw, rectangle, minimum width=0.75cm, minimum height=1cm]
  at (0,0) (s) {$\sigma_{L,\qbit}$};
  \node[draw, rectangle, minimum width=0.4cm, minimum height=0.4cm]
  at ([xshift=-0.5cm, yshift=-0.4cm]s.west) (0left) {$1$};
  \draw (s.west) -- ++(-1,0);
  \draw ([yshift=0.4cm]s.west) -- ++(-1,0);
  \draw (0left.east) -- ([yshift=-0.4cm]s.west);
  \draw (s.east) -- ++(0.5,0);
  \draw ([yshift=0.4cm]s.east) -- ++(0.5,0);
  \draw ([yshift=-0.4cm]s.east) -- ++(0.5,0);

  \draw ([xshift=1.25cm]s.east) to[out=0,in=180] ++(2,0.4);
  \draw ([xshift=1.25cm,yshift=0.4cm]s.east) to[out=0,in=180] ++(2,-0.4);
  \node[draw, rectangle, minimum width=0.4cm, minimum height=0.4cm]
  at ([xshift=2cm, yshift=-0.4cm]s.east) (0right) {$1$};
  \draw (0right.east) -- ([xshift=3.25cm,yshift=-0.4cm]s.east);
  \node at ([xshift=0.875cm]s.east) {$=$};

\end{tikzpicture} \\
      \begin{tikzpicture}[baseline=(current bounding box.south)]
  \node[draw, rectangle, minimum width=0.75cm, minimum height=1cm]
  at (0,0) (s) {$\sigma_{L,\qbit}$};
  \node[draw, rectangle, minimum width=0.4cm, minimum height=0.4cm]
  at ([xshift=-0.5cm, yshift=-0.4cm]s.west) (1swap) {$0$};
  \draw (s.west) -- ++(-1,0);
  \draw ([yshift=0.4cm]s.west) -- ++(-1,0);
  \draw (1swap.east) -- ([yshift=-0.4cm]s.west);
  \draw (s.east) -- ++(0.5,0);
  \draw ([yshift=0.4cm]s.east) -- ++(0.5,0);
  \draw ([yshift=-0.4cm]s.east) -- ++(0.5,0);

  \draw ([xshift=1.25cm]s.east) -- ++(2,0);
  \draw ([xshift=1.25cm,yshift=0.4cm]s.east) -- ++(2,0);
  \node[draw, rectangle, minimum width=0.4cm, minimum height=0.4cm]
  at ([xshift=2cm, yshift=-0.4cm]s.east) (1right) {$0$};
  \draw (1right.east) -- ([xshift=3.25cm,yshift=-0.4cm]s.east);
  \node at ([xshift=0.875cm]s.east) {$=$};
\end{tikzpicture}      
    \end{gathered}    
  \end{equation*}
  We generalize the conditional
  swap $\sigma_{\Gamma}$ to conditional swaps for
  environments. For environments $\Gamma_{1}$ and $\Gamma_{2}$,
  we construct a circuit
  \begin{equation*}
    \Gamma_{1},\Gamma_{2},l:\bit \rhd
    \tau_{\Gamma_{1},\Gamma_{2},l:\bit} \rhd
    \Gamma_{1},\Gamma_{2},l:\bit
  \end{equation*}
  by sequentially composing the circuits
  $\sigma_{L,\boolT}$. To be concrete, we define
  $\tau_{\Gamma_{1},\Gamma_{2},l:\bit}$ by
  induction on the size of $\Gamma_{i}$. When
  $\Gamma_{1}$ and $\Gamma_{2}$ are empty, we
  define $\tau_{\Gamma_{1},\Gamma_{2},l:\bit}$ to
  be the identity gate on $l:\bit$. When
  $\Gamma_{i}$ is a disjoint union
  $\{l_{i}:\boolT\} \cup \Theta_{i}$, we define
  $\tau_{\Gamma_{1},\Gamma_{2},l:\bit}$ to be the
  sequential composition of:
  \begin{itemize}
  \item $\sigma_{(l_{1},l_{2},l),\baseT}$ acting on the control bit
    (the last bit) and $\{l_{1} : \boolT\}$ in
    $\Gamma_{1}$ and $\{l_{2} : \boolT\}$ in the
    second $\Gamma_{2}$; and
  \item $\tau_{\Theta_{1},\Theta_{2},l:\bit}$
    acting on the control bit and the remaining part
    $\Theta_{1}$ and $\Theta_{2}$.
  \end{itemize}
  The behavior of the circuit
  $\tau_{\Gamma_{1},\Gamma_{2},l:\bit}$ is similar
  to that of $\sigma_{L,l:\bit}$. If the control
  bit is $1$, then
  $\tau_{\Gamma_{1},\Gamma_{2},l:\bit}$ swaps
  (q)bits in $\Gamma_{1}$ and (q)bits in
  $\Gamma_{2}$; and if the control bit is $0$, then
  $\tau_{\Gamma_{1},\Gamma_{2},l:\bit}$ does
  nothing. Let us start the inductive proof
  of the claim. We give the proof only for the
  non-trivial case
  \begin{equation*}
    \Gamma , l:\bit \rhd
    \ite{l}{D}{E}\rhd \Delta
  \end{equation*}
  with
  $\Gamma \rhd D \rhd \Delta$ and
  $\Gamma \rhd E \rhd \Delta$. For
  simplicity, we suppose that
  $\Gamma$ and $\Delta$ are
  of the following forms:
  \begin{equation*}
    \Gamma = \Delta = \{l_{1}:\qbit\}.
  \end{equation*}
  Under this assumption, we define a circuit
  \begin{math}
    \Gamma, l:\bit \rhd F \rhd \Delta
  \end{math}
  to be
  \begin{equation*}
    0_{l_{2}};
    \new^{l_{2}}_{l_{2}};
    \tau_{\Gamma,\Xi,l:\bit};
    D';
    E; \\
    \tau_{\Gamma,\Xi,l:\bit};
    \discard^{l};
    \meas^{l_{2}}_{l_{2}};
    \discard^{l_{2}}
  \end{equation*}
  where $\Xi = \{l_{2}:\qbit\}$ and
  $\Xi \rhd D' \rhd \Xi$ is a
  quantum gate obtained by renaming $l_{1}$ in $D$
  by $l_{2}$. The structure of this circuit
  can be illustrated as follows:
  \[
    \begin{tikzpicture}[baseline=(current bounding box.south)]
  \node[draw, rectangle, minimum width=0.5cm, minimum height=0.5cm] (0) {$0$};
  \node[draw, rectangle, minimum width=0.5cm, minimum height=0.5cm, right=0.5cm of 0] (new) {$\new$};
  \node[draw, rectangle, minimum width=0.5cm, minimum height=1.7cm, right=0.5cm of new, yshift=-0.6cm] (tau) {$\tau$};
  \node[draw, rectangle, minimum width=0.5cm, minimum height=0.5cm, right= 0.5cm of tau, yshift=0.6cm] (D) {$D'$};
  \node[draw, rectangle, minimum width=0.5cm, minimum height=0.5cm, right= 0.5cm of tau] (E) {$E$};
  \node[draw, rectangle, minimum width=0.5cm, minimum height=1.7cm, right=0.5cm of E] (tau2) {$\tau$};
  \node[draw, rectangle, minimum width=0.5cm, minimum height=0.5cm, right=0.5cm of tau2, yshift=0.6cm] (meas) {$\meas$};

  \draw (new.east) -- node[above]{$\scriptstyle \Xi$} (new.east-|tau.west);
  \draw (tau.west) -- (0.west |- tau.west) node[left] {$\Gamma$};
  \draw (0.east) -- (new.west);
  \draw ([yshift=-0.6cm]tau.west) -- node[above] {$\scriptstyle \bit$}
  ([yshift=-0.6cm]tau.west -| 0.west);
  \draw ([xshift=0.5cm]meas.east) -- ++(0,0.2);
  \draw ([xshift=0.5cm]meas.east) -- ++(0,-0.2);
  \draw ([xshift=0.55cm]meas.east) -- ++(0,0.15);
  \draw ([xshift=0.55cm]meas.east) -- ++(0,-0.15);
  \draw ([xshift=0.6cm]meas.east) -- ++(0,0.1);
  \draw ([xshift=0.6cm]meas.east) -- ++(0,-0.1);
  \draw (meas.east) -- ([xshift=0.5cm]meas.east);
  \draw (tau2.east) -- (tau2.east-|meas.east) -- ++(0.5,0)
  node[right] {$\Xi$};
  \draw ([xshift=0.5cm,yshift=-1.2cm]meas.east) -- ++(0,0.2);
  \draw ([xshift=0.5cm,yshift=-1.2cm]meas.east) -- ++(0,-0.2);
  \draw ([xshift=0.55cm,yshift=-1.2cm]meas.east) -- ++(0,0.15);
  \draw ([xshift=0.55cm,yshift=-1.2cm]meas.east) -- ++(0,-0.15);
  \draw ([xshift=0.6cm,yshift=-1.2cm]meas.east) -- ++(0,0.1);
  \draw ([xshift=0.6cm,yshift=-1.2cm]meas.east) -- ++(0,-0.1);
  \draw ([yshift=0.6cm]tau.east) -- (D.west);
  \draw (tau.east) -- (E.west);
  \draw ([yshift=0.6cm]tau2.west) -- (D.east);
  \draw (tau2.west) -- (E.east);
  \draw ([yshift=-0.6cm]tau.east) --
  node[below] {$\scriptstyle \bit$}
  ([yshift=-0.6cm]tau2.west);
  \draw ([yshift=0.6cm]tau2.east) -- (meas.west);
  \draw ([yshift=-0.6cm]tau2.east) --
  ([yshift=-0.6cm]tau2.east-|meas.east) --
  ++(0.5,0);
  ;

\end{tikzpicture}
  \]
  The gate $\myground$ indicates a discarded bit. Observe that, if we
  feed the bit $1$ to $l:\bit$, then the circuit above will compute
  $D$ and ``kill'' $E$, that is, feed it with $\ket 0$ and measure and
  discard its result; conversely, if we feed the bit $0$ to $l:\bit$,
  then the circuit above will similarly compute $E$ and ``kill'' $D$.
  In fact, we have $\interp{\ite{C}{D}{E}} = \interp{F}$. If we obtain
  a circuit $Y$ by applying this if-then-else elimination to a circuit
  $X$, then it is easy to see that the size of $Y$ is proportional to
  the sum of the size of $X$.
\end{proof}

\section{Bisimulation between the \circuittkmz and the \evaltkm}

In this section we establish a bisimulation between the machines \evaltkm and \circuittkmz. 
More precisely, we will establish a bisimulation relating a single execution of \circuittkmz,
producing a quantum circuit, with a family of probabilistic executions of
\evaltkm performed in parallel, giving rise to a distribution of quantum states.

This section is organized as follows: in Subsection \ref{subsec:slicing} we first discuss some properties of (extended) circuits and their semantics; then, in Subsection \ref{subsec:evaltkm} we provide a formal definition of the \evaltkm (by adapting the machine from~\cite{DalLago2017} to our setting); in Subsection \ref{subsec:evalpara} we introduce a notion of parallel execution for the \evaltkm; in Subsection \ref{subsec:bisi} we define and prove the bisimulation between parallel executions of the \evaltkm and executions of the \circuittkmz; finally, in Subsection \ref{subsec:termi} we deduce the termination and confluence of the \circuittkmz.

\subsection{Slicing (Extended) Circuits}\label{subsec:slicing}

\subsubsection{Basic Circuits}

Given a circuit $C$, we will suppose given an indexing $\alpha_1,\dots,\alpha_n\in \mathbb N$ of all $\meas$-gates occurring in $C$. The label $\alpha$ of a gate $\meas^l_{l'}$ should not be confused with the labels $l,l'$ of the gate's wire. 
We will use $\ell$ to refer to either a wire label or a $\meas$ gate index.

For any circuit $C$, let:
\begin{itemize}
\item $\LITE(C)$ the set of labels $l$ such that $\ite{l}{D_1}{D_2}$ occurs in $C$;
\item $\LMEAS(C)$ the set of indices $\alpha$ of $\meas$-gates occurring in $C$;
\item $\mathcal L(C)=\LITE(C)\cup\LMEAS(C)$.
\end{itemize}

We define the following sets of addresses:
\begin{definition}\label{def:circadd}
Let $C$ be a circuit. A set $L\subseteq \mathcal L(C)$ is called a \emph{path} if one of the following hold:
\begin{itemize}
\item $C=c\neq\meas$ and $L=\emptyset$;
\item $C=\meas$ has index $\alpha$ and $L=\{\alpha\}$;
\item $C=D_1;D_2$ and $L=L_1\cup L_2$, where $L_i$ is a path in $D_i$;
\item $C=\ite{l}{D_1}{D_2}$ and either $L=\{l\}\cup L_1$, with $L_1$ a path in $D_1$ or
$L=\{l\}\cup L_2$, with $L_2$ a path in $D_2$.

\end{itemize}

An \emph{address for $C$} is a partial function $s:\mathcal L(C)\rightharpoonup
\{0,1\}$ whose domain $\mathrm{dom}(s)\subseteq \mathcal L(C)$ is a path of $C$,
and the following condition holds: for all subterms of $C$ of the form
$\ite{l}{D_1}{D_2}$, such that $\mathrm{dom}(s)=\{l\}\cup L_b$, where $L_b$ is a
path in $D_i$, $s(l)=b$.

We let $@C$ indicate the set of addresses of $C$.
\end{definition}

Given a circuit $C$ and an address $s\in @C$, let us define what it means for a gate $U$ occurring in $C$ to \emph{occur at address $s$}:
\begin{itemize}
\item if $C=U$, then $U$ occurs in $C$ at address $\emptyset$;
\item if $C=C_1;C_2$ and $U$ occurs in $C_b$ at address $s$, then it occurs in $C$ at address $s$;
\item if $U=\ite{l}{D_1}{D_2}$, and $U$ occurs in $D_b$ at address $s'\in @D_b$, then $U$ occurs in $C$ at address $s=s'\cup\{(l,b)\}$.
\end{itemize}

In the following, we will adopt the following notation for projections and injections maps:
Given an environment of the form $\Gamma=\bit^l,\Delta$, we define, for $b\in \{0,1\}$, the maps 
\begin{align*}
\pi^l_b
&: \interp{\bit^l, \Delta} \simeq \interp{\Delta} \times \interp{\Delta} \to
\interp{\Delta} & 
\iota^l_b & : \interp{\Delta} \to \interp{\Delta} \times
\interp{\Delta} \simeq \interp{\bit^l, \Delta}
\end{align*}
In this way, the semantics
of if-then-else in circuits can be defined as:
$$
\interp{\ite{l}{C}{D}}= \interp{C}\circ \pi^l_0+ \interp{D}\circ \pi^l_1.
$$

For any well-typed circuit $\Gamma\rhd C\rhd \Delta$ and address $s\in @C$, we
define the CPM $C^{\sharp s}:\interp\Gamma\to\interp\Delta$ by induction as
follows:

\begin{itemize}
\item if $C=c\neq \meas$, then $C^{\sharp s}=\interp{c}$;

\item if $C=\meas^{l:\bit}_{l':\qbit}$ with index $\alpha$, then $C^{\sharp s}=\iota_{s(\alpha)}^l\circ \pi_{s(\alpha)}^l\circ \interp{\meas}$, that is, $C^{\sharp s}=\Phi(p)$, where $p:(2)\to (1,1)$ is the map given by
$$p\begin{pmatrix}
        a & b \\ c & d
      \end{pmatrix} = (1-s(\alpha)) (a, 0)+s(\alpha) (0, d).$$
            
\item if $C=D;E$, then $C^{\sharp s}=E^{\sharp(s|_{\mathcal L(E)})}\circ D^{\sharp(s|_{\mathcal L(D)})}$;
      
\item if $C=  \ite{l}{D_1}{D_2}$, then 
$C^{\sharp s}:=D_{s(l)}^{\sharp(s|_{\mathcal L(D_{s(\alpha)})})}\circ \pi_{s(l)}^l.$

\end{itemize}

For readability, but with a slight abuse of notation, in the following we will often write $C^{\sharp s}$ even when $s$ is a partial function whose \emph{restriction} $s|_{\mathcal L(C)}$ to $\mathcal L(C)$ yields an address in $@C$.

Given a circuit $C$ and a partial function $s:\mathcal L(C)\rightharpoonup\{0,1\}$, whenever $s\notin @C$ we set $C^{\sharp s}=0$. This allows us to rewrite the definition of $C^{\sharp s}$, when $C=\ite{l}{D,E}$ as
\begin{equation}\label{eq:itesl}
C^{\sharp s}=
(\ite{l}{D}{E})^{\sharp s}= (1-s(l))\Big(D^{\sharp s}\circ \pi_0^l\Big )
+ s(l)\Big(E^{\sharp s}\circ \pi_1^l\Big ).
\end{equation}

\begin{definition}[ite-closed circuit]
A well-typed circuit $\Gamma\rhd C\rhd \Delta$ is \emph{ite-closed} if any subcircuit of the form $\ite{l}{D}{E}$ of $C$ occurs in a composition $F;\ite{l}{D}{E}$, for some subcircuit $\Gamma'\rhd F\rhd \bit^l,\Delta'$.
\end{definition}

\begin{lemma}[Boolean choice]\label{lemma:circuitbool}
For any ite-closed circuit $\Gamma\rhd C\rhd \bit^l,\Delta$ and $s\in @C$, either
$\pi^l_0\circ C^{\sharp s}=0$ or $\pi^l_1\circ C^{\sharp s}=0$.

\end{lemma}
\begin{proof}
By induction on $C$, we only consider the critical cases:
\begin{itemize}
\item if $C=\meas^l$ has index $\alpha$, then, by definition, if $s(\alpha)=0$, then $\pi^l_1\circ C^{\sharp s}=0$ and conversely, if $s(\alpha)=1$, then $\pi^l_0\circ C^{\sharp s}=0$.

\item if $C=D;E$, then by induction hypothesis for some $b\in \{0,1\}$, 
$\pi^l_b\circ E^{\sharp s}=0$, so $\pi^l_b\circ C^{\sharp s}=\pi^l_b\circ E^{\sharp s}\circ D^{\sharp s}=0$.

\item if $C=D;\ite{l'}{F_1}{F_2}$, where $F_1,F_2$ have an output $\bit^{l}$, then, by induction hypothesis, for some $b\in \{0,1\}$, $\pi^{l'}_b\circ D^{\sharp s}=0$ and, for some $c_0,c_1\in \{0,1\}$, $\pi^{l}_{c_0}\circ F_{0}^{\sharp s}=\pi^{l}_{c_1}\circ F_{1}^{\sharp s}=0$. Using the linearity of completely positive maps, which implies $\Phi\circ 0=0$, we consider two cases:
\begin{itemize}
\item if $s(l')=b$, then 
$\pi^{l}_{c_b}\circ C^{\sharp s}=\pi^{l}_{c_b}\circ
F_b^{\sharp s}\circ \pi^{l'}_b\circ D^{\sharp s}=
\pi^{l}_c\circ
F_b^{\sharp s}\circ 0=
0$.

\item if $s(l')=1-b$, then 
$\pi^{l}_{c_{1-b}}\circ C^{\sharp s}=\pi^{l}_{c_{1-b}}\circ 
F_{1-b}^{\sharp s}\circ \pi^{l'}_{1-b}\circ D^{\sharp s}=
0\circ \pi^{l'}_{1-b}\circ D^{\sharp s}=
0$. \qedhere

\end{itemize}

\end{itemize}

\end{proof}

\begin{proposition}\label{prop:circuitsum}
For any well-typed ite-closed circuit $\Gamma\rhd C\rhd \Delta$, 
$
\interp C=\sum_{s\in @C}C^{\sharp s}.
$
\end{proposition}
\begin{proof}
By induction on $C$. We consider the two critical cases:
\begin{itemize}
\item $C=D;E$: observe that $\mathcal L(C)=\mathcal L(D)+\mathcal L(E)$, so any $s\in @C$ is uniquely of the form $s_D+s_E$, with $s_D\in @D$ and $s_E\in @E$. We then have 
\begin{align*} 
\interp{C}&= \interp{E}\circ \interp{D}\\
&\stackrel{\tiny{\text{[IH]}}}{=}
\left( \sum_{s\in @E}E^{\sharp s}\right)\circ  
\left( \sum_{s\in @D}D^{\sharp s}\right)
\\ &= \sum_{s_D\in @D, s_E\in @E}E^{\sharp s_E}\circ D^{\sharp s_D}
\\&=\sum_{s\in @C}E^{\sharp s|_{\mathcal L(E)}}\circ D^{\sharp s|_{\mathcal L(D)}}=
\sum_{s\in @C}(D;E)^{\sharp s}=\sum_{s\in @C}C^{\sharp s}.
\end{align*}

\item $C= D;\ite{l}{F_1}{F_2}$; by Lemma \ref{lemma:circuitbool} for any $s\in @D$ (and a fortiori for any $s\in @C$) there exists $b_s\in \{0,1\}$ such that $\pi^l_{1-b_s}\circ D^{\sharp s}=0$.
We can then compute:
\begin{align*}
\interp{C}&= \big(\interp{E}\circ \pi^l_0
+  \interp{F}\circ \pi^l_1\big)\circ \interp{D}
=  \interp{E}\circ \pi^l_0\circ \interp{D}
+ \interp{F}\circ \pi^l_1\circ \interp{D}
\\ &\stackrel{\tiny{\text{[IH]}}}{=}
\left(
\sum_{s_1\in @F_1, s_D\in @D}F_1^{\sharp s_1} \circ \pi^l_0\circ D^{\sharp s_D}\right) \\
&\quad + \left( \sum_{s_2\in @F_1,s_D\in @D}F_2^{\sharp s_2} \circ \pi^l_1\circ D^{\sharp s_D}\right) \\
&\stackrel{(\star)}{=}
\sum_{s\in @C}
F^{\sharp s}_{b_s}\circ \pi^l_{b_s}\circ D^{\sharp s}
\\&\stackrel{(\star\star)}{=}
\sum_{s\in @C}
F^{\sharp s}_{s(l)}\circ \pi^l_{s(l)}\circ D^{\sharp s}
\\&=\sum_{s\in @C}C^{\sharp s}.
\end{align*}
Where:
\begin{itemize}
\item the equality $A\stackrel{(\star)}{=}B$ holds because all terms $F_{1-b_s}^{\sharp s}\circ \pi_{1-b_s}^l\circ D^{\sharp s}$ are equal to $0$, since $\pi_{1-b_s}^l\circ D^{\sharp s}=0$;

\item the equality $A\stackrel{(\star\star)}{=}B$ is justified as follows: 
for all $s\in @C$, if $s(l)\neq b_s$, then $F^{\sharp s}_{s(l)}\circ \pi^l_{s(l)}\circ D^{\sharp s}=0$ so the sum $B$ coincides with $\sum_{s\in @C, s(l)=b_s}
F^{\sharp s}_{s(l)}\circ \pi^l_{s(l)}\circ D^{\sharp s}$, which is in fact the same as $A$. \qedhere

\end{itemize}
\end{itemize}
\end{proof}

\subsubsection{Extended Circuits}

We consider now slicing of extended circuits. 
For simplicity, we restrict here our attention to extended circuits \emph{without} if-then-else constructs of the form $\ite{l}{F_1}{F_2}$, but only containing if-then-else constructs of the form $C\xmapsto{l}(D,E)$.
We will consider the more general case in Subsection \ref{subs:excircplus}.

Let $E$ be such an extended circuit; the set of circuits $\mathrm{circ}(E)$ is defined by induction as 
\begin{align*}
\mathrm{circ}(C)&=\{C\},\\
\mathrm{circ}(C\xmapsto{l}(D,E))&=\{C\}\cup \mathrm{circ}(D)\cup \mathrm{circ}(E).
\end{align*}
As in the previous subsection, we suppose given an indexing $\alpha_1,\dots, \alpha_n$ of all $\meas$ gates occurring in the circuits in $\mathrm{circ}(E)$, with distinct gates always corresponding to distinct indices.

We define the following sets:
\begin{itemize}
\item $\LITES(E)$ is the set of labels in the guard of some if-then-else of the form $C\xmapsto{l}{(D,E)}$.

\item $\LMEAS(E)$ is the set of indices of $\meas$-gates in the circuits in $\mathrm{circ}(E)$;

\item $\mathcal L(E)=\LITES(E)\cup \LMEAS(E)$.
\end{itemize}

Notice that, since we no more consider if-then-elses of the form $\ite{l}{F}{G}$, an address $s\in @C$ for a (basic) circuit is now simply a \emph{total} function $s:\LMEAS(C)\to \{0,1\}$.

We define addresses for extended circuits:
\begin{definition}\label{def:address2}
Let $E$ be an extended circuit. 
We define an \emph{address for $E$} as a finite set of pairs $\{(l_0,b_0),\dots, (l_n,b_n)\}$ such that the sequence 
$[(l_0,b_0),\dots,(l_n,b_n)]$ indicates a path from the root of $E$ to one of its leaves.
We let $@E\subseteq (\LITES(E)\rightharpoonup \{0,1\})$ indicate the set of addresses of $E$.
We indicate addresses as $\ADD,\ADDH$.
\end{definition}

\begin{definition}
Given a well-typed circuit $\Gamma\rhd C\rhd \bit^l,\Delta$, let $\Gamma\rhd\DISC{l}{C}\rhd \Delta$ be the circuit obtained by discarding the bit of label $l$.
For all extended circuit $E$ and address $\ADD\in @E$ we define the circuits $E@\ADD$ and $E|_\ADD$ by induction as follows:
\begin{itemize}
\item if $E=C$, then $E@\ADD=E|_\ADD=E$;
\item if $E=C\xmapsto{l}(E_1,E_2)$ and $\ADD(l)=b$, then $E@\ADD=E_b@\ADD$ and 
$E|_\ADD=\DISC{l}{C};E_b|_\ADD$.

\end{itemize}
A gate $c$ is said to \emph{occur at address $\ADD$} if $c$ occurs in the circuit $E|_\ADD$.
\end{definition}

We introduce now \emph{super-}addresses, which are partial functions taking into account \emph{both} if-then-elses of the form $C\xmapsto{l}{(D,E)}$ and $\meas$-gates:

\begin{definition}

A partial function $s:\mathcal L(E)\rightharpoonup \{0,1\}$ is called a \emph{super-address of $E$} if the following hold:
\begin{enumerate}
\item its restriction to $\LITES(E)$, that we indicate as $\ADD_s$, is an address;
\item for any $\alpha\in \LMEAS(E)$, $s(\alpha)$ is defined iff the corresponding $\meas$-gate occurs in $E$ at address $\ADD_s$.

\end{enumerate}
We let $@^{\mathsf s}E$ indicate the set of super-addresses of $E$. 

\end{definition}

\begin{example}
Let $E$ be
$$
E=
C_0\xmapsto{l_0}\Big(
C_{10}\xmapsto{l_1}(C_{11},C_{12}),
C_{20}\xmapsto{l_2}(C_{21},C_{22})
\Big).
$$
The addresses of $E$ are the partial functions of either of the form $\ADD=\{(l_0,0),(l_1,b)\}$ or $\ADD=\{(l_0,1),(l_2,b)\}$, for some $b\in \{0,1\}$.
Moreover, suppose $C_{11}$ contains a gate $\meas$ of index $\alpha$ and $C_{20}$ contains a
gate $\meas$ of index $\beta$; then $s=\{(l_0,0), (l_1,1),(\alpha,0)\}$ and $t=\{(l_0,1),
(l_2,0),(\beta,0)\}$ are super-addresses, while $s=\{(l_0,0),
(l_1,1)\}$ and $t=\{(l_0,1), (l_2,0),(\alpha,0),(\beta,1)\}$ are not super-addresses.
\end{example}

\begin{rem}
The notion of super-address for extended circuits is different from the corresponding notion of address for a basic  circuit with if-then-else: this is due to the fact that, while the extended circuits truly are,
wrt the $\xmapsto{l}{(D,E)}$, binary branching trees whose nodes are labelled
with circuits, basic circuits, wrt to the $\ite{l}{D}{E}$, may have a more complex
structure (e.g.~a circuit may be of the form $C=C_1;C_2$ so that the
if-then-elses inside $C_1$ are, so to say, incomparable with those inside
$C_2$).
\end{rem}

\begin{lemma}\label{lemma:addresses}
For any extended circuit $E$, $@^{\mathsf s}E=@\translate(E)$
\end{lemma}
\begin{proof}
The fundamental observation is that the circuits of the form $\translate(E)$ belong to the following circuit sub-grammar:
$$
B:= C\mid( \ite{l}{B}{B});C
$$
where $C$ indicates an arbitrary if-then-else-free circuit. 
Call such circuits $B$ \emph{restricted}. All restricted circuits are ite-closed. Moreover, they correspond to binary trees whose nodes are labelled by if-then-else-free circuits. One can easily check by induction that an address for a restricted circuit is the same as a super-address for the corresponding extended circuit. \qedhere
\end{proof}

We will often use of the following useful construction on super-addresses:
\begin{definition}[super-address extension]
For all extended circuit $E$, set $S\subseteq{\mathcal L(E)}\rightharpoonup \{0,1\}$, address $\ADD\in @E$ and label $\ell\notin\mathcal L(E)$, let
\[
\EXT(E,S,\ADD,\ell);=\{s\cup\{(\ell,b)\}\mid s\in S, \ADD_s=\ADD, b\in \{0,1\}\}\cup \{s\in S\mid \ADD_s\neq \ADD\}.
\]
\end{definition}
The following fact, easily checked, shows a first use of the construction above:

\begin{lemma}\label{lemma:superaddress}
For all extended circuits $E,F$ and address $\ADD\in @E$, let $F$ be obtained from $E$ by replacing the circuit $C:=E@\ADD$ by $C\xmapsto{l}{(\mathrm{id},\mathrm{id})}$. 
Then $@^{\mathsf s}F=\EXT(E,@^{\mathsf s}E,\ADD,l)$.
\end{lemma}

For any well-typed extended circuit $\Gamma\rhd E\rhd \mathcal I$ and super-address $s\in @^{\mathsf s}E$ such that $\Gamma\rhd E@\ADD_s\rhd \Delta$, we define the CPM $E^{\sharp s}:\interp\Gamma\to\interp\Delta$ by induction as follows:
\begin{itemize}
\item if $E=C$, then $E^{\sharp s}=C^{\sharp{s}}$;

\item if $E=  C\xmapsto{l}(F_0,F_1)$, then $E^{\sharp s}= F_{s(l)}^{\sharp s}\circ \pi_{s(l)}^l \circ C^{\sharp s}$.
\end{itemize}

Given an extended circuit $E$, for all partial function $s:\mathcal L(E)\rightharpoonup\{0,1\}$, not in $ @^{\mathsf s}E$ we set $E^{\sharp s}=0$. This allows us to rewrite the definition of $E^{\sharp s}$, when $E=C\xmapsto{l}{(F,G)}$, as
$$
E^{\sharp s}=
\left(C\xmapsto{l}{(F,G)}\right)^{\sharp s}= (1-s(l))\Big(F^{\sharp s}\circ \pi_0^l\circ C^{\sharp s}\Big )
+ s(l)\Big(G^{\sharp s}\circ \pi_1^l\circ C^{\sharp s}\Big ).
$$

\begin{lemma}\label{lemma:extequiv}
For any uniform extended circuit $E$ and super-address $s\in @^{\mathsf s}E$, $E^{\sharp s}=\translate(E)^{\sharp s} $. Moreover, $E^{\sharp s}$ is of the form
\begin{equation}\label{eq:dag}
E^{\sharp s}=C_0^{\sharp s}\circ (\pi_{s(l_1)}^{l_1}\circ C_1^{\sharp s})\circ \dots\circ
(\pi_{s(l_n)}^{l_n}\circ C_n^{\sharp s}),
\tag{$\dag$}
\end{equation}
where $l_1,\dots, l_n$ indicates the list (in reversed chronological order) of labels encountered in exploring the tree of $E$ from root to leaves, and 
 $C_0,\dots, C_n$ are the circuits (in reversed chronological order) encountered.
\end{lemma}
\begin{proof}
By induction on $E$:
\begin{itemize}

\item if $E=C$, then $E=\translate(E)$, so it is clear that $E^{\sharp s}=\translate(E)^{\sharp s}$. Equation \eqref{eq:dag} also trivially holds with $n=0$ and $C_0:=C$.

\item if $E=C\xmapsto{l}{(F,G)}$, then $\translate(E)=C;\ite{l}{\translate(F)}{\translate(G)}$; if $s(l)=0$, then 
\begin{align*}
E^{\sharp s}&=(1-s(l))(F^{\sharp s}\circ\pi_0^l\circ C^{\sharp s})+s(l)(G^{\sharp s}\circ\pi_1^l\circ C^{\sharp s})\\
&\stackrel{\tiny{\text{[IH]}}}{=}
(1-s(l))(\translate(F)^{\sharp s}\circ\pi_0^l\circ C^{\sharp s})+s(l)(\translate(G)^{\sharp s}\circ\pi_1^l\circ C^{\sharp s})\\
&=
\Big((1-s(l))\big(\translate(F)^{\sharp s}\circ\pi_0^l\big)+s(l)\big(\translate(G)^{\sharp s}\circ\pi_1^l\big)\Big)\circ C^{\sharp s}
=(\translate(E))^{\sharp s}.
\end{align*}
Moreover, by induction hypothesis, equation \eqref{eq:dag} holds for $F$ and $G$, and we can conclude that the equation holds for $E$ as well since $E^{\sharp s}$ is of either of the forms 
$F^{\sharp s}\circ(\pi_0^l\circ C^{\sharp s})$ and $G^{\sharp s}\circ(\pi_1^l\circ C^{\sharp s})$. \qedhere
\end{itemize}
\end{proof}

\begin{proposition}\label{prop:sumext}
For any uniform extended circuit $E$, $\interp{E}=\interp{\tau(E)}=\sum_{s\in @^{\mathsf s}E}E^{\sharp s}$.

\end{proposition}
\begin{proof}
By Lemma \ref{lemma:addresses} we have that $@\translate(E)=@^{\mathsf s}E$.
Using Proposition \ref{prop:circuitsum}, Lemma \ref{lemma:extequiv} and the definition of $\interp{E}$ we have that $\interp{E}=\interp{\translate(E)}=\sum_{s\in @^{\mathsf s}E}\tau(E)^{\sharp s}=\sum_{s\in @^{\mathsf s}E}E^{\sharp s}$.
\end{proof}

In the following we need the following auxiliary lemmas.

Given $r$ qubits, a set $I\subset \{1,\dots, r\}$, and a gate $c$ over $r$ qubits, we indicate as ${c}^I $ the gate that applies $c$ to the qubits indexed by $I$ and is the identity on all other qubits.

\begin{lemma}\label{lemma:help1}
For all extended circuits $E$, $s\in @^{\mathsf s}E$ and gate $c$, let $E'$ be obtained from $E$ by replacing the circuit $E@\ADD_s$ by $E@\ADD_s;c$. Then
\begin{itemize}
\item[i.] if $c$ is unitary or $c=\new$, then 
$(E')^{\sharp s}=\interp{c^I}\circ E^{\sharp s}$;
\item[ii.] if $c=\meas$ of index $\alpha$, then $@^{\mathsf s}E'=\EXT(E,@^{\mathsf s}E,\ADD_s,\alpha)$ and 
$(E')^{\sharp(s\cup\{(\alpha,b)\})}=\iota_{b}\circ \pi_{b}\circ \interp{c^I}\circ E^{\sharp s}$.

\end{itemize}
\end{lemma}
\begin{proof}
By \eqref{eq:dag} we can write $E^{\sharp s}$ as $(E@\ADD_s)^{\sharp s}\circ \phi$ and $(E')^{\sharp s}$ as $(E@\ADD_s;c)^{\sharp s}\circ \phi$. We then have for $c\neq\meas$, 
$$
(E')^{\sharp s}= (E@\ADD_s;c^I)^{\sharp s}\circ \phi=
\interp{c^I}^{\sharp s}\circ (E@\ADD_s)^{\sharp s}\circ \phi=
\interp{c^I}\circ E^{\sharp s}.
$$
and, for $c=\meas$ and $s'=s\cup\{(\alpha,b)\}$,
\begin{align*}
(E')^{\sharp s'} &= (E@\ADD_s;\meas^I)^{\sharp s'}\circ \phi 
 = \interp{\meas^I}^{\sharp s'}\circ (E@\ADD_s)^{\sharp s}\circ \phi \
= \iota_{b}^l\circ \pi_{b}^l\circ\interp{\meas^I}\circ E^{\sharp s}.
\end{align*}
\end{proof}

\begin{lemma}\label{lemma:help2}
Let $ E$ be a well-typed extended circuit, let 
$s\in @^{\mathsf s}E$ and let $\Gamma'\rhd C\rhd \bit^l,\Delta'$ be $E@\ADD_s$. Let $E'$ be $E$ with $C$ replaced by $C\xmapsto{l}{(\mathrm{id},\mathrm{id})}$. Then $@^{\mathsf s}E'=\EXT(E,@^{\mathsf s}E,\ADD_s,l)$ and for all $b\in \{0,1\}$, $(E')^{\sharp(s\cup(l,b))}=\pi^l_b\circ E^{\sharp s}$.
\end{lemma}
\begin{proof}
Using \eqref{eq:dag}, we have
\begin{align*}
\pi^l_b\circ E^{\sharp s}
&=
\pi^l_b\circ \Big( C^{\sharp s}\circ\phi\Big)
\\
&=
(1-b)\Big( \pi^l_0\circ C^{\sharp s}\circ\phi\Big)+
b\Big(
 \pi^l_1\circ C^{\sharp s}\circ\phi\Big)
 \\
&=
(1-b)\Big(
\mathrm{id}\circ( \pi^l_0\circ C^{\sharp s}\circ\phi)\Big)+
b\Big(
\mathrm{id}\circ( \pi^l_1\circ C^{\sharp s}\circ \phi)\Big)\\
&=
(E')^{\sharp(s\cup(l,b))}.
\end{align*}
\end{proof}

\subsection{The Machine \evaltkm}\label{subsec:evaltkm}

We now provide a more formal definition of the machine \evaltkm. Our definition follows the essence of \cite{DalLago2017}, but is adapted to our setting.

\subsubsection{Quantum+Classical Registers}

Let us first define the kind of registers used by the machine.

\begin{definition}[quantum+classical register]
A \emph{quantum+classical environment} is a list $L=[l_1:\mathbb B_1, \dots, l_n:\mathbb B_n]$ of assignments relating labels $l_1,\dots, l_n$ with either $\qbit$ or $\bit$; for such $L$, let $L_{\qbit}$ and $L_{\bit}$ indicate the lists obtained from $L$ by selecting, respectively, the labels of type $\qbit$ and those of type $\bit$.

A \emph{quantum+classical register} is a triple $m=[L,Q,v]$, where:
\begin{itemize}
\item $L$ is a quantum+classical assignment;
\item $Q\in \mathbb C^{2^{|L_\qbit|}}$ is a quantum state over $|L_\qbit|$ qubits;
\item $v\in \{0,1\}^{|L_\bit|}$ is a list of $|L_\bit|$ Boolean values.
\end{itemize}

\end{definition}

Let us define the \emph{mixed state} associated with a quantum+classical register:

\begin{definition}
Define $\tocpm(0)=(1,0)\in \complex\times\complex$ and
$\tocpm(1)=(0,1)\in \complex\times\complex$. For any natural numbers
$m,n$, vector $Q\in \complex^{m}$ and Boolean vector $v\in \{0,1\}^n$,
let
\begin{itemize}
\item $ \tocpm(Q)=Q\cdot Q^{\dag} \in V_{m}=\complex^{m\times m}$,
where $\cdot$ indicates the outer product;
\item $\tocpm(v)=\iota_v(1)\in \otimes_{i=1}^nV_{(1,1)}=V_{(1,\dots,
1)}=\prod_{i=1}^{2^n}  \complex$, where $\iota_v: \complex\to
\prod_{i=1}^{2^n} \complex$ is the linear map $\iota_v(x)=(0,\dots,
0,x,0,\dots, 0)$ placing $x$ at the position corresponding to the
``address'' $v$;
\end{itemize}

For any environment $\Delta=\{l_1:\mathbb B_1,\dots,l_n:\mathbb B_n\} $, with $l_1<\dots<l_n$, we have $\interp{\Delta}:=\interp{\mathbb B_1}\otimes\dots\otimes\interp{\mathbb B_n}$.
Given a quantum+classical register $L=[l_1:\mathbb B_1,\dots,l_n:\mathbb B_n]$,
define $\interp{L}:=\interp{|L|}$. However, recall that $l_1,\dots, l_n$ need not be in $<$-order, so in fact 
$\interp{L}=\interp{B_{\sigma(1)}}\otimes\dots\otimes\interp{B_{\sigma(n)}}$ for some uniquely defined permutation $\sigma\in \mathfrak S_n$.
Call then $\varphi_L$ the (uniquely defined) isomorphism
\[
\varphi_L: \interp{L_\qbit}\otimes\interp{L_\bit}\longrightarrow \interp{\mathbb B_{\sigma(1)}}\otimes\dots\otimes\interp{\mathbb B_{\sigma(n)}}=\interp{L}.
\]
For any quantum+classical memory $m=[L,Q,v]$, we finally define 
\[
\tocpm(m):= \varphi_L(\tocpm(Q)\otimes\tocpm(v))\in \interp{L}.
\]
\end{definition}

The intuition for the definition of $\tocpm(m)\in\interp{L}$ is that the qubits
and bits in $Q$ and $v$ are re-arranged to be ordered following the linear
ordering of the corresponding labels. We now consider the case of
\emph{pseudo-distributions} of quantum+classical registers:
\begin{definition}
Let $\theta=\{m_1^{\theta_1},\dots, m_n^{\theta_n}\}$ be a pseudo-distribution of quantum registers, where 
$m_n=[L_n,Q_n,v_n]$.
$\theta$ is said \emph{uniform} where for all $i,j=1,\dots, n$, $L_i=L_j$. For a uniform such $\theta$, we define
\[
\tocpm(\theta):=\sum_{i=1}^n \theta_i\cdot \tocpm(m_i)\in \interp{L_1}.
\]
\end{definition}

\begin{example}
Consider the quantum states $Q_0,Q_1\in \complex^2$ and the Boolean
values $v_0=0$ and $v_1=1$.  Let $L=[l_1:\qbit,l_2:\bit]$ (notice that $\varphi_L$ is the identity). Let $\theta$ be the uniform pseudo-distribution $\{
[L,Q_0, v_0]^p,[L,Q_1, v_1]^{1-p}\}$. We have then
$$
\tocpm(\theta)=
\left(pQ_0Q_0^{\dag}, (1-p)Q_1Q_1^{\dag}\right)
\in V_{2 \otimes (1,1)}=  \complex^{2\times 2}\times \complex^{2\times 2}=\interp{L}.
 $$

\end{example}

The following lemma explains how the application of unitary gates to a distribution of quantum states relates to the corresponding mixed states.

\begin{lemma}\label{lemma:CU} For any unitary gate
$U:\qbit^{n}\to\qbit^{n}$ and pseudo-distribution of quantum states
$\theta=\{Q_1^{\theta_1},\dots,Q_k^{\theta_k}\}$, where $Q\in \mathbb C^{2^n\times 2^n}$, $\interp{U}(
\tocpm(\theta))=\tocpm(\theta_U)$, where
$\theta_U=\theta_U\{({U}Q_1)^{\theta_1},\dots,({U}Q_k)^{\theta_k}\}$.
\end{lemma}
\begin{proof}
We have 
$  \interp{U}( \tocpm(\theta))  =\interp{U}( \sum_i\theta_i Q_iQ_i^{\dag}) 
   = \sum_i\theta_i ({U}Q_i)({U}Q_i)^{\dag} 
   = \tocpm(\theta_U)$.
\end{proof}

\subsubsection{Definition of the Machine}

We will now provide a sketch of the definition of the \evaltkm.

Let us recall that a configuration of the \evaltkm\ is a triple $\config=(\pi,\MM,m)$, where 
all tokens in $\MM$ have \emph{the same} address $\ADD_{\config}$ (called \emph{the address of $\config$}), and 
$m=[L,Q,c]$ is a quantum+classical register satisfying $|L|=\mathbf{ENV}(\MM)$. In other words, the labels assigned in $L$ coincide with the positions of the tokens in $\MM$.  

Given a derivation $\pi$,
\begin{varitemize}
\item an \emph{initial configuration} for $\pi$ has the form $(\pi, \mathbf{TKN}(\negative_\pi\cup \negones_\pi\cup\ones_\pi^{\downarrow}), m)$, where $m=[L,Q,c]$ is any quantum+classical register such that
$|L|=\mathbf{ENV}(\MM)$;

\item a \emph{final configuration} for $\pi$ has the form $(\pi, \MM, m)$, where
$\mathbf{PSN}(\MM)=\mathbf{TKN}(\positive_\pi\cup \posones_\pi\cup\ones_\pi^{\downarrow}\cup\guard\pi)$ and 
 $m=[L,Q,c]$ is any quantum+classical register such that
$|L|=\mathbf{ENV}(\MM)$.

\end{varitemize}

Observe that the positions of the tokens in the initial and final configurations are as for the \circuittkmz. However, while in the final configuration the \circuittkmz may contain \emph{multiple} tokens over any final position (each with a different address),  the \evaltkm may only contain \emph{one} such token, since all tokens have the same address.

The reduction of the \evaltkm is, as in \qlambda, probabilistic, and is thus described by a monadic reduction relation $\config\totkm\{\config_1^{\mu_1},\dots, \config_n^{\mu_n}\}$, that we describe below.

The \emph{structural rules} are defined similarly to the \circuittkmz: more precisely, for each structural rule
\[
(\{\mathfrak a:A\}\cup\MM,E)\totkm (\set{\mathfrak b:B}\cup\MM, E[l_{\mathfrak a}\mapsto l_{\mathfrak b}])
\]
of the \circuittkmz, there is a corresponding (monadic) structural rule
\[
(\{\mathfrak a:A\}\cup\MM,m)\totkm\left\{ (\set{\mathfrak b:B}\cup\MM, m[l_{\mathfrak a}\mapsto l_{\mathfrak b}])^1\right\}
\]
of the \evaltkm. This also holds for the guard rules.

The \emph{circuit rules} of the \evaltkm are defined as follows:
\begin{varitemize}
\item for all $U\neq \new, \meas$, we have a transition
{\small
\begin{align*}
 (\pi,\set{\mathfrak a_1 : \baseT_1, \dots, \mathfrak a_n : \baseT_n} \cup \MM, 
             [L,Q,c])
            \totkm
\left\{ 
 \left (\pi, \set{\mathfrak a_{n+1} : \baseT_{n+1}, \dots, \mathfrak a_{2n} : \baseT_{2n}} \cup \MM,
            [L[l_{\mathfrak a_i}\mapsto l_{\mathfrak a_{n+i}}], U^{I}Q, c]
        \right   )^1
        \right\},
            \end{align*}
            }
    where $U^IQ$ indicates the application of the gate $U$ to $Q$ following the labels $I=\{l_{\mathfrak{a}_1},\dots, l_{\mathfrak{a}_n}\}$. 

\item    When $U=\new$, the transition is 
{\small
\begin{align*}
 (\pi,\set{\mathfrak a : \baseT} \cup \MM, 
             [L,Q,c])
            \totkm\left\{
           \left (\pi, \set{\mathfrak b : \baseT'} \cup \MM,
            [L[l_{\mathfrak a}:\bit\mapsto l_{\mathfrak b}:\qbit], Q\otimes\ket{b}, c-\set{(l,b)}]
        \right   )^1\right\}.
            \end{align*}
            }

\item    When $U=\meas$, the operation $U^IQ$ is replaced by an actual \emph{measurement} of the qubit at $l_{\mathfrak a_1}$, and the transition occurs with the probability associated $|\alpha_b|^2$ with the measured output:
\begin{align*}
 (\pi,\set{\mathfrak a : \baseT} \cup  \MM, 
           &  [L,Q,c])\\
           \totkm\Big\{&
           \left (\pi, \set{\mathfrak b : \baseT'} \cup \MM,
            [L[l_{\mathfrak a}:\qbit\mapsto l_{\mathfrak b}:\bit], Q_0, c\cup\set{(l,0)}]
        \right   )^{|\alpha_0|^2},\\
        &  \left (\pi, \set{\mathfrak b : \baseT'} \cup \MM,
            [L[l_{\mathfrak a}:\qbit\mapsto l_{\mathfrak b}:\bit], Q_1, c\cup\set{(l,1)}]
        \right   )^{|\alpha_1|^2}
        \Big\},
            \end{align*}
    where $Q=\alpha_0(Q_0\otimes \ket 0)+\alpha_1(Q_1\otimes\ket 1)$.
    \end{varitemize}
    
    \medskip
    
The asynchronous rules are defined as follows: 
\begin{varitemize}   
\item there is a transition
\[
(\pi,\set{(\sigma:\bit,\ADD)}\cup\MM,m)\totkm \left\{(\pi,\set{(\sigma:\bit,\ADD_b)}\cup\ADD_b(\MM), m^*)^1\right\}, 
\]
provided that the label $l=\mathsf{lab}(\bit)$ does not occur in $\ADD$, and where:
	\begin{itemize}
	\item $m^*$ is obtained from $m$ by removing the value $(l,b)$ contained in the classical register;
	\item $\ADD_b$ is $\ADD_b\cup\{(l,b)\}$.
	\end{itemize}
\item For each conditional transition 
{\small
\[
  (\set{( \bit, \ADD)} \cup \set{\mathfrak a :  A} \cup \MM, E)
  \totkm
  (\set{( \bit,\ADD)} \cup \set{\mathfrak b: B} \cup \MM, E[l_{\mathfrak a}\mapsto l_{\mathfrak b}]_\ADD), \text{provided }(l\to b)\in \ADD
\]
}
of the \circuittkmz there is a corresponding transition
{\small
\[
  (\set{( \bit, \ADD)} \cup \set{\mathfrak a :  A} \cup \MM,m)
  \totkm
  \left\{
  (\set{( \bit,\ADD)} \cup \set{\mathfrak b: B} \cup \MM, m[l_{\mathfrak a}\mapsto l_{\mathfrak b}]_\ADD)^1\right\}, \text{provided }(l\to b)\in \ADD
\]
}
of the \evaltkm.

\end{varitemize}

The following result, which expresses the soundness of the \evaltkm with respect to the evaluation of quantum closures, can be proved along the lines of Theorem 29 in \cite{DalLago2017}.

\begin{proposition}
For any Boolean type derivation $\pi:\Gamma\vdash M:A$ and quantum+classical register $m=[L,Q,c]$, 
$\DISTM{Q,L_{\qbit},M}=\DIST{Q,L,M}$.
\end{proposition}

\subsection{Coherent Configurations of the \evaltkm}\label{subsec:evalpara}

We now introduce a \emph{parallel} probabilistic reduction for the $\evaltkm$: a
state for this reduction will be given by a pseudo-distribution of configurations of
the $\evaltkm$ which, intuitively, correspond to the outcomes of different ways
in which the machine execution could have gone.

\begin{definition}
A \evaltkm\ configuration $\config=(\pi, \MM,m)$ is called \emph{well-sliced}
when for all $(l,\sigma,\ADD),(l',\sigma',\ADD')\in \MM$, one has $\ADD=\ADD'$.
For a well-sliced configuration $\config$ with shared set of addresses $\ADD$,
we let $@\config:=\ADD$.
\end{definition}

\begin{lemma}
Any reachable configuration $\config$ of the $\evaltkm$ is well-sliced, and the set $@\config$ contains all choices made when traversing an if-then-else.
\end{lemma}
\begin{proof}
This follows from the observation that the unique rule which modifies the addresses $G$ in the tokens is the guard rule, and this rule applies uniformly to \emph{all} the tokens in the current configuration.
\end{proof}

We need the following notion:
\begin{definition}[coherent pseudo-distribution]
For all extended circuits $E$, a \emph{$E$-indexed coherent pseudo-distribution} is a pseudo-distribution
$\mu=\{\config_s^{\mu_s}\}_{s\in @^{\mathsf s}E}$ (where $\config_s=(\pi,\MM_s,m_s)$) of well-sliced configurations of the $\evaltkm$ such that for all $s,t\in @^{\mathsf s}E$:
\begin{enumerate}
\item $@\config_s=\ADD_s$;

\item if $\ADD_s=\ADD_t$, then $\MM_s=\MM_t$.
\end{enumerate}

\end{definition}

In a coherent pseudo-distribution of configurations, two configurations over equivalent slices must have the same tokens, but may have different quantum+classical registers. Intuitively, two configurations over non-equivalent super-addresses correspond to two runs of the $\evaltkm$ which differ in the choices made when entering some if-then-else rule (in which case the runs cannot have the tokens in the same positions), while two configurations over equivalent super-addresses correspond to runs of the $\evaltkm$ which do not (yet) differ in the choices made when entering some if-then-else rule, but which differ in the output produced by some $\meas$-gate (but whose tokens have nonetheless moved in the same way).

In the following we will relate a single run of the $\circuittkmz$ with many parallel runs of the $\evaltkm$: a single configuration of the former, producing a circuit $E$, will be related with an $E$-indexed pseudo-distribution of coherent configurations of the latter; one token move of the former will be related with one or more token moves of the latter. For this, we need to introduce a notion of parallel reduction of pseudo-distributions as we do below.

\begin{definition}\label{def:redtype}
A \emph{reduction type of the $\evaltkm$} is a tuple $a=(\MM_0,\MM_1,\dots,\MM_n)$ of well-sliced sets of tokens of any of the following kinds:
\begin{itemize} 

\item \emph{deterministic type}: $n=1$, and $\MM_1$ is obtained from $\MM_0$ by moving a finite number of tokens according to any of the structural rules of the $\evaltkm$;

\item \emph{choice type}: $n=2$ and $\MM_1=\MM_2$ is obtained from $\MM_0$ by moving one token through a $\meas$-gate;

\item \emph{pseudo-deterministic type}: $n=2$, $\MM_0$ contains a token in the guard of an if-then-else of label $l$ and, $\MM_1$ is obtained by updating the addresses with $(l,0)$, while $\MM_2$ is obtained by updating the addresses with $(l,1)$.
\end{itemize}
\end{definition}

Let us first define what it means for the $\circuittkmz$ to make a reduction of type $a$.
Given a $\circuittkmz$-configuration $(\pi,\MM,E)$ and an address $\ADD\in @E$, let $\MM|_\ADD$ indicate the set of tokens in $\MM$ whose address is $\ADD$.

\begin{definition}
Let $\config=(\pi,\MM,E),\donfig=(\pi,\MM',E')$ be two configurations of the $\circuittkmz$, let $\ADD\in @E$, and $a=(\MM_0,\MM_1,\dots, \MM_n)$, $n=1,2$, be a reduction type.
We say that there is a \emph{reduction step of type $a$ from $\config$ to $\donfig$ at address $\ADD$}, noted $\config\to^a_\ADD\donfig$ if the following hold:
\begin{itemize}
\item $\MM|_\ADD=\MM_0$ and $\MM'$ is obtained from $\MM$ by replacing the tokens $\MM_0$ with $\bigcup_{i=1}^n\MM_i$;
\item $E'$ is obtained by updating $E$ following the reduction associated with $a$.

\end{itemize}
\end{definition}

We now consider the (more delicate) case of parallel probabilistic reduction of type $a$ for the $\evaltkm$:

\begin{definition}[parallel reduction of the \evaltkm]\label{def:parallelred}
Let $E,F$ be two extended circuits. 
Let $\mu=\{\config_s^{\mu_s}\}_{s\in @^{\mathsf s}E}$, $\nu=\{\donfig_t^{\nu_t}\}_{t\in @^{\mathsf s}F}$ be, respectively, an $E$-indexed and a $F$-indexed coherent pseudo-distribution of configurations of the $\evaltkm$. Let $a$ be a reduction type and $\ADD\in @E$. We say that there is a \emph{parallel reduction step of type $a$ from $\mu$ to $\nu$ at address $\ADD$}, noted $\mu\mtocl^a_{\ADD}\nu$, if the following hold:

\begin{itemize}
\item $F$ is obtained from $E$ by applying the circuit transformation corresponding to the reduction type $a$ at address $\ADD$;

\item for all $s\in @^{\mathsf s}E$ such that $\ADD\neq \ADD_s$, $s\in @^{\mathsf s}F$ and $\config_s=\donfig_s$, $\mu_s=\nu_s$ and, conversely, for all $s\in @^{\mathsf s}F$ such that $\ADD\neq \ADD_s$, $s\in @^{\mathsf s}E$ and $\config_s=\donfig_s$, $\mu_s=\nu_s$ (i.e.~the reduction does not move any token outside address $\ADD$);
\item one of the following holds:
	\begin{enumerate}
	
	\item $a=(\MM_0,\MM_1)$ is of deterministic type, the tokens in $\MM_0,\MM_1$ are at address $\ADD$,
	$@^{\mathsf s}E=@^{\mathsf s}F$,  
	and for all $t\supseteq \ADD$, 
	$\config_t=(\pi,\MM_0,m_t)$ and 
	$\donfig_t=(\pi,\MM_1,m'_t)$, and $m'_t$ is obtained from $m_t$ by possibly applying the gate corresponding to the reduction, and $\nu_t=\mu_t$;

	\item $a=(\MM_0,\MM_1,\MM_1)$ is of choice type, corresponding to some gate $\meas^l$ of index $\alpha$, the tokens in $\MM_0,\MM_1$ are at address $\ADD$, 
	$@^{\mathsf s}F=\EXT(E,@^{\mathsf s}E,\ADD,\alpha)$, 
 and for all $t\in @^{\mathsf s}E$ and $b\in \{0,1\}$, $\config_t=(\pi,\MM_0,m_t)$, 
	$\donfig_{t\cup\{(\alpha,b)\}}=(\pi,\MM_1,m'_{t,b})$, where $m'_{t,b}=[L'_t,Q_t^b,v_t::(l:b)]$ is obtained from $m_t=[L_t,Q_t,v_t]$ by applying the $\meas$-gate to the state $Q_t=\lambda_0(Q_{t}^0\otimes\ket 0)+\lambda_1(Q_{t}^1\otimes \ket 1)$ and measuring the value $b$ (notice that $L'_t$ is obtained from $L_t$ by updating $l:\qbit$ into $l:\bit$); correspondingly, we have $\nu_{t\cup\{(\alpha,b)\}}=\lambda_{b}\cdot\mu_t$;
	
	 \item $a=(\MM_0,\MM_1,\MM_2)$ is of pseudo-deterministic type, 
	 the tokens in $\MM_0$ are at address $\ADD$, those in $\MM_{b+1}$ are at address $\ADD\cup\{(l,b)\}$, for $b=0,1$, $\MM_0$ contains one token $(l,\sigma,\ADD)$ in the guard $\bit^l$ of an if-then-else, 
	 $@^{\mathsf s}F=\EXT(E,@^{\mathsf s}E,\ADD,l)$, and
	 for all $s\supseteq \ADD$, 
	$\config_s=(\pi,\MM_0,m_s)$ and 
	 $\donfig_{s\cup\{(l,b)\}}=(\pi,\MM_{1+b},m'_s)$, with $m'_s$ obtained from $m_s$ by removing the
value $(l,b')$ contained in the classical register; finally, $\nu_{s\cup\{(l,b)\}}=\mu_s$ if $b=b'$ and is $0$ otherwise.
	\end{enumerate}

\end{itemize}
\end{definition}

\begin{rem}
Given a parallel reduction step $\mu\mtocl_{\ADD}^a \nu$ of pseudo-deterministic type (with the notation from  above), triggered by an if-then-else with guard $\bit^l$, 
and a super-address $s\in @^{\mathsf s}E$, there are four possible evolutions, depending on the value of the slice choice $(l,b)$ made for the if-then-else and for the value $(l,b')$ that is actually red from the classical register:
\begin{enumerate}
\item $b=b'=0$: in this case, the tokens move onto the left-hand part of the if-then-else (since the value $b'=0$ is read from the register), and 
the slice $s\cup\{(l,b=0)\}$
takes into account those tokens; 
accordingly, the probability of this slice coincides with the probability of the slice $s$.

\item $b=0,b'=1$: in this case, the tokens move onto the right-hand part of the if-then-else (since the value $b'=1$ is read from the register), but 
the slice $s\cup\{(l,b=0)\}$ only
takes into account the tokens moving onto the left-hand part (yet no token has any chance to get there); accordingly, the probability of this slice is $0$.

\item $b=b'=1$: similar to case 1.
\item $b=1,b'=0$: similar to case 2.
\end{enumerate}
\end{rem}

\begin{rem}[pseudo-distribution $a_{\ADD}(\mu)$]
For fixed $a$ and $\ADD$, the relation $\mu\mtocl^a_{\ADD}\nu$ is in fact \emph{functional}:
for all $\nu'$ such that $\mu\mtocl^a_{\ADD}\nu'$, it must be $\nu=\nu'$. In other words, the data $\mu,a,\ADD$ uniquely determines the produced pseudo-distribution $\nu$, that we will henceforth indicate as $a_{\ADD}(\mu)$.
\end{rem}

We now establish the termination and soundness of parallel reduction:

\begin{proposition}[termination of parallel $\evaltkm$-reduction]\label{thm:paraterm}
There exists no infinite chain $\mu_0\mtocl^{a_1}_{\ADD_1}\dots \mtocl^{a_{n-1}}_{\ADD_{n-1}} \mu_n \mtocl^{a_n}_{\ADD_n}\dots$.
\end{proposition}
\begin{proof}
Let $\mu\mtocl^a_{\ADD}\nu$, with the notations from Definition \ref{def:parallelred};
for all $s\in @^{\mathsf s}E$, there exists a configuration $\config_s$ and finitely many configurations $\donfig_{1},\dots, \donfig_{n}$, such that $\config_s\to [\donfig_{1}^{\rho_1},\dots, \donfig_n^{\rho_n}]$, such that $\sum_i\rho_i=\mu_s$. In particular, if $\mu_s>0$, then there exists $i=1,\dots, n$ such that $\rho_i>0$. 
Suppose now there exists an infinite chain $\mu_0\mtocl^{a_1}_{\ADD_1}\dots \mtocl^{a_{n-1}}_{\ADD_{n-1}} \mu_n \mtocl^{a_n}_{\ADD_n}\dots$. One could then construct (via K\"onig's Lemma) an infinite chain of $\evaltkm$-reductions, each happening with probability $>0$, hence contradicting the termination of the $\evaltkm$ (\cite[Theorem 10]{DalLago2017}).
\end{proof}

To establish soundness we need the following definition:
\begin{definition}[pruned distribution]
For any pseudo-distribution $\mu=\{x_1^{\mu_1},\dots, x_n^{\mu_n}\}$, let $\widetilde{\mu}$, called the \emph{pruned distribution} of $\mu$, be the distribution $\widetilde\mu=\{x_i^{\mu_i}\in \mu\mid\mu_i>0\}$.
\end{definition}

In the result below, we implicitly assume a pruned variant of the reduction of the $\evaltkm$, i.e.~given by rules of the form $\config\to \{\config_1^{\mu_1},\dots, \config_n^{\mu_n}\}$, where, for all $i=1,\dots, n$, $\mu_i>0$. In other words, we remove from the reduction all configurations that have probability $0$.

\begin{proposition}[soundness and completeness of parallel $\evaltkm$-reduction]\label{prop:soundpara}
For any chain of parallel reductions
$
\{\config^1 \} \mtocl^{a_1}_{G_1} \mu_1 \mtocl^{a_2}_{G_2}\dots \mtocl^{a_{n-1}}_{G_{n=1}}\mu_{n-1} \mtocl^{a_n}_{G_n}\mu_n$ there exists a (pruned) $\evaltkm$-reduction $\config\to^*\widetilde{\mu_n}$.
Conversely, for any (pruned) $\evaltkm$-reduction $\config\to^*\mu$ there exists a 
chain of parallel reductions
$
\{\config^1 \} \mtocl^{a_1}_{G_1} \mu_1 \mtocl^{a_2}_{G_2}\dots \mtocl^{a_{n-1}}_{G_{n=1}}\mu_{n-1} \mtocl^{a_n}_{G_n}\mu_n$ such that $\mu=\widetilde{\mu_n}$.

\end{proposition}
\begin{proof}
Both claims follow, by induction on $n$, from the observation that a parallel reduction $\{\config^1\}\mtocl^a_{\ADD}\mu$ can always be converted into a pruned $\evaltkm$-reduction $\config\to\widetilde{\mu_n}$ and conversely.
\end{proof}

\subsection{The Bisimulation Theorem}\label{subsec:bisi}

We introduce a family of relations between coherent distributions for
the \evaltkm\ and \circuittkmz\ configurations that we will prove to be a bisimulation between $\circuittkmz$-reduction and parallel probabilistic $\evaltkm$-reduction.

We can now introduce our bisimulation relation.

\begin{definition}\label{def:Rmc}
Let $\pi$ be a derivation and $E,F$ be extended circuits. Let $m=[L,Q,v]$ be a quantum+classical register for the initial state of
the \evaltkm over $\pi$. For all $E$-indexed coherent pseudo-distributions $\mu=\{\config_s^{\mu_s}\}_{s\in @^{\mathsf s}E}$ of 
\evaltkm\ configurations, where $\config_s=(\pi,  \MM_s, m_s)$, with $m_s=[L_s,Q_s,v_s]$, and
for all \circuittkmz-configurations $\config=(\pi,  \MM, F)$, the relation 
 $
 \mu \ \mathscr R^{m}_{\pi} \ \config
 $
holds when:
\begin{description}
\item[(index)] $E=F$;
	
\item[(token)] for all $s\in @^{\mathsf s}E$, $\MM_s=\MM|_{\ADD_s}$;

\item[(mixed state)] 
For all $s\in @^{\mathsf s} E$, $\mu_s\cdot \tocpm(m_s)= E^{\sharp s} ( \tocpm(m))$.
\end{description}
\end{definition}

The (index) condition just states that the $\evaltkm$-configurations are indexed by the super-addresses of the circuit in the $\circuittkmz$-configuration. 
The (token) condition states that the token positions at some address $\ADD$ in the $\circuittkmz$ coincide with the token positions in \emph{all} configurations $\MM_s$ of the $\evaltkm$ such that $\ADD_s= \ADD$ (i.e.~all configurations for the same slice). Recall that there may be many such configurations, all differing from each other due to different outcomes of measurements.
Observe that, as a consequence of the (token) condition, since $m_s=[L_s,Q_s,v_s]$ satisfies $|L_s|=\mathbf{ENV}(\MM_s)$, we have that the circuit $E$ has the typing $\mathbf{ENV}(\negative_\pi)\rhd E\rhd \mathcal J_{\MM}$, where the environment $\mathcal J_{\MM}@\ADD_s$ coincides with $\mathbf{ENV}(\MM|_{\ADD_s})=\mathbf{ENV}(\MM_s)$: in other words, the output wires of $E@\ADD_s$ correspond to the entries of the register $m_s$.
 Finally, the (mixed state) condition states that, for any super-address $s\in @^{\mathsf s}E$, the mixed state corresponding to the quantum+classical register at slice $s$ multiplied by its probability $\mu_s$ coincides with the state produced by applying the extended circuit $E$, sliced at $s$, to the initial quantum+classical register $m$. 

We will now show that these conditions are precisely the good invariants for the execution of \emph{both} machines.

\begin{theorem}[bisimulation between the \evaltkm and the \circuittkmz]\label{thm:noif-then-else-simulation} For any extended circuit $E$, $E$-indexed coherent
distribution of \evaltkm\ configurations $\mu$, any
\circuittkmz\ configuration $\config$, for any quantum+classical register
$m$, reduction type $a$ and address $\ADD\in @E$,
\begin{enumerate}
\item $\mu \ \mathscr R^{m}_{\pi} \ \config\ \text{and} \ \config\to^a_{\ADD}\config' \ \Rightarrow \ \exists\mu'\ \text{s.t.} \ \mu\mtocl^a_{\ADD}\mu'\ \text{and}\
\mu' \ \mathscr R^{m}_{\pi}\ \config'.
$
\item $
\mu \ \mathscr R^{m}_{\pi} \ \config\ \text{and} \ \mu\mtocl^a_{\ADD}\mu'\
 \Rightarrow \ \exists\config'\ \text{s.t.} \  
\config\to^a_{\ADD}\config'\  \text{and}\
\mu' \ \mathscr R^{m}_{\pi}\ \config'.
$
\end{enumerate}
\end{theorem}
\begin{proof}
Since $\config'=a_{\ADD}(\config)$ and $\mu'=a_{\ADD}(\mu)$ are uniquely determined by $a$ together with $\ADD$, $\config$ and $\mu$, respectively, we will actually prove the following statements:
\begin{itemize}
\item[(a.)] $\mu \ \mathscr R^{m}_{\pi} \ \config\ \Rightarrow \ \Big( \config\to^a_{\ADD} a_{\ADD}(\config) \ \Leftrightarrow \  \mu\mtocl^a_{\ADD} a_{\ADD}(\mu)\Big)$;
 \item[(b.)] $\mu \ \mathscr R^{m}_{\pi} \ \config\  \Rightarrow \  a_{\ADD}(\mu) \ \mathscr R^{m}_{\pi}\ a_{\ADD}(\config)$,
  provided that $a_{\ADD}(\mu)$ and $a_{\ADD}(\config)$ are well-defined.
\end{itemize}

Let us first prove the statement (a.), so suppose $\mu \ \mathscr R^{m}_{\pi} \ \config$ and let $\mu=\{\config_s^{\mu_s}\}_{s\in @^{\mathsf s}E}$, where 
$\config_s=\{\pi, \MM_s,m_s\}$, and 
$\config=(\pi, \MM, E)$.
\begin{description}

\item[($\config\to^a_{\ADD} a_{\ADD}(\config) \ \Rightarrow \  \mu\mtocl^a _{\ADD} a_{\ADD}(\mu)$)]
Suppose reduction $a$ applies to some set of tokens with underlying address $\ADD$;
since $\mu$ is $E$-indexed, 
by the (token) condition, for all configurations $\config_s$, with $\ADD_s=\ADD$, the tokens $\MM_s$ are in the same position as in $\MM|_{\ADD}$; 
 hence, the corresponding $\evaltkm$-reduction rule $a$ applies to all such $\config_{s}$.

\item[($ \mu\mtocl^a_{\ADD} a_{\ADD}(\mu) \ \Rightarrow \ \config\to^a_{\ADD} a_{\ADD}(\config)$)]
Suppose reduction $a$ applies to all configurations $\config_s$ such that $\ADD_s= \ADD$. By the (token) condition $\MM_s=\MM|_{\ADD}$, all such tokens exist with same position and address in $\MM$ at address $\ADD$. Hence, the corresponding $\circuittkmz$-reduction rule $a$ applies to $\config$ at address $\ADD$.

\end{description}

The proof of statement (b.) is more delicate and requires checking all reduction rules.
Suppose again $\mu \ \mathscr R^{m}_{\pi} \ \config$ and let $\mu$ and $\config$ be defined as above. Let $\mu':=a_{\ADD}(\mu)=\{\donfig_t^{\mu'}\}_{t\in @^{\mathsf s}F}$, where $\donfig_t=(\pi,\MM_t,m'_t)$, $m'_t=(L'_t,Q'_t,v'_t)$, and $\config':=a_{\ADD}(\config)=(\pi,\MM',F)$.
We will check that, in any possible case, $\mu' \mathscr R^{m}_{\pi}\config'$ must hold.

The claim is easily checked for all structural rules: the corresponding reduction type is deterministic, so $\mu'=a_{\ADD}(\mu)$ is indexed over $@^{\mathsf s}F=@^{\mathsf s}E$, and moreover no change to the quantum+classical register is produced.

Instead, we need to consider in more detail the \evaltkm\ transitions corresponding to (1) traversing
 a unitary gate $c$, a \new gate or a \meas gate, and (2) updating addresses once a token reaches the guard of an if-then-else rule.

Let us start with (1). Let $c$ be either unitary or in $\{\new,\meas\}$. 
We have $F=E[E@\ADD \mapsto (E@ \ADD);c]$ and 
$\MM'$ differs from $\MM$ in that $r$ tokens placed at address $\ADD$ moved from the input to the output positions of gate $c$.
Given $r$ qubits, a set $I\subset \{1,\dots, r\}$, and a gate $c$ over $r$ qubits, we indicate as ${c}^I $ the gate that applies $c$ to the qubits indexed by $I$ and is the identity on all other qubits.

We consider the three
cases for $c$ separately:
\begin{itemize}

\item \emph{unitary gate $c^{\vec l}_{\vec r}=U$ applied to the qubits of labels $I=\{l_1,\dots, l_k\}$}: the reduction has deterministic type, so 
$\mathcal L(F)=\mathcal L(E)$ and $@^{\mathsf s}F=@^{\mathsf s}E$; 
$\mu'$ is given by $\donfig_s=\config_s$ for all $s\not\supseteq\ADD$, while for $s\supseteq \ADD$, $\donfig_{s}=
(\pi,  \MM'_{s}, m'_{s})$, and 
$\MM'_{s}$ is obtained by moving the tokens situated at the input positions of $U$ towards the corresponding output positions (here relying on the (token) condition that holds by hypothesis between $\MM_s$ and $\MM|_{\ADD}$).

The (index) condition holds since $\mu'$ is indexed by $@^{\mathsf s}F=@^{\mathsf s}E$. 
From the (token) condition for $\mu$ and $\config$ one also immediately deduces the (token) condition for $\mu'$ and $\config'$.

Let us look at the (mixed state) condition:   
$m'_{s}=[L'_s,Q'_{s},v_{s}]$ has
$L'_s$ defined by replacing $l_1,\dots, l_k$ by $r_1,\dots, r_k$ in $L_s$ and  
$Q'_{s}={U}^{I}Q_{s}$. The probabilities $\mu'_s$ are given simply as $\mu'_s=\mu_s$. 

Observe that what we write $U^IQ_s$ is in fact the result of applying $U$ to the qubits that the register $L_s$ associates with the labels $l_1,\dots, l_k$. In other words, in terms of mixed states, it corresponds to actually computing $\varphi_{L'_s}^{-1}\circ\interp{Q^I}\circ\varphi_{L_s}$, since we have:
\[
\begin{tikzpicture}
	\begin{pgfonlayer}{nodelayer}
		\node [style=none] (0) at (-2.25, 0.75) {$\interp{(L_s)_{\qbit}\otimes(L_s)_{\bit}}$};
		\node [style=none] (1) at (-2.25, -0.75) {$\interp{(L'_s)_{\qbit}\otimes(L_s)_{\bit}}$};
		\node [style=none] (2) at (1.55, 0.75) {$\interp{L}$};
		\node [style=none] (3) at (1.55, -0.75) {$\interp{L'_s}$};
		\node [style=none] (4) at (-2.175, 0.5) {};
		\node [style=none] (5) at (-2.175, -0.5) {};
		\node [style=none] (6) at (1.5, 0.5) {};
		\node [style=none] (7) at (1.5, -0.5) {};
		\node [style=none] (8) at (-0.75, 0.75) {};
		\node [style=none] (9) at (1.25, 0.75) {};
		\node [style=none] (10) at (-0.75, -0.75) {};
		\node [style=none] (11) at (1.25, -0.75) {};
		\node [style=none] (12) at (0.25, 1) {$\varphi_{L_s}$};
		\node [style=none] (13) at (1.95, 0) {$\interp{U^I}$};
		\node [style=none] (14) at (0.25, -1) {$\varphi_{L'_s}$};
		\node [style=none] (15) at (-2.775, 0) {''$U^I Q_s$''};
	\end{pgfonlayer}
	\begin{pgfonlayer}{edgelayer}
		\draw [style={arrows={->[]}}] (4.center) to (5.center);
		\draw [style={arrows={->[]}}] (8.center) to (9.center);
		\draw [style={arrows={->[]}}] (10.center) to (11.center);
		\draw [style={arrows={->[]}}] (6.center) to (7.center);
	\end{pgfonlayer}
\end{tikzpicture}
\]

We thus have that 
\begin{align*}
\tocpm(m'_s)&=\varphi_{L'_s}(\tocpm(Q'_s)\otimes \tocpm(v_s))\\
&=
\varphi_{L'_s} \Big( \big(\varphi_{L'_s}^{-1}\circ \interp{U^I}\circ \varphi_{L_s}\big) \big(\tocpm(Q_s)\otimes \tocpm(v_s)\big)\Big )\\
&=
\varphi_{L'_s}\circ \varphi_{L'_s}^{-1}\circ \interp{U^I}\circ \varphi_{L_s} \big(\tocpm(Q_s)\otimes \tocpm(v_s)\big)\\
&=
\interp{U^I} \Big(\varphi_{L_s} \big(\tocpm(Q_s)\otimes \tocpm(v_s)\big)\Big)=
\interp{U^I} (\tocpm(m_s)).
\end{align*}

By hypothesis, we know that for all $s\in @^{\mathsf s} E$,
$\mu_s\cdot \tocpm(m_s)=E^{\sharp s}(\tocpm(m))$.
The condition $\mu_s\cdot\tocpm(m_s)=f^{\sharp s}(\tocpm(m))$ then obviously holds for all $s\not\supseteq\ADD$.
For $s\supseteq\ADD$ we have, using Lemma \ref{lemma:help1}:
\begin{align*}
\mu'_s\cdot\tocpm(m'_{s})
& = \mu_s\cdot (\interp{U^{I}}(
\tocpm(m_{s})))=
\interp{U^{I}} (\mu_s\cdot \tocpm(m_{s}))
\stackrel{\text{\tiny[Hyp]}}{=}  
 \interp{U^{I}}\left({E^{\sharp s}} (\tocpm(m))\right)\\
& \ \ \ {=} \ \  \left(\interp{U^{I}}\circ E^{\sharp s}\right) (\tocpm(m))
\stackrel{\text{\tiny{[Lemma \ref{lemma:help1} (i.)]}}}{=}  
F^{\sharp s} (\tocpm(m)).
\end{align*}

\item \emph{$\new$ gate applied to qubit $I=\{i\}$}: as in the previous case, 
the reduction has deterministic type, $\mathcal L(F)=\mathcal L(E)$ and $@^{\mathsf s}E=@^{\mathsf s}F$; 
 $\mu'$ is given by $\donfig_s=\config_s$ for all $s\not\supseteq\ADD$, while
for all $s\in \supseteq\ADD$, $\donfig_{s}=
(\pi,  \MM'_{s}, m'_{s})$, with
$\MM'_{s}$ differing from $\MM_{s}$ in that its unique token placed at the input of the $\new$-gate (which, as in the case above, exists in $\MM_s$ by hypothesis) is moved from the input to the output position of $\new$ (yielding one deterministic reduction step).
The (index) and (token) condition can then be checked similarly to the previous case.

Let us check the (mixed condition). We have $m'_{s}=[L'_s,Q'_{s},v'_{s}]$, where $L'_s$ differs from $L_s$ in that the $i$-th label $l:\bit$ is changed to $l:\qbit$, and moreover 
we have $Q'_{s}=Q_s\otimes \ket{b}$ and $v'_s$ is obtained from $v_s$ by removing the value $(l,b)$;
the associated probabilities are given simply as 
$\mu'_s=\mu_s$. 

By hypothesis, we know that for all $s\in @^{\mathsf s} E$,
$\mu_s\cdot\tocpm(m_s)=E^{\sharp s}(\tocpm(m))$.
We only need to check that $\mu_s\cdot\tocpm(m_s)=F^{\sharp s}(\tocpm(m))$ holds for $s\supseteq\ADD$.
First observe that
$$
\interp{\new^I} : V_{n}\otimes V_{(1,1)}\otimes V_{(1,\dots, 1)} \ \longrightarrow \  V_{n}\otimes V_2 \otimes V_{(1,\dots, 1)}
$$
that is, up to isomorphisms, we have that
$$
\interp{\new^I} : \prod_{i=1}^{2^{m+1}}\complex^{n\times n}\simeq \Big( \prod_{i=1}^{2^{m}}\complex^{n\times n} \times  \prod_{i=1}^{2^{m}}\complex^{n\times n}\Big)\ \longrightarrow \
\prod_{i=1}^{2^m}\complex^{n\times n}\otimes \complex^{2\times 2}
$$
is given pointwise by
\begin{align*}
&\interp{\new^I}((x_1,\dots, x_{2^m}),(y_1,\dots, y_{2^m}) ) \\
&= (x_1\otimes  \ket{0}\ket{0}^\dag,\dots, x_{2^m}\otimes  \ket{0}\ket{0}^\dag)+ (y_1\otimes \ket{1}\ket{1}^\dag,\dots, y_{2^m}\otimes \ket{1}\ket{1}^\dag).
\end{align*}
Using this, we can deduce that, calling $\iota_{v_s}$ the morphism $\iota_{v_s}:\interp{(L_s)_{\qbit}}\to\interp{L_s}$,
$$
\interp{\new^I} \left( \tocpm(m_s) \right)
=
\interp{\new^I}  \left (\iota_{v_s} (\tocpm(Q_s)) \right)=
 \iota_{v'_s}(\tocpm(Q'_s))=\tocpm(m'_s).
$$
By reasoning similarly to what we did for $\interp{U^I}$, 
we can then compute, using Lemma \ref{lemma:help1}:
\begin{align*}
\mu_s\cdot\tocpm(m'_s)&=
\mu_s\cdot
\interp{\new}^I\left (
 \tocpm(m_s) \right)=
\interp{\new}^I\left ( \mu_s\cdot
 \tocpm(m_s) \right)
 \stackrel{\text{\tiny [Hyp]}}{=}
\interp{\new^I} \Big(E^{\sharp s}(\tocpm(m))\Big)\\
&=\Big(\interp{\new^I}\circ E^{\sharp s}\Big)(\tocpm(m))
\stackrel{\tiny\text{[Lemma \ref{lemma:help1} (i.)]}}{=}F^{\sharp s}(\tocpm(m)).
\end{align*}

\item \emph{\meas gate of index $\alpha$ applied to qubit $I=\{i\}$}: the reduction has choice type, $\mathcal L(F)=\mathcal L(E)\cup\{\alpha\}$ and $@^{\mathsf s}F=\EXT(E,@^{\mathsf s}E,\ADD,\alpha)$. 
For all $s\in @^{\mathsf s}E$ with $\ADD_s= \ADD$, the configuration $\config_s$ gives rise to two configurations $\donfig_{s\cup\{(\alpha,0)\}},\donfig_{s\cup{(\alpha,1)}}$ in $\mu'$ (notice that $\ADD_{s\cup\{(\alpha,b)\}}=\ADD_s=\ADD$).
For $t=s\cup\{(\alpha,b)\}$, the configuration $\donfig_{t}=(\pi,\MM'_{t},m'_t)$ is defined as follows:
$\MM'_t$ is obtained from $\MM_s$ by moving the token placed in the input position of the $\meas$-gate towards the output position; 
in
$m'_t=[L'_t,Q'_t,v'_t]$ 
$L'_t$ is obtained from $L_s$ by changing the output label $l$ of the $\meas$-gate from $l:\qbit$ to $l:\bit$, 
$Q'_t$ is defined below, and $v'_t$ is obtained from $v_s$ by  
adding the value $(l,b)$.

The (index) and (token) conditions are then clearly satisfied by construction.

Let us check the (mixed state) condition.
As before, by hypothesis, we know that for all $s\in @^{\mathsf s} E$,
$\mu_s\cdot \tocpm(m_s)=E^{\sharp s}(\tocpm(m))$.
Let $s\in @^{\mathsf s}E$ be such that $\ADD_s=\ADD$, fix $b\in \{0,1\}$ and let $t=s\cup\{(\alpha,b)\}$; let us look more closely at how the quantum register and the assigned probabilities for the configuration $\config'_t$ are defined. Recall that each quantum state $Q_s\in V_n\otimes V_2$ can be written as
$Q_s=\alpha_0Q_{s,0}+\alpha_1Q_{s,1}$, where
$Q_{s,\ell}=\sum_j \gamma_{j,\ell} \ket{\phi_j^{\ell}}\otimes \ket{\ell}$. 
We then have $Q'_{t}:=\sum_j \gamma_{j,b} \ket{\phi_j^{b}}$ and $\mu'_{t}=\mu_s\cdot |\alpha_{b}|^{2}$.

We can now compute
\begin{align}\label{eqqiqi}
  \begin{split}
\tocpm(Q_s)&=
|\alpha_0|^2 \cdot Q_{s,0}Q_{s,0}^{\dag} + 
\alpha_0\alpha_1^{\dag} \cdot Q_{s,0}Q_{s,1}^{\dag} + 
\alpha_1\alpha_0^{\dag} \cdot Q_{s,1}Q_{s,0}^{\dag} +  
|\alpha_1|^2 \cdot Q_{s,1}Q_{s,1}^{\dag}.
\end{split}
\end{align}
Observe that
\begin{align*}
|\alpha_b|^2 \cdot Q_{s,b}Q_{s,b}^{\dag}&=|\alpha_b|^2 \cdot \sum_{j,j'} \gamma_{j,b}\gamma_{j',b}^\dag\cdot\left (\ket{\phi_j^{b}}\bra{\phi_{j'}^{b}}\otimes \ket{b}\bra{b}\right)\\ &=\sum_{j,j'} \gamma_{j,b}\gamma_{j',b}^\dag\cdot\left (\ket{\phi_j^{b}}\bra{\phi_{j'}^{b}}\otimes |\alpha_b|^2\cdot \ket{b}\bra{b}\right).
\end{align*}
This means that $|\alpha_0|^2 \cdot Q_{s,0}Q_{s,0}^{\dag}$ and $|\alpha_1|^2 \cdot Q_{s,1}Q_{s,1}^{\dag}$ are respectively of the form
\begin{equation}\label{eqalpha}
\sum_k \delta^{0}_k\cdot \left( x^{0}_k\otimes
\begin{pmatrix}
|\alpha_0|^2
&
0
\\
0&
0
\end{pmatrix} \right)\qquad 
\sum_k \delta^{1}_k\cdot \left( x^{1}_k\otimes
\begin{pmatrix}
0
&
0
\\
0&
|\alpha_1|^2
\end{pmatrix} \right)
\end{equation}
where the $\delta_k^b$ span all scalars
$\gamma_{j,b}\gamma_{j',b}^\dag$ and the $x_k^b$ span all vectors
$\ket{\phi_j^b}\bra{\phi_{j'}^b}$. Furthermore, notice that we have:
\begin{equation}\label{eqalpha1}
\begin{split}
&\sum_k\delta^{b}_k\cdot x^{b}_k 
=\sum_{j,j'}\gamma_{j,b}\gamma_{j',b}^\dag\cdot \ket{\phi_j^b}\bra{\phi_{j'}^b} 
=(Q'_{t})(Q'_{t})^\dag 
=\tocpm(Q'_{t}).
\end{split}
\end{equation}

Similarly, $\alpha_0\alpha_1^{\dag} \cdot Q_{s,0}Q_{s,1}^{\dag}$ and $
\alpha_1\alpha_0^{\dag} \cdot Q_{s,1}Q_{s,0}^{\dag}$ are respectively of the form
\begin{equation}\label{eqalpha2}
\sum_k \delta^{01}_k\cdot \left( x^{01}_k\otimes
\begin{pmatrix}
0
&
\alpha_0\alpha_1^{\dag}
\\
0&
0
\end{pmatrix} \right)\qquad
\sum_k \delta^{10}_k\cdot \left( x^{10}_k\otimes
\begin{pmatrix}
0
&
0
\\
\alpha_1\alpha_0^{\dag}&
0
\end{pmatrix} \right)
\end{equation}

Now consider the gate $\meas$. We have that
$$
\interp{\meas^I}:
 V_{n}\otimes V_2 \otimes V_{(1,\dots, 1)}
 \ \longrightarrow \
V_{n}\otimes V_{(1,1)}\otimes V_{(1,\dots, 1)}
$$
that is
$$
\interp{\meas^I}:
\prod_{i=1}^{2^m}\complex^{n\times n}\otimes \complex^{2\times 2}
\ \longrightarrow \
\prod_{i=1}^{2^{m+1}}\complex^{n\times n}\simeq \Big( \prod_{i=1}^{2^{m}}\complex^{n\times n} \times  \prod_{i=1}^{2^{m}}\complex^{n\times n}\Big)
$$
and is given by
\begin{align*}
&\interp{\meas^I}: \left( x_1\otimes \begin{pmatrix}a_1 & b_1 \\ c_1& d_1 \end{pmatrix}, \dots,  x_{2^m}\otimes \begin{pmatrix}a_{2^m} & b_{2^m} \\ c_{2^m}& d_{2^m} \end{pmatrix} \right) 
\mapsto
\left(
\left(
a_1 x_1,\dots, a_{2^m}x_{2^m}
\right),
\left(
d_1x_1, \dots, d_{2^m}x_{2^m}
\right)
\right).
\end{align*}
In particular, we have
\begin{align*}
&\interp{\meas^I}: \left(0, \dots,0, x\otimes \begin{pmatrix}a & b \\ c& d \end{pmatrix},0 ,\dots, 0 \right)
\mapsto
((0,\dots,0, a\cdot x, 0,\dots, 0), (0,\dots,0, d\cdot x,0,\dots, 0))
\end{align*}
that is, letting $\iota_{v_s}$ the morphism $\iota_{v_s}:\interp{(L_s)_\qbit}\to\interp{L_s}$, we have 
\begin{align*}
\interp{\meas^I}\left( \iota_{v_s}\left(x\otimes \begin{pmatrix}a & b \\ c& d \end{pmatrix}\right)\right)& =  a\cdot\iota_{v_s\cup\{(\alpha,0)\}}(x)+d\cdot\iota_{v_s\cup\{(\alpha,1)\}}(x).
\end{align*}

Using \eqref{eqqiqi}, \eqref{eqalpha} and \eqref{eqalpha1} we thus have
\begin{align*}
\interp{\meas^I}\left( \iota_{v_s}\left( \tocpm(Q_i)\right)\right) & = |\alpha_0|^2\cdot \iota_{v_s\cup\{(\alpha,0)\}}\left( \sum_k \delta_k^0\cdot x_k^0\right)+
|\alpha_1|^2\cdot \iota_{v_s\cup\{(\alpha,1)\}}\left( \sum_k \delta_k^1\cdot x_k^1\right) \\
&= |\alpha_0|^2\cdot \iota_{v_s\cup\{(\alpha,0)\}}\left( \tocpm(Q'_{s\cup\{(l,0)\}})\right)+
|\alpha_1|^2\cdot \iota_{v_s\cup\{(\alpha,1)\}}\left( \tocpm(Q'_{s\cup\{(l,1)\}})\right).
\end{align*}
\noindent and thus, since $\mu'_{t}=\mu_s|\alpha_b|^2$, we have 
\begin{align*}
\mu_s\cdot
\interp{\meas^I} \left(
\tocpm(m_s) \right)&=
\mu_s\cdot
\interp{\meas^I} \left( \iota_{v_s}(\tocpm(Q_s)) \right)
\\
&=
\sum_{\ell
=0,1} \mu'_{s\cup\{(\alpha,\ell)\}}\cdot \iota_{v_s\cup\{(\alpha,\ell)\}}(\tocpm(Q'_{s\cup\{(l,\ell)\}}))\\
&=
\sum_{\ell
=0,1} \mu'_{s\cup\{(l,\ell)\}}\cdot ( \tocpm(m'_{s\cup\{(\alpha,\ell)\}})).
\end{align*}
By reasoning as we did for the case of $\interp{U^I}$, we can thus finally compute, using Lemma \ref{lemma:help1},
\begin{align*}
\mu'_t\cdot
\tocpm(m'_t)&=
\pi^l_b\left(\mu_s\cdot \interp{\meas^I}\left(
\tocpm(m_s)\right)\right)
=
\left(\pi^l_b\circ \interp{\meas^I}\right) \left(\mu_{s} \cdot\tocpm(m_s) \right)\\
& \stackrel{\text{\tiny [Hyp]}}{=}
\left(\pi^l_b\circ\interp{\meas^I} \right)\Big(E^{\sharp s} (\tocpm(m))\Big)
=\Big(\pi^l_b\circ\interp{\meas^I}\circ E^{\sharp s}\Big)(\tocpm(m))\\
&
\stackrel{\tiny\text{[Lemma \ref{lemma:help1} (ii.)]}}{=}
F^{\sharp t}(\tocpm(m)).
\end{align*}
\end{itemize}

It remains to consider the case (2) of a token entering as first in the guard of an if-then-else. Say the corresponding occurrence of $\bit$ is labeled by $l$.

We have a reduction of pseudo-deterministic type, with $\mathcal L(F)=\mathcal L(E)\cup\{l\}$ and $@^{\mathsf s}F=\EXT(E,@^{\mathsf s}E,\ADD,l)$. 
For all $s\in @^{\mathsf s}E$ with $\ADD_s=\ADD$, we have two corresponding configurations $\donfig_{s\cup\{(l,b)\}}$, with $b\in 0,1$. Letting $t=s\cup\{(l,b)\}$, we have 
$\donfig_t=(\pi, \MM'_{t},m_t)$, where $\MM'_t=\{(l',\sigma,\ADD\cup\{(l,b)\})\mid (l',\sigma,G)\in \MM_s\}$ and $m'_{t}=(L'_{t},Q'_t,v'_t)$, where $L'_t=L_s,Q'_t=Q_s$ and $v'_t$ is obtained from $v_s$ by removing the value $(l,b')$.
The assigned probability $\mu_t$ is $\mu_s$ if $b=b'$ and $0$ otherwise.

The (index) condition is verified by construction. 
For the (token) condition, observe that, for all $u\in @^{\mathsf s}F$, with $u\supseteq \ADD$, letting $u^*=u-\{(l,u(l))\}$, 
and noticing that $\ADD_u=\ADD_{u^*}\cup\{(l,u(l))\}$, we have
 \begin{align*}
 \MM'_u&=\{(l',\sigma,\ADD\cup\{(l,u(l))\})\mid(l',\sigma, \ADD)\in \MM_{u^*}\}\\
 &\stackrel{\tiny\text{Hyp}}{=}
  \{(l',\sigma,\ADD\cup\{(l,u(l))\})\mid(l',\sigma, \ADD)\in \MM|_{\ADD_{u^*}}\}=\MM'|_{\ADD_u}.
 \end{align*}

It remains to show the (mixed state) condition. 
Let $s\in @^{\mathsf s}E$ with $\ADD_s=\ADD$, let $b\in \{0,1\}$ and $t=s\cup\{(l,b)\}$.
From Lemma \ref{lemma:help2} (ii.) we have $\pi_b(E^{\sharp s}
(x))=F^{\sharp t}(x)$.

We must consider two cases:
\begin{itemize}
\item if $b\neq b'$, from 
$\mu_s\cdot\tocpm(m_{s})=
\iota_{v_s}( {\mu}_s\cdot 
(\tocpm(Q_s))$ and the fact that $v'_t$ contains the value $(l,b')$ we deduce $\pi_{b}(\mu_s\cdot
\tocpm(m_{s}))=0$. 
From the hypothesis $\mu_s\cdot\tocpm(m_s)=E^{\sharp s}(\tocpm(m))$ we deduce then $\pi_b(E^{\sharp s}(\tocpm(m)))=0$. Using Lemma \ref{lemma:help2} (ii.) we then have
\begin{align*}
\mu'_t\cdot
\tocpm(m_{t})&=0
=\pi_b(E^{\sharp s}(\tocpm(m)))=
F^{\sharp t}(\tocpm(m)).
\end{align*}

\item if $b=b'$, 
then, recalling that $v_s=v'_t\cup\{(l,b)\}$, we can first compute
\begin{align}\label{eq:dagg}
  \begin{split}
  \mu'_t\cdot
\tocpm(m'_t)&=\mu_s\cdot \iota_{v'_t}(\tocpm(Q_s)) 
=
\mu_s\cdot\left(\pi^l_{b}\left(
 \iota_{v_s}(\tocpm(Q_s))
\right) \right)= \pi^l_{b}(\mu_s\cdot\tocpm(m_s)).
  \end{split}
\tag{$\ddag$}
\end{align}
Using Lemma \ref{lemma:help2} (ii.) we then have
\begin{align*}
\mu'_t\cdot\tocpm(m'_t)&\stackrel{\tiny\eqref{eq:dagg}}{=}
 \pi^l_b\big(\mu_s\cdot\tocpm(m_s)\big) \stackrel{\tiny\text{[Lemma \ref{lemma:help2} (ii.)]}}{=}\pi^l_b(E^{\sharp s}(\tocpm(m))) = F^{\sharp t}(\tocpm(m)).
\end{align*}
\end{itemize}

\end{proof}

From the Bisimulation Theorem we can now deduce the following result:

\begin{theorem}[Soundness of the \circuittkmz wrt to the \evaltkm]
Let $\pi :\ \vdash M :A$ be a Boolean typing derivation and suppose
that the \circuittkmz{} produces, over $\pi$, the final configuration
$\config=(\pi,\mathbf{PSN}(\MM), E )$. 
Then, 
for any quantum+classical register $m$, 
the \evaltkm over $\pi$ and $m$ produces a 
$E$-indexed pseudo-distribution $\mu=\{\config_s^{\mu_s}\}_{s\in @^{\mathsf s}E}$ of final configurations such that 
$
\interp{E}(\tocpm(m)) =\tocpm(\mu)$.
\end{theorem}
\begin{proof}
By initializing the \evaltkm over $m$, by applying Lemma \ref{lemma:init-config-related} and Theorem
\ref{thm:noif-then-else-simulation} a finite number of times,
we end up producing a pseudo-distribution of final configurations $\mu$ such that $\mu\ \mathscr R^m_\pi \ \config$. 
Since $\config$ is final, we know that $\mathbf{ENV}(\negative_\pi)\rhd E\rhd \mathbf{ENV}(\positive_\pi)$, i.e.~$E$ is uniform. Moreover, we also know that for all $s\in @^{\mathsf s}E$, $|{L_s}|=\mathbf{ENV}(\positive_\pi)$,  where $m_s=[L_s,Q_s,v_s]$, so $\tocpm(m_s)\in \interp{\mathbf{ENV}(\positive_\pi)}$.
We can then compute 
\begin{align*}
\interp{E}(\tocpm(m))
&\stackrel{\tiny\text{Lemma \ref{prop:sumext}}}{=}
\sum_{s\in @^{\mathsf s}E}E^{\sharp s}(\tocpm(m))\\
&\stackrel{\tiny\text{Prop.~\ref{thm:noif-then-else-simulation}}}{=}
\sum_{s\in @^{\mathsf s}E}\nu_s\cdot \tocpm(m_s)\\
&=\tocpm(\mu)\in \interp{\mathbf{ENV}(\positive_\pi)}.
\end{align*}
\end{proof}

\subsection{Termination and Confluence of the \circuittkmz}\label{subsec:termi}

We will exploit Theorem \ref{thm:noif-then-else-simulation} to prove several properties of the \circuittkmz. 

First, recall that the initial configuration of both the \circuittkmz{} and the
\evaltkm{} are similar: their sets of tokens are the same. Moreover, the quantum+classical memory is the one whose labels are in bijection with the
occurrence of $\qbit$ (resp.~$\bit$) of the conclusion of $\pi$, where all the values are set to $\ket 0$ (resp. $0$).

\begin{lemma}[Initial Configurations are Related]
  \label{lemma:init-config-related}
  For any typing derivation $\pi : \Delta \vdash M : A$, letting $m$
  be the register in $\initm(\pi)$, we have
$  \mu\set{\initm(\pi)} \ \mathscr R^{m}_{\pi} \ \initcz(\pi)$.
 \end{lemma}
\begin{proof}
  Each condition is straightforward as the extended circuit is simply the
  identity, the tokens are in the same positions and there is a single address
  in the extended circuit.
\end{proof}

\begin{theorem}[Termination of \circuittkmz]
  \label{thm:termination-circuit-tkm}
    Given a type derivation $\pi : \Delta \vdash M : A$, the \circuittkmz
    terminates in a final
    configuration.
\end{theorem}
\begin{proof}
  By Lemma~\ref{lemma:init-config-related}, $\initm(\pi)$ is related
  to $\config$. Then proceed by absurdity. Assume we have an infinite rewrite
  sequence $\config \to \config_1 \to \dots$. By
  Theorem~\ref{thm:noif-then-else-simulation} (1), we can deduce that there
  exists an infinite chain of reductions $\initm(\pi) \mtocl \mu_1 \mtocl \dots$ of the \evaltkm, where each pseudo-distribution $\mu_i$ is in relation with $\config_i$. 
  This contradicts then Theorem \ref{thm:paraterm}.
  The fact that \circuittkmz must terminate in a final configuration is, again, given
  by the fact that \evaltkm terminates in a final configuration and thanks to
  Theorem~\ref{thm:noif-then-else-simulation}.
\end{proof}

Let us now consider confluence.
\begin{definition}[configuration equivalence]

A $\circuittkmz$-configuration $\config=(\pi,\MM,E)$ is \emph{uniform} if $E$ is uniform.
 For all uniform $\circuittkmz$-configurations
 $\config_1=(\pi_1,\MM_1,E_1)$ and $\config_2=(\pi_2,\MM_2,E_2)$,
  $\config_1 \equiv
\config_2$ holds whenever $\pi_1=\pi_2$ and $\interp{E_1}=\interp{E_2}$.
\end{definition}

\begin{lemma}
  \label{lemma:config-related-same-distri-same-sem}
  Let $\config=[\pi,\MM,E], \config'=[\pi,\MM',E']$ be two uniform \circuittkmz configurations over $\pi$, with $@^{\mathsf s}E=@^{\mathsf s}E'$, and let $\mu$ 
  be a coherent pseudo-distribution of \evaltkm configurations. If, for all quantum+classical register $m = [ L,Q,v]$ for $\pi$, we have $\mu \mathscr R^{m}_\pi \config$ and $\mu \mathscr R^{m}_\pi
  \config'$, then $\config\equiv\config'$.
\end{lemma}
\begin{proof}
  Let $\config = (\pi, \MM, E)$ and $\config' = (\pi, \MM', E')$. As both
  $\config$ and $\config'$ are in relation with $\mu$, by the definition of
  $\mathscr R^{m}_\pi$ we know that, thanks to the (mixed state) condition,
  that the semantics of $E$ and $E'$, at any slice $s$, is the same for any
  quantum + classical register $m$ and are thus equal by $\interp{-}$ by
  definition of $E^{\#s}$ and Proposition~\ref{prop:sumext}. 
\end{proof}

\begin{theorem}[Confluence Modulo Equivalence]

For all $\circuittkmz$-configurations $\config, \config',\config''$, 
if $\config\to\config'$ and $\config\to\config''$, then there exists uniform configurations $\tilde{\config}'$ and 
$\tilde{\config}''$ such that $\config'\to^*\tilde{\config}'$, $\config''\to^*\tilde{\config}''$ and 
$\tilde{\config}'\equiv \tilde{\config}''$.
\end{theorem}
\begin{proof}
  It suffices to observe that the following diagram commutes for all quantum+classical register $m$, where $\mu \simeq
  \config$ is shorthand for $\mu\mathscr R^m_\pi \config$:
\begin{center}
  \begin{tikzpicture}
	\begin{pgfonlayer}{nodelayer}
		\node [style=none] (0) at (0, 0.5) {$\config$};
		\node [style=none] (1) at (-1.5, 0.5) {$\config'$};
		\node [style=none] (2) at (1.5, 0.5) {$\config''$};
		\node [style=none] (3) at (-1.5, 1.75) {$\tilde{\config}'$};
		\node [style=none] (4) at (1.5, 1.75) {$\tilde{\config}''$};
		\node [style=none] (5) at (0, -0.5) {$\mu$};
		\node [style=none] (6) at (-1.5, -0.5) {$\mu'$};
		\node [style=none] (7) at (1.5, -0.5) {$\mu''$};
		\node [style=none] (8) at (0, -1.5) {$\tilde{\mu}$};
		\node [style=none] (9) at (-0.25, 0.5) {};
		\node [style=none] (10) at (-1.25, 0.5) {};
		\node [style=none] (11) at (1.25, 0.5) {};
		\node [style=none] (12) at (-1.5, 1.5) {};
		\node [style=none] (13) at (-1.5, 0.75) {};
		\node [style=none] (14) at (1.5, 1.5) {};
		\node [style=none] (15) at (1.5, 0.75) {};
		\node [style=none] (16) at (0.25, 0.5) {};
		\node [style=none] (17) at (-1.25, 1.775) {};
		\node [style=none] (18) at (-1.25, 1.65) {};
		\node [style=none] (19) at (1.25, 1.775) {};
		\node [style=none] (20) at (1.25, 1.65) {};
		\node [style=none] (21) at (0, 1.95) {$\interp{-}$};
		\node [style=none] (22) at (-1.25, -0.475) {};
		\node [style=none] (23) at (-1.25, -0.6) {};
		\node [style=none] (24) at (-0.25, -0.475) {};
		\node [style=none] (25) at (-0.25, -0.6) {};
		\node [style=none] (26) at (0.25, -0.475) {};
		\node [style=none] (27) at (0.25, -0.6) {};
		\node [style=none] (28) at (1.25, -0.475) {};
		\node [style=none] (29) at (1.25, -0.6) {};
		\node [style=none] (30) at (1.425, 0.275) {};
		\node [style=none] (31) at (1.6, 0.275) {};
		\node [style=none] (32) at (1.425, -0.225) {};
		\node [style=none] (33) at (1.6, -0.225) {};
		\node [style=none] (34) at (-0.075, 0.275) {};
		\node [style=none] (35) at (0.1, 0.275) {};
		\node [style=none] (36) at (-0.075, -0.225) {};
		\node [style=none] (37) at (0.1, -0.225) {};
		\node [style=none] (38) at (-1.575, 0.275) {};
		\node [style=none] (39) at (-1.4, 0.275) {};
		\node [style=none] (40) at (-1.575, -0.225) {};
		\node [style=none] (41) at (-1.4, -0.225) {};
		\node [style=none] (42) at (-1.65, -0.725) {};
		\node [style=none] (43) at (-1.4, -0.725) {};
		\node [style=none] (44) at (1.35, -0.725) {};
		\node [style=none] (45) at (1.6, -0.725) {};
		\node [style=none] (46) at (-0.65, -1.225) {};
		\node [style=none] (47) at (-0.4, -1.225) {};
		\node [style=none] (48) at (0.35, -1.225) {};
		\node [style=none] (49) at (0.6, -1.225) {};
		\node [style=none] (50) at (-2.65, 1.75) {$\tilde{\mu}$};
		\node [style=none] (51) at (2.75, 1.75) {$\tilde{\mu}$};
		\node [style=none] (52) at (-1.775, 1.8) {};
		\node [style=none] (53) at (-1.775, 1.625) {};
		\node [style=none] (54) at (-2.275, 1.8) {};
		\node [style=none] (55) at (-2.275, 1.625) {};
		\node [style=none] (56) at (2.3, 1.8) {};
		\node [style=none] (57) at (2.3, 1.625) {};
		\node [style=none] (58) at (1.8, 1.8) {};
		\node [style=none] (59) at (1.8, 1.625) {};
		\node [style=none] (60) at (0.4, 0) {$H$};
		\node [style=none] (61) at (2.5, 0) {Theorem~\ref{thm:noif-then-else-simulation}};
		\node [style=none] (62) at (-2.5, 0) {Theorem~\ref{thm:noif-then-else-simulation}};
		\node [style=none] (63) at (2, 2.35) {Theorem~\ref{thm:noif-then-else-simulation}};
		\node [style=none] (64) at (-2, 2.35) {Theorem~\ref{thm:noif-then-else-simulation}};
		\node [style=none] (65) at (0, -2) {Confluence~\cite[Theorem 10]{DalLago2017}};
		\node [style=none] (66) at (0, 2.5) {Lemma~\ref{lemma:config-related-same-distri-same-sem}};
	\end{pgfonlayer}
	\begin{pgfonlayer}{edgelayer}
		\draw [style={arrows={->[]}}] (9.center) to (10.center);
		\draw [style={arrows={->[]}}] (13.center) to (12.center);
		\draw [style={arrows={->[]}}] (15.center) to (14.center);
		\draw [style={arrows={->[]}}] (16.center) to (11.center);
		\draw (17.center) to (19.center);
		\draw (18.center) to (20.center);
		\draw [style={arrows={->[]}}] (25.center) to (23.center);
		\draw [style={arrows={->[]}}] (24.center) to (22.center);
		\draw [style={arrows={->[]}}] (26.center) to (28.center);
		\draw [style={arrows={->[]}}] (27.center) to (29.center);
		\draw (31.center) to (33.center);
		\draw [in=45, out=-150, looseness=1.25] (30.center) to (32.center);
		\draw (35.center) to (37.center);
		\draw [in=45, out=-150, looseness=1.25] (34.center) to (36.center);
		\draw (39.center) to (41.center);
		\draw [in=45, out=-150, looseness=1.25] (38.center) to (40.center);
		\draw [style={arrows={->[]}}] (43.center) to (47.center);
		\draw [style={arrows={->[]}}] (42.center) to (46.center);
		\draw [style={arrows={->[]}}] (44.center) to (48.center);
		\draw [style={arrows={->[]}}] (45.center) to (49.center);
		\draw (53.center) to (55.center);
		\draw [in=-45, out=120, looseness=1.25] (52.center) to (54.center);
		\draw (57.center) to (59.center);
		\draw [in=-45, out=120, looseness=1.25] (56.center) to (58.center);
	\end{pgfonlayer}
\end{tikzpicture}
  \end{center}
\end{proof}
\section{The $\evaltkm(^*)$ Simulates the $\circuittkm$}

We will now establish a simulation result between the $\circuittkm$ and the $\evaltkm$.
While, as we have seen, any transition of the $\circuittkmz$ is (bi)simulated by a corresponding transition of the $\evaltkm$, we will show that any application of the synchronous rule of the $\circuittkm$ can be simulated by a sequence of transitions of the $\evaltkm$.

\subsection{Extended Circuits with if-then-else}\label{subs:excircplus}

In the bisimulation argument for the $\circuittkmz$ and the $\evaltkm$ we restricted our attention to extended circuits that do not contain the constructor $\ite{-}{-}{-}$. This is no more possible now, since this constructor is precisely the one introduced by the synchronous guard rule. 
We thus need to extend our analysis of extended circuits to this new case.

For any extended circuit, let $\LITE(E)$ be the set of labels such that $E$ contains a sub-circuit of the form $\ite{l}{C}{D}$. Notice that $\LITE(E)\cap\LITES(E)=\emptyset$.

\begin{definition}\label{def:circaddbis}[$\Sigma$-path, $\Sigma$-address]
Let $C$ be a circuit and $\Sigma\subseteq \LITE(C)$. A set $L\subseteq \mathcal L(C)$ is called a \emph{$\Sigma$-path} if one of the following hold:
\begin{itemize}
\item $C=c\neq\meas$ and $L=\emptyset$;
\item $C=\meas$ of index $\alpha$ and $L=\{\alpha\}$;
\item $C=D_1;D_2$ and $L=L_1\cup L_2$, where $L_i$ is a path in $D_i$;
\item $C=\ite{l}{D_1}{D_2}$, with $l\in \Sigma$, and either $L=\{l\}\cup L_1$, with $L_1$ a path in $D_1$ or
$L=\{l\}\cup L_2$, with $L_2$ a path in $D_2$.
\item $C=\ite{l}{D_1}{D_2}$, with $l\notin \Sigma$, and $L=L_1$, with $L_1$ a path in $D_1$, or $L=L_2$, with 
 $L_2$ a path in $D_2$.

\end{itemize}

A \emph{$\Sigma$-address for $C$} is a partial function $s:\mathcal L(C)\rightharpoonup \{0,1\}$ whose domain $\mathrm{dom}(s)\subseteq \mathcal L(C)$ is a $\Sigma$-path of $C$, and the following (inductive) condition holds: if 
$C=\ite{l}{C_1}{C_2}$, with $l\in \Sigma$, and $s(l)=b$, then $s-\{(l,b)\}$ is a $\Sigma$-address for $C_b$.

We let $@_\Sigma C$ indicate the set of $\Sigma$-addresses of $C$ and we define $@C:=@_{\LITE(C)}(C)$

For any $s\in @_\Sigma C$, let $\mathrm{ext}(s)=\{t\in @C\mid s\subseteq t\}$, and let the CPM morphism
$C^{\sharp s}$ be defined by
$$
C^{\sharp s}=\sum_{t\in \mathrm{ext}(s)}C^{\sharp t}.
$$
\end{definition}

Proposition \ref{prop:circuitsum} can now be restated in the following form:
\begin{proposition}\label{prop:circuitsumbis}
For any well-typed ite-closed circuit $\Gamma\rhd C\rhd \Delta$ and $\Sigma\subseteq \LITE(C)$, 
$
\interp C=\sum_{s\in @_\Sigma C}C^{\sharp s}.
$
\end{proposition}
\begin{proof}
Observing that, for all $s,s'\in @_\Sigma C$, 
if $s\neq s'$ then $\mathrm{ext}(s)\cap\mathrm{ext}(s')=\emptyset$, 
and that for all $t\in @C$, its restriction $t_\Sigma$ to $\Sigma$ is in $@_\Sigma C$,
we deduce that the sets $\mathrm{ext}(s)$, for $s\in @_\Sigma C$, form a partition of $@C$.
We thus have:
\begin{align*}
\sum_{s\in @_\Sigma C}C^{\sharp s}&=
\sum_{s\in @_\Sigma C}\sum_{t\in \mathrm{ext}(s)}C^{\sharp t}=\sum_{s\in @_ C}C^{\sharp s},
\end{align*}
so we conclude by Proposition \ref{prop:circuitsum}.
\end{proof}

Let us now consider extended circuits. The set of addresses $@E$ of an extended circuit is defined as before.
Super-addresses are now defined relatively to a set $\Sigma\subseteq \LITE(E)$:

\begin{definition}[$\Sigma$-super-addresses]
Let $E$ be an extended circuit and $\Sigma\subseteq \LITE(E)$.
A set $L\subseteq \mathcal L(E)$ is called a \emph{$\Sigma$-path in $E$} if one of the following hold:
\begin{itemize}
\item $E=C$ and $L$ is a $\Sigma$-path in $C$;
\item $E=C\xmapsto{l}(E_1,E_2)$ and $L=L_0\cup\{l\}\cup L'$, where $L_0$ is a $\Sigma$-path in $C$ and $L'$ is a $\Sigma$-path in either $E_1$ or $E_2$.
\end{itemize}

A partial function $s:\mathcal L(E)\to\{0,1\}$ is called a \emph{$\Sigma$-super-address of $E$} if its domain is a $\Sigma$-path $L\subseteq \mathcal L(E)$ and the following (inductive) conditions holds \begin{itemize}
\item if $E=C\xmapsto{l}(E_1,E_2)$, with $l\in \Sigma$, then $s-\{(l,b)\}$ is a $\Sigma$-super-address for $E_b$;
\item for all circuit $C\in \mathrm{Circ}(E)$, the restriction $s|_{\mathcal L(C)}$ is a $\Sigma$-path in $C$.

\end{itemize}
We let $@^{\mathsf s}_\Sigma E$ indicate the set of $\Sigma$-super-addresses of $E$ and define $@^{\mathsf s}E:=@^{\mathsf s}_{\LITE(E)}(E)$

For any $s\in @_\Sigma^*E$, the address $\ADD_s$ is defined as before as $\ADD_s:= s|_{\LITES(E)}\in @E$.
Moreover, let $\mathrm{ext}(s)=\{ t\in @^{\mathsf s}(E)\mid s\subseteq t\}$, and define the CPM morphism $E^{\sharp s}$ as 
$$
E^{\sharp s}=\sum_{t\in \mathrm{ext}(s)}E^{\sharp t}.
$$

\end{definition}

Lemma \ref{lemma:addresses}, Lemma \ref{lemma:superaddress}, 
Lemma \ref{lemma:extequiv}, Proposition \ref{prop:sumext}, Lemma \ref{lemma:help1} and Lemma \ref{lemma:help2} can be reformulated as follows (the proofs being essentially the same):

\begin{lemma}\label{lemma:addressesbis}
For any extended circuit $E$ and $\Sigma\subseteq \LITE(E)$, $@_\Sigma ^*E=@_{\Sigma\cup \LITES(E)} \tau(E)$.
\end{lemma}

\begin{lemma}\label{lemma:superaddressbis}
For all extended circuits $E,F$,  $\Sigma\subseteq \LITE(E)$ and $\Sigma$-address $\ADD\in @_\Sigma E$, if $F$ is obtained from $E$ by replacing the circuit $C:=E@\ADD$ by $C\xmapsto{l}{(\mathrm{id},\mathrm{id})}$, then $@_\Sigma^*F=\EXT(E,@^{\mathsf s}_\Sigma E,\ADD,l)$.
\end{lemma}

\begin{lemma}\label{lemma:extequivbis}
For any uniform extended circuit $E$, $\Sigma\subseteq \LITE(E)$ and super-address $s\in @_\Sigma ^*E$, $E^{\sharp s}=\tau(E)^{\sharp s} $. Moreover, $E^{\sharp s}$ is of the form
\begin{equation}
\label{eq:dag3}
E^{\sharp s}=C_0^{\sharp s}\circ (\pi_{s(l_1)}^{l_1}\circ C_1^{\sharp s})\circ \dots\circ
(\pi_{s(l_n)}^{l_n}\circ C_n^{\sharp s}),
\tag{$\dag$}
\end{equation}
where $l_1,\dots, l_n$ indicates the list (in reversed chronological order) of labels encountered in exploring the tree of $E$ from root to leaves, and 
 $C_0,\dots, C_n$ are the circuits (in reversed chronological order) encountered.
\end{lemma}

\begin{proposition}\label{prop:sumextbis}
For any uniform extended circuit $E$ and $\Sigma\subseteq \LITE(E)$, $\interp{E}=\sum_{s\in @_\Sigma^*E}E^{\sharp s}$.

\end{proposition}

\begin{lemma}\label{lemma:help1bis}
For any extended circuit $E$, $\Sigma\subseteq \LITE(E)$, $s\in @_\Sigma^*E$ and gate $c$, let $E'$ be obtained from $E$ by replacing the circuit $E@\ADD_s$ by $E@\ADD_s;c$. Then
\begin{itemize}
\item[i.] if $c$ is unitary or $c=\new$, then 
$(E')^{\sharp s}=\interp{c^I}\circ E^{\sharp s}$;
\item[ii.] if $c=\meas^l$ of index $\alpha$, then $@^{\mathsf s}E'=\EXT(E,@_\Sigma^*E,\ADD_s,\alpha)$ and 
$(E')^{\sharp( s\cup\{(\alpha,b)\})}=\iota_{b}\circ \pi_{b}\circ \interp{c^I}\circ E^{\sharp s}$.

\end{itemize}
\end{lemma}

\begin{lemma}\label{lemma:help2bis}
Let $ E$ be a well-typed extended circuit, $\Sigma\subseteq \LITE(E)$, 
$s\in @^{\mathsf s}_\Sigma E$ and let $\Gamma'\rhd C\rhd \bit^l,\Delta'$ be $E@\ADD_s$. Let $E'$ be $E$ with $C$ replaced by $C\xmapsto{l}{(\mathrm{id},\mathrm{id})}$. Then 
$@_\Sigma^*E'=\EXT(E,@_\Sigma^*E,\ADD_s,l)$ and for all $b\in\{0,1\}$, $(E')^{\sharp(s\cup(l,b))}=\pi^l_b\circ E^{\sharp s}$.
\end{lemma}

Finally, we have the following new lemma:
\begin{lemma}\label{lemma:optiaddress}
Let $E,E_1,E_2$ be extended circuits with disjoint sets of labels. For all 
addresses $\ADD\in @E$ and fresh label $l$, letting 
$$F:= E\Big [E@\ADD\ \ \mapsto\ \  E@\ADD;\ite{l}{\tau(E_1)}{\tau(E_2)} \Big],$$
then  
$$
@_\emptyset^* F=\{ s\in @_\emptyset^*E\mid \ADD_s\neq \ADD\}\cup \{s\cup t\mid s\in @_\emptyset^*E,\ADD_s=\ADD, t\in @_\emptyset \tau(E_b), b\in\{0,1\}\}.
$$

\end{lemma}
\begin{proof}
A $\emptyset$-super-address of $F$ only contains labels for the if-then-elses of the form $C\xmapsto{l}(E,F)$ and for the $\meas$-gates. 
So, $s\in @_\emptyset F$ is either incompatible with $\ADD$, and then it is an $\emptyset$-super address in $E$, or, 
if $\ADD_s=\ADD$, it must contain an $\emptyset$-super address of $E$, together with 
values for the labels of the $\meas$-gates in either $E_1$ or $E_2$, that is, with some (uniquely determined) function $t\in @_\emptyset\tau(E_b)$, for some $b=0,1$.
\end{proof}

\subsection{The Machine $\evaltkm^*$}
In this section we introduce a slight variant of the $\evaltkm$, that we call $\evaltkm^*$.

A token for the $\evaltkm^*$ is of the form $(l,\sigma, \ADD,\ADDH)$, where $\ADD,\ADDH$ are two disjoint addresses; 
all rules are as for the $\evaltkm$ except for the rules that act on addresses:
\begin{itemize}
\item in the guard rule, the value $(l,b)$ read and deleted from the classical register is added \emph{either} to the addresses $\ADD$ \emph{or} to the addresses $\ADDH$ (but not both), the same choice being made for all tokens;
\item when a token wants to enter one of the branches of an if-then-else, it looks for the presence of the value $(l,b)$ (telling which branch has to be chosen) in \emph{both} addresses $\ADD,\ADDH$. As these are disjoint, it may only find it in one.

\end{itemize}
For any token $\alpha=(l,\sigma, \ADD,\ADDH)$ of the $\evaltkm^*$, we define the $\evaltkm$-tokens
\begin{align*}
\mathrm{erase}(\alpha)&=(l,\sigma,\ADD)&
\mathrm{join}(\alpha)&=(l,\sigma,\ADD\cup \ADDH).
\end{align*}
For any set of tokens $\MM$ of the $\evaltkm^*$, we let $\mathrm{erase}(\MM)=\{\mathrm{erase}(\alpha)\mid\alpha\in \MM\}$ indicate the corresponding set of tokens of the $\evaltkm$, and similarly for $\mathrm{join}(\MM)$.

For the following result it is enough to check what happens with the rules for the if-then-else, as all other rules are identical for the two machines:
\begin{proposition}[bisimulation between the $\evaltkm$ and the $\evaltkm^*$]\label{prop:bisi2}
For all configurations $\config, \config'$ of the $\evaltkm^*$, $\config\to\{\config_i^{\mu_i}\}_{i\in I}$ iff $\mathrm{join}(\config)\to
\{\mathrm{join}(\config_i)^{\mu_i}\}_{i\in I}$ holds for the $\evaltkm$.
\end{proposition}

\begin{definition}\label{def:wellsliced2}
A configuration $\config$ of the $\evaltkm^*$ is \emph{well-sliced} when all tokens have the same left address, that we note $\ADD_{\mathrm{left}}(\config)$, while their right addresses are defined on the same set of labels $\mathcal L_{\mathrm{right}}(\config)$ (although they may give different values). 

\end{definition}

The notion of indexed pseudo-distribution is adapted as follows:
\begin{definition}[coherent pseudo-distribution]\label{def:cohe2}
For all extended circuits $E$, a \emph{$E$-indexed coherent pseudo-distribution} of the $\evaltkm^*$ is a distribution
$\mu=\{\config_s^{\mu_s}\}_{s\in @^{\mathsf s}_\emptyset E}$ (where $\config_s=(\pi,\MM_s,m_s)$) of well-sliced configurations of the $\evaltkm^*$ such that for all $s,t\in @_\emptyset^*E$:
\begin{enumerate}

\item $\ADD_{\mathrm{left}}(\config_s)=\ADD_s$;

\item if $\ADD_s=\ADD_t$, then $\mathrm{erase}(\MM_s)=\mathrm{erase}(\MM_t)$;

\item for all $s\in  @_\emptyset^*E$, $\mathcal L_{\mathrm{right}}(\MM_s)\subseteq \LITE(E)$.
\end{enumerate}
\end{definition}

Observe that we index configurations by $\emptyset$-super-addresses $s\in @_\emptyset^*E$; in other words, we do \emph{not} count the values assigned to if-then-elses in the form $\ite{l}{C_1}{C_2}$; moreover, the condition $\ADD_{\mathrm{left}}(\config_s)=\ADD_s$ says that values stored in the second address $\ADDH$ of a configuration are also not taken into consideration in the indexing.
At the same time, we ask that the labels occurring in the right addresses of any token in the distribution are taken from the labels of the if-then-elses in the form $\ite{l}{C_1}{C_2}$ in $E$ (although such labels may take different values for each token).

\begin{rem}\label{rem:erase}
Even without using Proposition \ref{prop:bisi2}, one of the two directions of the bisimulation of Theorem \ref{thm:noif-then-else-simulation} (the simulation of $\circuittkmz$ via $\evaltkm$) continues to work if we replace the $\evaltkm$ with the $\evaltkm^*$, as soon as the (token) condition is adapted as follows:
\begin{description}
\item[(token$^*$)] for all $\ADD\in @E$ and $s\supseteq \ADD$, $\mathrm{erase}(\MM_s)=\MM|_\ADD$.
\end{description}
The idea is that a run of the $\circuittkmz$ is simulated by a (parallel) run of the $\evaltkm^*$ in which the second address $\ADDH$ is never updated (and so it remains empty all the time).  
\end{rem}

\begin{rem}\label{rem:confflag}
We can suppose the following facts:
\begin{itemize}
\item if a configuration $\config$ of the $\evaltkm^*$ contains a token in the position of a $\bit^l$ with label $l$, then one of two mutually exclusive conditions holds: either the classical register contains $(l,b)$ for some $b\in \{0,1\}$, or precisely one of the addresses $\ADD,\ADDH$ of the token contains $(l,b)$ for some $b\in \{0,1\}$. 

\item If a configuration $\config$ of the $\circuittkmz$ contains a token in the position of a $\bit^l$ with label $l$, then one of two mutually exclusive conditions holds: either the flag of the token is off, or the flag is on and the address $\ADD$ of the token contains $(l,b)$ for some $b\in \{0,1\}$.

\end{itemize}
The first assumption is easily shown to hold for all configurations appearing in a run of the $\evaltkm^*$, since a label $(l,b)$ first appears in the classical register, and is possibly moved to one of the addresses $\ADD,\ADDH$ after a guard transition is performed.
The second assumption is also easily shown to hold for all configurations appearing in a run of the $\circuittkm$, since flags are initially off, and they are turned on precisely when the corresponding label is added to the address.
Recall that the synchronous rule may only apply to a guard bit if the corresponding flag is off.

\end{rem}

The fundamental reason for considering the machine $\evaltkm^*$ is that we will make use of the erasure property below.

For any address $\ADD$ and token $\alpha=(l,\sigma,\ADDH_1,\ADDH_2)$ of the $\evaltkm^*$, let 
$$\mathrm{pushright}(\ADD,\alpha)=(l,\sigma, \ADDH_1-\ADD,\ADDH_2\cup (\ADD\cap
\ADDH_1)).$$ For any set of tokens $\MM$, let
$\mathrm{pushright}(\ADD,\MM)=\{\mathrm{pushright}(\ADD,\alpha)\mid\alpha\in
\MM\}$. Finally, for any configuration $\config=(\pi, \MM, m)$, 
define $\mathrm{pushright}(\ADD,\config)=(\pi,
\mathrm{pushright}(\ADD,\MM),m)$.

The operation $\mathrm{pushright}(\ADD,\alpha)$ simply pushes all values $(l,b)\in \ADD$ that occur in the first address $\ADDH_1$ into the second address $\ADDH_2$.
The lemma below states that, in a reduction of the $\evaltkm^*$ one can always push all produced values to the second address.

\begin{lemma}[right execution]\label{lemma:rightexec}
For any run $\config\to^*\{\donfig_i^{\mu_i}\}_{i\in I}$ of the $\evaltkm^*$, there is a corresponding run of the form 
$\config\to^* \{(\donfig'_i)^{\mu_i}\}_{i\in I}$, where 
$\donfig'_i=
\mathrm{pushright}(\ADD_{\mathrm{left}}(\donfig_i)-\ADD_{\mathrm{left}}(\config),\donfig_i).$
\end{lemma}
\begin{proof}
The execution starts from the addresses of $\config$ and possibly adds new values $(l,b)$ to either $\ADD_{\mathrm{left}}(\config)$ or the right addresses whenever it applies a guard rule, yielding a final left address $\ADD_{\mathrm{left}}(\donfig)\supseteq \ADD_{\mathrm{left}}(\config)$.
To obtain the second execution, it is enough, in all these situations, to always add values to the right address and never to the first.
\end{proof}

Parallel reduction $\mu\mtocl^a_G \mu'$ of type $a$ at address $\ADD$ for the $\evaltkm^*$ can be defined as for the $\evaltkm$, by only considering the left address $\ADD$. At the same time, a new reduction type is added, corresponding to a guard rule whose Boolean choice is added to the right address, thus leading to no branching in the parallel reduction (contrarily to what happens with the pseudo-deterministic reduction type).

The new reduction type, called \emph{pruned pseudo-deterministic} is added as a new case of Definition \ref{def:parallelred}:

\begin{definition}

A reduction type of the $\evaltkm^*$ can also be of:
\begin{itemize}

\item \emph{pruned pseudo-deterministic type}: $n=1$, $\MM_0$ contains a token in the guard of an if-then-else and $\MM_1$ is obtained from $\MM_0$ by adding in the right address of each token the value $(l,b)$, for some fixed $b=\{0,1\}$.
\end{itemize}

The corresponding parallel reduction step $\mu\mtocl^a_G\nu$ is given by:
\begin{itemize}	
	 \item $a=(\MM_0,\MM_1)$ is of pruned pseudo-deterministic type, $\MM_0$ contains one token $(l,\sigma,\ADD,\ADDH)$ in the guard $\bit^l$ of an if-then-else, $@^{\mathsf s}_\emptyset F=@^{\mathsf s}_\emptyset E$;
	 for all $s\in @^{\mathsf s}_\emptyset F$, with $\ADD_s=\ADD$, 
	 we have $\config_s=(\pi,\MM_0,m_s)$, 
	 $\donfig_s=(\pi,\MM_1,m'_s)$, with $m'_s$ obtained from $m_s$ by deleting the value $(l,b)$, while $\MM_1$ is obtained from $\MM_0$ by adding the value $(l,b)$ to the right address;
	 finally, $\nu_s=\mu_s$.
\end{itemize}

A parallel reduction $\mu_1\mtocl^{a_1}_{G_1}\dots\mtocl^{a_n}_{G_n}\nu$ is called \emph{pruned} if it does not contain reductions of (non-pruned) pseudo-deterministic type. We indicate a pruned parallel reduction as $\mu\mtocl^*_{\mathsf p}\nu$.
\end{definition}

Right executions can be used to obtain the following useful lemma:

\begin{lemma}[pruning lemma]\label{lemma:pruning}
For any parallel reduction $\{\config^1\}\mtocl^*\{\donfig_s^{\mu_s}\}_{s\in @^{\mathsf s}_\emptyset E}$, there exists
an injective function $h: @^{\mathsf s}_\emptyset \tau(E)\to @^{\mathsf s}_\emptyset E$ and a pruned reduction
$\{\config^1\}\mtocl^*_{\mathsf p}\{\donfig_{h(u)}^{\mu_{h(u)}}\}_{u\in @^{\mathsf s}_\emptyset \tau(E)}$.
\end{lemma}
\begin{proof}[Proof sketch]
The idea is to replace any reduction of pseudo-deterministic type by a corresponding pruned one. Observe that, once we fix a choice for any $\meas$-gate in $E$ (by some total function $u\in \{0,1\}^{\LMEAS(E)}=@_\emptyset^*\tau(E)$), then, during a reduction, a \emph{unique} choice is determined for any if-then-else.
This means that given $u\in @_\emptyset^*\tau(E)$, we can select a super-address $h(u)\in @_\emptyset^* E$ containing \emph{both} the choices $u$ on the $\meas$-gates and the choices for the if-then-else determined by the values in $u$. 
The pruning yields then the desired distribution.
\end{proof}

We will now establish a bisimulation result for the $\evaltkm$ and the $\evaltkm^*$ with respect to parallel reduction.
The only delicate point to adapt Proposition \ref{prop:bisi2} is the case of a reduction of pruned pseudo-deterministic type of the $\evaltkm^*$.

\begin{definition}
Given extended circuits $E,F$, the relation $E\sqsubseteq F$ is defined as follows:
\begin{itemize}
\item if $E=\tau(F)$, then $E\sqsubseteq F$;

\item if $E_1\sqsubseteq F_1$ and $E_2\sqsubseteq F_2$, then $(C\xmapsto{l}{E_1,E_2})\sqsubseteq (C\xmapsto{l}{F_1,F_2})$.
\end{itemize}
\end{definition}

\begin{definition}[parallel bisimulation]
Let $\mu=\{\config_s^{\mu_s}\}_{s\in @^{\mathsf s}E}$ be a coherent pseudo-distribution of configurations of the $\evaltkm$ and 
let $\nu=\{\donfig_t^{\nu_t}\}_{t\in @^{\mathsf s}_\emptyset F}$ be a coherent pseudo-distribution of configurations of the $\evaltkm^*$.
The relation $\mu \ \mathscr T \ \nu$ holds iff the following conditions hold:
\begin{description}

\item[(circuit)] $F\sqsubseteq E$ and $E$ has no $\ite{-}{-}{-}$ construct;
\item[(injection)] there is an injective function $h:@^{\mathsf s}_\emptyset F\to @^{\mathsf s}E$ such that 	\begin{itemize}
	\item for all $t\in @^{\mathsf s}_\emptyset F$, $t= h(t)\vert_{@^{\mathsf s}_\emptyset F}$, $\config_{h(t)}=\mathrm{join}(\donfig_t)$ and $\mu_{h(t)}=\nu_t$;
	\item for all $s\notin \mathrm{Im}(h)$, $\mu_s=0$. 
	\end{itemize}
\end{description}
\end{definition}

Observe that a consequence of the definition above is that, given $\mu\mathscr T\nu$, 
for all configuration $\config_s$ with $\mu_s>0$ there exists $t$ such that $h(t)=s$ and 
$\mathrm{join}(\donfig_t)=\config_{s}$.

For any $\small a\in \{\text{deterministic},\text{choice},\text{pseudo-deterministic},\text{pruned pseudo-deterministic}\},$ let \\ $a^\bullet\in\{\text{deterministic},\text{choice},\text{pseudo-deterministic}\}$ be $a$, if the former is either deteministic or choice, and pseudo-deterministic in all other cases.

\begin{proposition}[parallel bisimulation of the $\evaltkm$ and the $\evaltkm^*$]\label{prop:bisi2}
For all coherent pseudo-distributions of configurations $\mu,\mu'$ of the $\evaltkm$ and for all coherent pseudo-distributions of configurations $\nu,\nu'$ of the $\evaltkm^*$ the following hold:

\begin{enumerate}
\item $\mu\ \mathscr T\ \nu$ and $\mu \mtocl^a_{\ADD} \mu'$ $\Rightarrow$ 
$\exists\nu^*, \exists \ADD'\subseteq \ADD, \exists b$ s.t.~
$\mu'\mathscr T\nu^*$, $a=b^\bullet$ and 
$\nu\mtocl^{b^\bullet}_{\ADD'} \nu^*$.

\item $\mu\ \mathscr T\ \nu$ and $\nu \mtocl^a_{\ADD} \nu'$ $\Rightarrow$ $\exists\mu^*,\exists \ADDH$ s.t.~
$\mu^*\mathscr T\nu'$ and 
$\nu\mtocl^{a^\bullet}_{\ADD\cup \ADDH} \nu'$.

\end{enumerate}
\end{proposition}
\begin{proof}
For point 1.,~all cases directly follow from Proposition \ref{prop:bisi2}.
For point 2.,~if  $a$ is not of pruned pseudo-deterministic type, the claim follows quite directly from Proposition \ref{prop:bisi2}.
Let us consider the case in which $a$ is of pruned pseudo-deterministic type:
then $\nu'$ is indexed by $F'= F[F@\ADD\mapsto \ite{l}{\mathrm{id}}{\mathrm{id}}]$, so $@^{\mathsf s}_\emptyset F=@^{\mathsf s}_\emptyset F'$.
We have transitions $\donfig_t\to\{(\donfig'_t)^1\}$, in which a classical bit $(l,b)$ is read from the classical register and added to the right address.
On the other hand, let $E'$ be obtained by replacing the circuit $E@\ADD$ with $E@\ADD\xmapsto{l}(\mathrm{id},\mathrm{id})$, so $@^{\mathsf s}_\emptyset E'=\EXT(E,@^{\mathsf s}_\emptyset E,\ADD,l)$.
Notice that from $F\sqsubseteq E$ it now follows $F'\sqsubseteq E'$.  
From $h:@^{\mathsf s}_\emptyset F\to @^{\mathsf s} E$ define
$h^*:@^{\mathsf s}_\emptyset F'\to @^{\mathsf s}E'$ by $h(t)=h(s)::(l,b)$. 
Now, for $s=h(t)$, we have two transitions $\config_{s}\to_1\config'_{s::(l,b)}$ and
$\config_s\to_0\config'_{s\cup\{(l,1-b)\}}$. Notice that $s\cup\{(l,b)\}=h^*(t)$ and $\mathrm{join}(\donfig'_t)=\config'_{s\cup\{(l,b)\}}$. We conclude that the distribution $\mu^*$ 
 given by all such $\config'_{s\cup\{(l,b)\}}$ satisfies $\mu^*\mathscr T\nu'$.
\end{proof}

\subsection{Inner, Initial and Final Configurations}

For the simulation argument we need to consider a more general notion of configuration, called \emph{inner}, which intuitively corresponds to the case in which the machine is visiting some \emph{sub-}derivation $\pi'$ of a larger derivation. This means that, while we are only interested in the positions of the tokens situated inside $\pi'$, we must consider also tokens lying outside $\pi'$ (but which never move), and, more importantly, a quantum state that comprises more qubits than those actually accounted for by $\pi'$.
Due to the well-known entanglement phenomenon, we cannot hope to retrieve the behavior of such, \emph{outer} qubits, just by tensoring their state with the state of the \emph{inner} qubits.

\begin{definition}[inner configuration of the $\evaltkm^*$]
Let $\pi$ be a derivation and $\pi'\leq\pi$ be a sub-derivation of $\pi$.
An \emph{inner configuration of the $\evaltkm^*$ for $\pi'\leq\pi$} is a tuple
$(\pi'\leq \pi,\MM^{\mathrm{inner}}\cup\MM^{\mathrm{outer}},m,c)$, such that:
\begin{itemize}
\item $\MM^{\mathrm{inner}}$ is a set of tokens lying in positions of $\pi'$, while $\MM^{\mathrm{outer}}$ is a set of tokens lying in positions of $\pi$ that are not in $\pi'$; 

\item  the quantum+classical register $m$ takes into account the tokens in $\MM^{\mathrm{inner}}\cup\MM^{\mathrm{outer}}$;

\item for any quantum+classical register $m'$ for $\MM^{\mathrm{inner}}$, 
$(\pi',\MM^{\mathrm{inner}}, m')$ is a configuration of the \evaltkm for $\pi'$.
\end{itemize}

Notice that the definition above implies that the addresses $\ADD,\ADDH$ of the tokens in 
$\MM^{\mathrm{inner}}\cup\MM^{\mathrm{outer}} $ only refer to labels $l$ occurring in $\pi'$. An inner configuration $\config=(\pi,\MM^{\mathrm{inner}}\cup\MM^{\mathrm{outer}},m)$ for $\pi'\leq\pi$ is \emph{well-sliced} when $\MM^{\mathrm{inner}}\cup\MM^{\mathrm{outer}}$ is well-sliced.

For any circuit $E$, a $E$-indexed pseudo-configuration of the $\evaltkm^*$ for $\pi'\leq\pi$ is a pseudo-distribution 
$\{\config_s^{\mu_s}\}_{s\in @_\emptyset^*E}$ of well-sliced inner configurations for the $\evaltkm^*$ for $\pi'\leq\pi$ such that:
\begin{enumerate}
\item for all $s\in @_\emptyset^*E$, $\ADD_{\mathrm{right}}(\config_s)=\ADD_s$;
\item if $\ADD_s=\ADD_t$, then $\MM^{\mathrm{inner}}_s=\MM^{\mathrm{inner}}_t$;
\item for all $s\in  @_\emptyset^*E$, $\mathcal L_{\mathrm{right}}(\MM^{\mathrm{inner}}_s)\subseteq \LITE(E)$.

\end{enumerate}

\end{definition}

Whenever $\ADD$ is an address in $\pi'$, the notion of parallel reduction $\mu\mtocl^a_G\mu'$ for $E$-indexed pseudo-distributions of inner configurations of the $\evaltkm^*$ can be defined in the same way as in the non-inner case.

\begin{definition}[inner configuration of the $\circuittkm$]
An \emph{inner configuration of the \circuittkm for $\pi'\leq\pi$} is a tuple
$(\pi'\leq \pi,\MM^{\mathrm{inner}}\cup\MM^{\mathrm{outer}},E)$, such that 
\begin{itemize}
\item $\MM^{\mathrm{inner}}$ is a set of tokens lying in positions of $\pi'$, while $\MM^{\mathrm{outer}}$ is a set of tokens lying in positions of $\pi$ that are not in $\pi'$;
\item $(\pi',\MM^{\mathrm{inner}},E)$ is a \circuittkm-configuration for $\pi'$.
\end{itemize}

\end{definition}

The definition above implies that the circuit $E$ acts over the space spanned by the tokens in $\MM^{\mathrm{inner}}$ only.

We also introduce notations for initial and final configurations:

\begin{definition}[initial and final configurations for the $\evaltkm^*$]
Suppose $\pi':\Delta\vdash N:A$ is a sub-derivation of
$\pi:\Delta'\vdash M:A'$.
An inner configuration $(\pi'\leq \pi, \MM^{\mathrm{inner}}\cup \MM^{\mathrm{outer}},m)$ is said \emph{initial for $\pi'$}, and noted $\initc(\pi',\leq \pi,\MM^{\mathrm{inner}}\cup \MM^{\mathrm{outer}},m)$, if for any quantum+classical register $m'$ for $\MM^{\mathrm{inner}}$, $(\pi', \MM^{\mathrm{inner}},m')$ is an initial configuration for $\pi'$.

An inner configuration $(\pi'\leq \pi, \MM^{\mathrm{inner}}\cup \MM^{\mathrm{outer}},m)$ is said \emph{final for $\pi'$}, and noted $\finc(\pi',\leq \pi,\MM^{\mathrm{inner}}\cup \MM^{\mathrm{outer}},m)$, if for any quantum+classical register $m'$ for $\MM^{\mathrm{inner}}$, $(\pi', \MM^{\mathrm{inner}},m')$ is a final configuration for $\pi'$.
\end{definition}

We introduce analogous notations for the circuit machine:

\begin{definition}[initial and final configurations for the \circuittkm]
Suppose $\pi':\Delta\vdash N:A$ is a sub-derivation of
$\pi:\Delta'\vdash M:A'$.
An inner configuration $(\pi'\leq \pi, \MM^{\mathrm{inner}}\cup \MM^{\mathrm{outer}},E)$ is said \emph{initial for $\pi'$}, and noted $\initcc(\pi',\leq \pi,\MM^{\mathrm{inner}}\cup \MM^{\mathrm{outer}},E)$, if $(\pi', \MM^{\mathrm{inner}},E)$ is an initial configuration for $\pi'$.

An inner configuration $(\pi'\leq \pi, \MM^{\mathrm{inner}}\cup \MM^{\mathrm{outer}},E)$ is said \emph{final for $\pi'$}, and noted $\fincc(\pi',\leq \pi,\MM^{\mathrm{inner}}\cup \MM^{\mathrm{outer}},E)$, if $(\pi', \MM^{\mathrm{inner}},E)$ is a final configuration for $\pi'$.
\end{definition}

\subsection{The Synchronous Reduction Type}

Whenever $\ADD$ is an address in $\pi'$, the notion of reduction $\config\mtocl^a_G \config'$ for inner configurations of the $\circuittkm$ can be defined in the same ways as in the non-inner case, except for a new case, corresponding to an application of the synchronous rule, that we define below:

\begin{definition}[synchronous reduction type]\label{def:optimizedtype}
We introduce a fourth reduction type, called \emph{synchronous}, given by tuples $a=(\MM_0,\MM_1)$, where
for some rule in the underlying derivation as follows:
\begin{equation}\label{eq:iterule}
 \begin{array}{c}
 {\infer{\Sigma_2, \Theta_3 \vdash \ite{N}{P}{Q} : A_3}{\deduce{\Sigma_1
        \vdash N : \bit^l}{\pi_0} \qquad \deduce{\Theta_1 \vdash P : A_1}{\pi_1} \qquad \deduce{\Theta_2 \vdash
        Q : A_2}{\pi_2}}}\end{array}
\end{equation}
the tokens in $\MM_0$ negatively saturate $\Theta_3,A_3$ and the tokens in $\MM_1$ positively saturate $\Theta_3,A_3$.

Given configurations $\config,\config'$ for $\pi$, we say  that there is a reduction of the $\circuittkm$ of asynchronous type $\config\to^a_{\ADD} \donfig$ at address $\ADD$ if
the derivation $\pi$ contains a
sub-derivation $\pi'$ as in \eqref{eq:iterule} 
and the following hold:
\begin{itemize}
\item $\config=( \pi,\MM, E)$, where
$
\MM|_{\ADD}=\Big( \MM_{\Sigma,\bit}\cup \MM_{\Theta,A}\cup\MM_{\lnot\Sigma,\Theta,A}\Big)\cup\MM^\dag
$ and the following hold:
\begin{itemize}
\item $\MM_{\Sigma,\bit}$,
contains tokens within $\pi_0$, including one over the conclusion $\bit^l$;
\item $\MM_{\Theta,A}$ negatively saturates $\Theta_3,A_3$;
\item $\MM^\dag$ contains tokens outside $\Sigma_2,\Theta_3,A_3$.

\end{itemize}
\item $\config'=(\pi, \NN,F)$, where
$\NN|_{\ADD}=\big(\MM_{\Sigma,\bit}\cup \NN_{\Theta,A}\cup \MM_{\lnot\Sigma,\Theta,A}\big)\cup\MM^\dag$ and the following hold:
{\small	\begin{itemize}
	\item $\initcc(\pi_1,\MM_{\Theta,A}, -) \to ^{*}\fincc(\pi_1, \NN_{\Theta,A}^*, E_1)$,
	where $\NN_{\Theta,A}^*$ contains the same tokens as $\NN_{\Theta,A}$, but addresses are updated with labels from $E_1$,
	\item $\initcc(\pi_2,\MM_{\Theta, A}, -) \to ^{*}\fincc(\pi_2,\NN_{\Theta,A}^{**}, E_2)$,
		where $\NN_{\Theta,A}^{**}$ contains the same tokens as $\NN_{\Theta,A}$, but addresses are updated with labels from $E_2$,
	\item $F$ is obtained from $E$ via the replacement
	$$E@\ADD \quad \mapsto \quad E@\ADD;\big((\ite{l}{\tau(E_1)}{\tau(E_2)})\big).$$
	\end{itemize}}
\end{itemize}

\end{definition}

\subsection{The Simulation Theorem}

We now introduce the simulation relation between inner
configurations:

\begin{definition}
Let $\pi$ be a derivation, $\pi'\leq\pi$ be a sub-derivation of
$\pi$ and let $E,F$ be extended circuits. Let ${\rho}$ indicate a uniform distribution of quantum+classical registers, with underlying environment $L$. We introduce a binary relation
$\mathscr S^{{\rho}}_{\pi'\leq \pi}$ such that, for all 
$E$-indexed coherent pseudo-distribution $\mu=\{\config_s^{\mu_s}\}_{s\in @_\emptyset^*E}$ of inner $\evaltkm^*$ 
configurations for $\pi'\leq \pi$, where $\config_s=(\pi'\leq \pi, \MM_s^{\mathrm{inner}}\cup \MM_s^{\mathrm{outer}},
m_s)$, with $m_s=[L_s,Q_s,v_s]$, and for all inner $\circuittkm$ configuration for $\pi'\leq \pi$,
$\config'=(\pi'\leq \pi,\MM^{\mathrm{inner}}\cup \MM^{\mathrm{outer}}, F)$,
$
\mu \ \mathscr S_{\pi\leq \pi'}^{{\rho}} \ \config'
$
holds when the following hold:
\begin{description}
\item[(index)] $E=F$,
\item[(token)] for all $s\in @_\emptyset^*E$, $\mathrm{erase}(\MM_s^{\mathrm{inner}})=\MM^{\mathrm{inner}}\vert_{\ADD_s}$;
\item[(mixed state)] for all $s\in @_\emptyset^*E$, $\mu_s\cdot\tocpm(m_s)=(F^{\sharp s}\otimes I^{\mathrm{outer}})(\tocpm(\rho))$, where $I^{\mathrm{outer}}$ indicates the identity on all bits and qubits corresponding to tokens in $\MM^{\mathrm{outer}}$.
\end{description}
\end{definition}

Observe that, when $\pi'=\pi$ and
${\rho}=\{m^1\}$, the binary relation $\mathscr S_{\pi'\leq
\pi}^{\tilde{\rho}}$ coincides with $\mathscr R^{m}_{\pi}$.

When $\pi'\leq\pi$ is a subderivation, we let 
$\mu
\mtocl^{\pi'}\mu'$ (resp.~$\config \to^{\pi'} \config'$) indicate a sequence of parallel reductions of the $\evaltkm^*$ (resp.~of reductions of the $\circuittkm$) that only move tokens inside $\pi'$.

We have now set up all the ingredients to establish the simulation
between the $\evaltkm^*$ and the $\circuittkm$.

\begin{theorem}[the $\evaltkm^*$ simulates the \circuittkm]
	\label{thm:simulation-optimized-rule}
Let $\pi'$ be a subderivation of $\pi$ and suppose that all executions of the $\circuittkm$ over $\pi$ terminate.
Then, for all
extended circuit $E$, $E$-indexed pseudo-distributions of inner configurations $\mu$ of the $\evaltkm^*$ for $\pi'\leq \pi$, for all inner configurations $\config$ of the $\circuittkm$ 
for $\pi'\leq \pi$, for all distributions of quantum+classical memories
${\rho}$, for all address $\ADD\in @_\emptyset E$ and reduction type $a$,
$$
\mu \ \mathscr S^{{\rho}}_{\pi'\leq\pi} \ \config
\text{ and }
 \config \to^{a}_{\ADD} \config'
 \ \ \Rightarrow \ \ 
 \exists 
 \mu' \text{ s.t. }
 \mu \mtocl^{\pi'}\mu' \text{ and }
 \mu' \ \mathscr S^{{\rho}}_{\pi'\leq\pi} \ \config'.
$$
 \end{theorem}
\begin{proof}

We argue by induction on the ``if-then-else-depth'' of $\pi'$.

For all transitions $\config\to \config'$ of deterministic, choice or quasi-deterministic type, the claim can be proved  as in the proof of Theorem
\ref{thm:noif-then-else-simulation}. Hence, we need only consider the case of a reduction $\config\to \config'$ of asynchronous type.

Following Definition \ref{def:optimizedtype}, let us suppose that $\config,\config'$ are as follows:
\begin{itemize}

\item $\config=(\pi'\leq\pi,\MM\cup\MM^{\mathrm{outer}},E)$ where $\MM$ contains the inner tokens and $\MM|_{\ADD}$ is as in Definition \ref{def:optimizedtype};

\item $\config'=(\pi'\leq\pi,\NN\cup\NN^{\mathrm{outer}},F)$ where $\NN$ contains the inner tokens and $\NN|_{\ADD}$ and $F$ are as in Definition \ref{def:optimizedtype}.

\end{itemize}

Further suppose that $\mu=\{\config_s^{\mu_s}\}_{s\in @^{\mathsf s}_\emptyset E}$, with $\config_s=(\pi'\leq \pi,\MM_s\cup\MM^{\mathrm{outer}},m_s)$, where $\MM_s$ contains the inner tokens.

Our first goal is to define the pseudo-distribution $\mu'=\{\donfig_t^{\mu'_t}\}_{t\in @^{\mathsf s}_\emptyset F}$.

From the reduction $\config \to_G^a \config'$ it follows that the tokens $\MM|_{\ADD}$ at address $\ADD$ in the $\circuittkm$ contain one token in the guard position $\bit^l$. 
As a consequence, the circuit $E@\ADD$ has a typing of the form $\Delta_1\rhd E\rhd \bit^l,\Delta_2$.
Since an asynchronous reduction takes place in the $\circuittkm$, this means that the flag of the bit is off, meaning that the label $(l,b)$ has not been added to the address $\ADD$ (i.e.~that $\ADD$ does not contain $(l,b)$, for any $b\in \{0,1\}$), cf.~Remark \ref{rem:confflag}.

By the (token) condition we deduce that, for all $s\supseteq \ADD$, also the tokens $\MM_s$ in the $\evaltkm^*$ contain one token in the guard position $\bit^l$. Moreover, the label $l$ does not occur in either of the addresses $\ADD,\ADDH$ of these tokens: for $\ADD$ this is a consequence of $ \mathrm{erase}(\MM_s)=\MM|_{\ADD}$; for $\ADDH$, this follows from 
$\mathcal L_{\mathrm{right}}(\MM_s)\subseteq \LITE(E)$ (cf.~Definition \ref{def:cohe2}) and the fact that $l$ cannot be in $\LITE(E)$, because, if $E$ contained some subcircuit $\ite{l}{C_1}{C_2}$, then $l$ would be the output label of two distinct $\bit$s, which is impossible by construction.
Therefore, by the assumption made in Remark \ref{rem:confflag}, the classical register in $m_s$ must store a value $(l,b_s)$ for the corresponding label.

For a context $\Gamma$ and a polarity $\epsilon\in\{+,-\}$, let
$|\Gamma^{\epsilon}|_{\qbit}$ (resp.~$|\Gamma^{\epsilon}|_{\bit}$)
indicate the number of positive atomic occurrences of $\qbit$s
(resp.~$\bit$s) in $\Gamma$, and let $ |\Gamma^{\epsilon}|=
|\Gamma^{\epsilon}|_{\qbit}\cdot |\Gamma^{\epsilon}|_{\bit}$.

Observe that, for all types $A,B$, a CPM morphism $\Phi:\interp{\bit}\otimes\interp{A}\to\interp{B}$ is of the form
$$
\Phi(x_1,x_2): \interp{A}\times\interp{A}\to\interp{B}
$$
since $\interp{\bit}\otimes\interp{A}\simeq\interp{A}\times\interp{A}$.
In particular, given that $\Theta\rhd E_i\rhd \Theta'$,
for all $s\supseteq\ADD$, the morphism
\begin{align*}
\Phi:=\Big(\ite{l}{\tau(E_1)}{\tau(E_2)}\Big)^{\sharp s} : 
\interp{\Theta}\times \interp{\Theta}\simeq \interp{\bit}\otimes\interp{\Theta}\longrightarrow \interp{\Theta'}
\end{align*}
recalling Equation \eqref{eq:itesl}, is given by
\begin{align}\label{eqphidef}
\Phi(x_1, x_2)= (1-s(l))\Phi_1(x_1)+s(l)\Phi_2(x_2),
\end{align}
where $\Phi_1=E_1^{\sharp s}={\tau(E_1)}^{\sharp s}$ and $\Phi_2={E_2^{\sharp s}}={\tau(E_2)}^{\sharp s}$.

Let us fix then some $s\in @^{\mathsf s}_\emptyset E$ with $\ADD_s=\ADD$.
Let us consider, for some $\ell\in\{1,2\}$, the sub-derivation
$\pi_{\ell}$: by comparing the initial configurations we can check that
we have
$$
\{\initc(\pi_{\ell}\leq \pi, \MM^{-\ADD}_s, {m^*_{s})^1\} 
\ \mathscr S_{\pi_{\ell}\leq\pi}^{m^*_{s}}}  \
 \initcc(\pi_{\ell}\leq \pi, \MM|^{-\ADD}_{\ADD}, -),
$$
where:
\begin{itemize}
\item the set of tokens $\MM^{-\ADD}_s$ (resp.~$\MM|^{-\ADD}_{\ADD}$) is obtained from $\MM_s$ (resp.~$\MM|_{\ADD}$) by replacing the address $\ADD$ with the empty set,
\item the register $m^*_{s}$ is obtained from $m_{s}$ by removing the value $(l,b_s)$.
\end{itemize}

If the execution of $\initcc(\pi_{\ell}\leq \pi, \MM|^{-\ADD}_{\ADD}, -)$ reached a deadlock, one could similarly obtain a deadlock in the execution of $\initcc(\pi_{\ell}\leq \pi, \MM, -)$, and, a fortiori, in the execution of the $\circuittkm$ from the start.
By applying the no-deadlock hypothesis, we deduce then that there exists a reduction sequence of the $\circuittkm$ leading to a $\pi_{\ell}$-final configuration 
$$\fincc\Big(\pi_{\ell}\leq \pi, \big(\MM_{\Sigma,\bit}^{-\ADD,+\mathcal L(E_\ell)}\cup\NN_{\Theta,A,\ell}\cup\MM_{\lnot \Sigma,\Theta,A}^{-\ADD,+\mathcal L(E_\ell)}\big)\cup\MM^\dag, E_{\ell}),$$
 where:
\begin{itemize}
\item $\NN_{\Theta,A,\ell}$ positively saturates $\Theta_3,A_3$ and its tokens have an address containing only labels from $E_{\ell}$ (so it does not contain $l$);
\item $\MM_{\Sigma,\bit}^{-\ADD,+\mathcal L(E_\ell)}$ is obtained from $\MM_{\Sigma, \bit}^{-\ADD}$ by duplicating the tokens following the addresses in $\NN_{\Theta,A,\ell}$ (corresponding to labels in $\LITES(E_\ell)$). More precisely, $\MM_{\Sigma,\bit}^{-\ADD,+\mathcal L}=\{(l',\sigma,\ADDH)\mid (l',\sigma,\ADD)\in \MM_{\Sigma, \bit}\text{ and }\exists l''\exists \tau\ (l'',\tau,\ADDH)\in
	\NN_{\Theta,A,\ell}\}$;

\item $\MM_{\lnot \Sigma,\Theta,A}^{-\ADD,+\mathcal L(E_\ell)}$ is defined similarly.
\end{itemize}

By applying the induction hypothesis to $\pi_{\ell}$ a finite number
of times, we deduce that there exists a reduction sequence of the $\evaltkm^*$ leading from
$\initc(\pi_\ell\leq \pi, \MM^{-\ADD}_s, m^*_{s})$ to some $E_{\ell}$-indexed pseudo-distribution of 
$\pi_{\ell}$-final configurations $\nu_{s,\ell}=\{\mathcal B_{s,\ell,t}^{\nu_{s,\ell,t}}\}_{t\in @_\emptyset^*E_\ell}=\{\finc(\pi_{\ell}\leq \pi, \mathcal P^{-\ADD, +\mathcal L(E_\ell)}_{s,\ell,t}, r_{s,\ell,t})^{\nu_{\ell,t}}\}_{t\in @_\emptyset^*E_{\ell}}$,
where $\mathcal P^{-\ADD, +\mathcal L(E_\ell)}_{s,\ell,t}, r_{s,\ell,t}$ is a set of tokens whose addresses do not contain $\ADD$ but contain labels from $E_\ell$. In other words, we have:
\begin{equation}\label{eq:exec1}
\left\{\initc\left(\pi_{\ell}\leq \pi, \MM^{-\ADD}_s, m^*_{s}\right)^1\right\}  \mtocl^{ \pi_{\ell}}
\left\{\finc\left(\pi_{\ell}\leq \pi, \mathcal P^{-\ADD,+\mathcal L(E_\ell),}_{s,\ell,t}, r_{s,\ell,t}\right)^{\nu_{s,\ell,t}}
\right\}_{t\in @_\emptyset^*E_{\ell}},
\end{equation}
as well as
\begin{equation}\label{eq:split}
\begin{split}
\Big\{
\finc\Big(\pi_{\ell}& \leq \pi,  \mathcal P^{-\ADD,+\mathcal L(E_\ell)}_{s,\ell,t},  r_{s,\ell,t}\Big)^{\nu_{s,\ell,t}}
\Big\}_{t\in @_\emptyset^*E_{\ell}} \\
&
\ 
\mathscr S_{\pi}^{m^*_s} 
\quad
\fincc
\left(
\pi_{\ell}\leq\pi, 
	\left(
	\MM_{\Sigma,\bit}^{-\ADD,+\mathcal L(E_\ell)}\cup\NN_{\Theta,A,\ell}			\cup\MM_{\lnot \Sigma,\Theta,A}^{-\ADD,+\mathcal L(E_\ell)}
	\right)
	\cup\MM^\dag, E_{\ell}
\right).
\end{split}
\end{equation}
This implies that for any $t\in @_\emptyset^*E_\ell$, the ($\mathrm{erase}(-)$ of the) tokens in $\mathcal P^{-\ADD,+\mathcal L(E_\ell)}_{s,\ell,t}$ coincide with the tokens in 
$\MM_{\Sigma,\bit}^{-\ADD,+\mathcal L(E_\ell)}\cup\NN_{\Theta,A,\ell}\cup\MM_{\lnot \Sigma,\Theta,A}^{-\ADD,+\mathcal L(E_\ell)}\cup \MM^\dag$ of address $\ADD_t$. 
Moreover, it also implies the following equation, for all $t\in @_\emptyset^*E_\ell$:
\begin{equation}\label{eqphiell}
( \mathrm{Id}_{\Sigma, \bit}\otimes E_\ell^{\sharp t}\otimes  \mathrm{Id}_{\lnot \Sigma,\Theta,A}\otimes \mathrm{Id}_{\MM^\dag})( \tocpm(m^*_{s}))=\nu_{s,\ell,t}\cdot\tocpm(r_{s,\ell,t}).
\end{equation}

We are now going to apply four transformations to the execution \eqref{eq:exec1}:
\begin{enumerate}

\item by adding the address $\ADD$ to all tokens we obtain a reduction 
$$
\left\{\initc(\pi_{\ell}\leq \pi, \MM_s, m^*_{s})^1\right\} 
\mtocl^{ \pi_{\ell}} \
\left\{\finc\left(\pi_{\ell}\leq \pi, \mathcal P^{+\ADD,+\mathcal L(E_\ell)}_{s,\ell,t}, r_{s,\ell,t}\right)^{\nu_{s,\ell,t}}
\right\}_{t\in @_\emptyset^*E_{\ell}};
$$
where $\mathcal P^{+\ADD,+\mathcal L(E_\ell)}_{s,\ell,t}$ is obtained from $\mathcal P^{-\ADD,+\mathcal L(E_\ell)}_{s,\ell,t}$ by adding the (left) address $\ADD$ to all its tokens.

\item By applying Lemma \ref{lemma:rightexec} we obtain a reduction 
$$
\left\{\initc(\pi_{\ell}\leq \pi, \MM_s, m^*_{s})^1\right\} 
\mtocl^{ \pi_{\ell}} \ 
\left\{\finc\left(\pi_{\ell}\leq \pi, \mathcal P^{+\ADD,-\mathcal L(E_\ell)}_{s,\ell,t}, r_{s,\ell,t}\right)^{\nu_{s,\ell,t}}
\right\}_{t\in @_\emptyset^*E_{\ell}},
$$
where $\mathcal P^{+\ADD,-\mathcal L(E_\ell)}_{s,\ell,t}$ is obtained from $\mathcal P^{+\ADD,+\mathcal L(E_\ell)}_{s,\ell,t}$ by pushing all values $(l',b')$ from labels of $E_{\ell}$ to the right address. Observe that the left address of all tokens in $\mathcal P^{+\ADD,-\mathcal L(E_\ell)}_{s,\ell,t}$ is now $\ADD$.

\item Finally, by applying the pruning lemma (Lemma \ref{lemma:pruning}), we obtain an injective function $h_s:@^{\mathsf s}_\emptyset\tau(E_\ell)\to @^{\mathsf s}_\emptyset E_\ell$ and a pruned reduction

$$
\left\{\initc(\pi_{\ell}\leq \pi, \MM_s, m^*_{s})^1
\right\} 
\mtocl^{ \pi_{\ell}}_{\mathsf p} \ 
\left\{\finc\left(\pi_{\ell}\leq \pi, \mathcal P^{+\ADD,-\mathcal L(E_\ell)}_{s,\ell,h_{s}(u)}, r_{s,\ell,h_s(u)}\right)^{\nu_{s,\ell,h_s(u)}}
\right\}_{u\in @_\emptyset^*\tau(E_\ell)}.
$$
\end{enumerate}

Now, recall that the register $m_s$ stores a value $(l,b_s)$ for the label $l$.
By applying a reduction of pruned pseudo-deterministic type followed by the reduction above we obtain then a parallel reduction $\mu \mtocl^{\pi'} \mu'$, with 
\begin{equation}\label{eq:psedu}
\begin{split}
\mu'&=\{\donfig_t^{\mu'_t}\}_{t\in @^{\mathsf s}_\emptyset F}\\
&:=
\left\{\finc\left( \pi'\leq\pi, \mathcal P^{+\ADD,-(l,b_s),-\mathcal L(E_{b_s})}_{s,b_s,u}, r_{s,b_s,h_s(u)}\right)^{\mu_s\cdot\nu_{s,b_s,t}}
\right\}_{s\supseteq\ADD, u\in @_\emptyset^*\tau(E_{b_s})} +
\Big \{\config_s^{\mu_s}\Big\}_{s\not\supseteq\ADD}
\end{split}
\end{equation}
where $\mathcal P^{+\ADD,-(l,\ell),-\mathcal L(E_{\ell})}_{s,\ell,u}$ is obtained from $\MM^{+\ADD,-\mathcal L(E_{\ell})}_{s,\ell,h_s(u)}$ by adding $(l,\ell)$ to the right address of each token.
Observe that, from Lemma \ref{lemma:optiaddress} it follows that the index set for the right-hand pseudo-distributions above
is precisely $@^{\mathsf s}_\emptyset F$. This concludes our argument that the parallel reduction $\mu\mtocl^{\pi'}\mu'$ is well-defined.

Our final goal is to show that $\mu' \ \mathscr S^{{\rho}}_{\pi'\leq \pi} \ \config'$. 
The previous argument ensures that the (index) condition holds.
For the (token) condition, first recall that, from the $\circuittkm$-configurations, for $\ell=0,1$, 
$$
\fincc\left(\pi_{\ell}\leq \pi, \left(\MM_{\Sigma,\bit}^{-\ADD,+\mathcal L(E_\ell)}\cup\NN_{\Theta,A,\ell}\cup\MM_{\lnot \Sigma,\Theta,A}^{-\ADD,+\mathcal L(E_\ell)}\right)\cup\MM^\dag, E_{\ell}\right),
$$
the application of the synchronous rule yields the configuration 
$
\donfig:= \fincc\left(\pi'\leq \pi,\NN, F\right),
$
where
$$
\NN|_\ADD=\Big( \MM_{\Sigma,\bit}\cup\NN^{+\ADD,-\mathcal L(E_\ell)}_{\Theta,A}\cup\MM_{\lnot \Sigma,\Theta,A}\Big)\cup \MM^\dag,
$$
and 
$\NN^{+\ADD,-\mathcal L(E_\ell)}_{\Theta,A}
=\{(l',\sigma,\ADD)\mid \exists \ADDH\ (l',\sigma,\ADDH)\in \NN_{\Theta,A,\ell},\ell\in\{0,1\}\}
$.

By the (token) condition for \eqref{eq:split} 
we deduce that, for all $u\in @^{\mathsf s}_\emptyset\tau( E_\ell)$,
$$
\mathrm{erase}\left(\mathcal P^{-\ADD,+\mathcal L(E_\ell)}_{s,\ell, h_s(u)}\right)=\left(
\MM^{-\ADD,+\mathcal L(E_\ell)}_{\Sigma,\bit}\cup \NN_{\Theta,A,\ell}\cup\MM^{-\ADD,+\mathcal L(E_\ell)}_{\lnot \Sigma,\Theta,A}\cup \MM^\dag\right) \Big\vert_{\ADD_{ h_s(u)}}.
$$ 
Observing that $\NN_{\Theta,A,\ell}$ only contains copies of the tokens in the positive positions of $\Theta\vdash A$ (one copy for each address $\ADDH\in @^{\mathsf s}_\emptyset(E_\ell)$), and that $\MM^{+\ADD,-\mathcal L(E_\ell)}_{\Sigma,\bit}$
contain copies of the tokens in $\MM_{\Sigma,\bit}$ (one copy for each address $\ADDH\in @^{\mathsf s}_\emptyset(E_\ell)$), and similarly for $\MM^{+\ADD,-\mathcal L(E_\ell)}_{\lnot \Sigma,\Theta,A}$, we deduce that
$$
\mathrm{erase}\left(\mathcal P^{+\ADD,-\mathcal L(E_\ell)}_{s,\ell, h_s(u)}\right)=
\MM_{\Sigma,\bit}\cup \NN^{+\ADD,-\mathcal L(E_\ell)}_{\Theta,A,\ell}\cup\MM_{\lnot \Sigma,\Theta,A}\cup \MM^\dag,
$$ 
and thus, finally, for all $t\in @^{\mathsf s}_\emptyset F$ with $t=s\cup u\supseteq\ADD$,
\begin{equation*}
\begin{split}
\mathrm{erase}(\NN_{t})=
\mathrm{erase}\left(\mathcal P^{+\ADD,-\mathcal L(E_\ell)}_{s,\ell, h_s(u)}\right) &= 
\MM_{\Sigma,\bit}\cup \NN^{+\ADD,-\mathcal L(E_\ell)}_{\Theta,A,\ell}\cup\MM_{\lnot \Sigma,\Theta,A}\cup\MM^\dag= \NN\vert_{\ADD}=\NN\vert_{\ADD_t}.
\end{split}
\end{equation*}

To conclude it remains to check the (mixed state) condition. For this, we need to check that, for all $s\supseteq \ADD$ and $u\in @^{\mathsf s}_\emptyset \tau(E_{b_s})$, 
\begin{equation}\label{eqphiall}
(F^{\sharp(s\cup u)}\otimes  \mathrm{Id}^{\mathrm{outer}})( \tocpm({\rho}))= 
\mu_s\cdot \nu_{s,\ell,h_s(u)}\cdot
\tocpm(r_{s,b_s,h_s(u)}),
\end{equation}
observing that, by hypothesis, and by \eqref{eq:split}, we have, for all $s\in @^{\mathsf s}_\emptyset E$, $\ell=0,1$ and $t\in @^{\mathsf s}_\emptyset E_\ell$,
\begin{align}
\mu_s\cdot\tocpm(m_s)&=({E}^{\sharp s}\otimes \mathrm{Id}^{\mathrm{outer}})(
\tocpm({\rho})) \label{eq:hypo1}\\
\nu_{s,\ell,t}\cdot \tocpm(r_{s,\ell,t})&=
(E_\ell^{\sharp t}\otimes \mathrm{Id}^{\mathrm{outer}})(\tocpm(m_s^*))\label{eq:hypo2}.
\end{align}

Recall that $F$ is obtained from $E$ by replacing the leaf $E@\ADD$ by $E@\ADD;\ite{l}{\tau(E_0)}{\tau(E_1)}$.

We have then that 
$$
F^{\sharp (s\cup u)}=F^{\sharp (s\cup \{(l,0)\}\cup u)}+ F^{\sharp (s\cup \{(l,1)\}\cup u)}=
 E_{0}^{\sharp u}\circ \pi^l_{0}\circ E^{\sharp s}+
E_{1}^{\sharp u}\circ \pi^l_{1}\circ E^{\sharp s}.
$$

Observing that 
$
\pi^l_{b_s}\circ
  \tocpm(m_s)=\tocpm(m^*_{s})$ and that $\pi^l_{1-b_s}\circ
  \tocpm(m_s)=0$, we can compute: 
\begin{align*}
(F^{\sharp (s\cup u)}\otimes  \mathrm{Id})( \tocpm({\rho})) &= 
\Big(\big(E_{0}^{\sharp u}\circ \pi^l_{0}\circ E^{\sharp s}+
E_{1}^{\sharp u}\circ \pi^l_{1}\circ E^{\sharp s}\big)\otimes \mathrm{Id}\Big)( \tocpm({\rho}))\\
&=
\Big(
(E_{0}^{\sharp u}\otimes \mathrm{Id})(\pi^l_{0}\otimes \mathrm{Id})
+
(E_{1}^{\sharp u}\otimes \mathrm{Id})(\pi^l_{1}\otimes \mathrm{Id})
\Big)
( E^{\sharp s}\otimes \mathrm{Id})( \tocpm({\rho}))\\
&\stackrel{\tiny\eqref{eq:hypo1}}{=}
\mu_s\cdot
\Big(
(E_{0}^{\sharp u}\otimes \mathrm{Id})(\pi^l_{0}\otimes \mathrm{Id})
+
(E_{1}^{\sharp u}\otimes \mathrm{Id})(\pi^l_{1}\otimes \mathrm{Id})
\Big)
(  \tocpm(m_s))\\
&=
\mu_s\cdot
(E_{b_s}^{\sharp u}\otimes \mathrm{Id})(  \tocpm(m^*_s))\\
&\stackrel{\tiny\eqref{eq:hypo2}}{=}
\mu_s\cdot \nu_{s,b_s,h_s(u)}\cdot \tocpm(r_{s,b_s,h_s(u)}).
\end{align*}
\end{proof}

Similarly to how we did for Theorem \ref{thm:noif-then-else-simulation} and the \circuittkmz, relying on the 
bisimulation between the \evaltkm and the \evaltkm$^*$ (Proposition \ref{prop:bisi2}) and on the simulation between the \circuittkm and the \evaltkm$^*$ (Theorem \ref{thm:simulation-optimized-rule}), we can now deduce the following soundness result:

\begin{theorem}[Soundness of the \circuittkm wrt to the \evaltkm]
Let $\pi :\ \vdash M :A$ be a Boolean typing derivation and suppose
that the \circuittkm{} produces, over $\pi$, the final configuration
$(\pi,\mathbf{PSN}(\MM), E )$. 
Then, 
for any quantum+classical register $m$, 
the \evaltkm over $\pi$ and $m$ produces a 
$E$-indexed pseudo-distribution $\mu=\{\config_s^{\mu_s}\}_{s\in @^{\mathsf s}E}$ of final configurations such that 
$
\interp{E}(\tocpm(m)) =\tocpm(\mu)$.
\end{theorem}

\subsection{Termination and Unicity of the Normal Form for the Synchronous rule}
Similarly to how we did for the \circuittkmz, we will deduce termination and confluence properties for the \circuittkm from its simulation via the \evaltkm.

\begin{theorem}
  Given a type derivation $\pi : \Delta \vdash M : A$, the \circuittkm
  terminates on $\initcz(\pi)$. Moreover, it terminates in a final configuration.
\end{theorem}
\begin{proof}
  The proof is the same as for Theorem~\ref{thm:termination-circuit-tkm}, except
  that we use Theorem~\ref{thm:simulation-optimized-rule}.
\end{proof}

\begin{lemma}
	\label{lemma:config-related-same-distri-same-sem-optimized}

	Let $E, F, F'$ be extended circuits. 
	Let $\pi$ be a derivation and $\pi' \leq
	\pi$ a subderivation of $\pi$.
	For all inner \circuittkm configurations $\config_1, \config_2$ for $\pi'
	\leq \pi$ with extended circuit $F $ and $F' $, for all $E$-indexed coherent pseudo-distribution $\mu$ of inner $\evaltkm^*$ configurations for $\pi' \leq \pi$, if, for all distribution of quantum + classical registers $ \rho$,
$\mu \mathscr S^{\rho}_{ \pi} \config_1$ and $\mu \mathscr
	S^{\rho}_{ \pi} \config_2$ both hold, then $\config_1\equiv\config_2$.

\end{lemma}
\begin{proof}
	The proof is similar to
	Lemma~\ref{lemma:config-related-same-distri-same-sem}.
\end{proof}

While we do not give a proof of confluence, we can however show that, starting from the initial configuration, if we reach two different final configurations then they must be equivalent:

\begin{theorem}[Unicity of Final Configuration Modulo Equivalence]
	Given a type derivation $\pi : \Delta \vdash M : A$, such that we have two
	runs from $\initcz(\pi)$ to the two \emph{final} configurations $\config_1$ and
	$\config_2$, then $\config_1 \equiv \config_2$.
\end{theorem}
\begin{proof}
	By definition of $\initcz(\pi)$, we know that, for all quantum+classical
	register $m$, we have $\initm(\pi) \mathscr S_{\pi}^{m}
	\initcz(\pi).$
	Now, by our hypothesis, $\initcz(\pi)$ reduces to both final
	configurations $\config_1$ and $\config_2$. By
	Theorem~\ref{thm:simulation-optimized-rule} we know that $\initm(\pi)$
	reduces to two final configurations $\donfig_1, \donfig_2$ such that
	$\donfig_1 \mathscr S^{m}_{\pi} \config_1$ and $\donfig_2
	\mathscr S^{m}_{\pi\leq\pi} \config_2$. By confluence of
	$\evaltkm^*$ we have that $\donfig_1 = \donfig_2$. We can then conclude by
	Lemma~\ref{lemma:config-related-same-distri-same-sem-optimized}
\end{proof}
\section{Type System for Synchronous Compilation}

\begin{proof}[Proof of Proposition \ref{prop:graph1}]

To prove the first direction, let us suppose that $\mathscr G_\pi$ contains a cycle and take nodes $i,j$ such that $i\to^* j\to^* i$. 
Observe first that either $i$ or $j$ has to be an if-then-else rule. Indeed, if both $i$ and $j$ are either a gate axiom or an axiom $\mathbb A_i\vdash \mathbb A_i$ coming from an identity axiom, then from the cycle above we deduce the existence of a cyclic token path in the derivation, which is impossible. 
Let us suppose then, w.l.o.g., that $i$ is an if-then-else rule. Then there exists a negative position in the conclusion of $i$ that waits from a token coming from a positive position in the conclusion of $j$, yet this token may only arrive after all negative positions in the conclusion of $j$ are saturated, and this in turn requires a token to arrive from a positive position in the conclusion of $i$: we have thus reached a deadlock. 

For the converse direction, let us suppose that $\mathscr G_\pi$ is acyclic; let $n$ be the depth of $\pi$, that is, the maximum number of nested if-then-else rules (when two if-then-else rules are nested if one occurs in a sub-derivation whose conclusion is either the second or third premiss of the other); 
for any node $i$ of $\pi$, let  
$K(i)$ be the maximum $k\in \mathbb N$ such that $\mathscr G_\pi$ contains a path of the form $j_1\to\dots\to j_{k}\to i$ (this maximum is well-defined thanks to the acyclicity of $\mathscr G_\pi$). 

We show, by induction on the depth $n$, that tokens starting from the negative positions in the conclusion of $\pi$ and unit axioms reach, after a finite number of steps, all positive positions in the conclusion of $\pi$: 
\begin{itemize}

\item $n=0$: then $\pi$ contains no if-then-else rue: here we can simply conclude recalling the fact that, in the absence of if-then-else, GoI always terminates with the tokens reaching all positive positions in the conclusion of $\pi$;

\item $n>0$: we prove the following claim:
\begin{center}
(C) For any synchronization node $i$ at depth $0$, all negative positions in the conclusion of $i$ are reached after a finite number of steps.
\end{center}
From claim (C) the conclusion follows by applying the induction hypothesis: if $i$ is an if-then-else rule, and its negative positions are saturated, by IH we deduce that the tokens reach the positive positions in the conclusion of the sub-derivations, of depth $<n$, with conclusion the second and third premiss of $i$, and thus that the positions of $i$ are saturated; if $i$ is a gate or axiom rule, then tokens freely transit from negative to positive positions. We can thus conclude that tokens reach all positive positions in the conclusion of any synchronization point at depth 0, and must then reach all positive positions in the conclusion of $\pi$.

We prove claim (C) by induction on $K(i)$:
\begin{itemize}
\item $K(i)=0$: take a negative position in the conclusion of $i$, and consider the path starting from it backwards; 
we claim that either the path reaches a negative position in the conclusion of $\pi$, or it reaches a unit axiom. Indeed, the only alternative is that the path starts ``going down'' in the derivation and then ``turns up'', finally reaching some synchronization point $j$, but this would imply then that the edge $j\to i$ belongs to $\mathscr G_\pi$, against the induction hypothesis. 
We conclude then that the negative position is reached by a simple path starting from the conclusion of $\pi$ or a unit axiom.
\item $ K(i)>0$: let $j_1,\dots, j_q$ be all synchronization points such that $\mathscr G_\pi$ contains an edge $j_\ell\to i$;
by arguing as in the previous case, we deduce that all negative positions of $i$ are reached by a simple path coming from either a negative position in the conclusion of $\pi$, either a unit axiom, or some positive position in the conclusions of $j_1,\dots, j_q$.
By the sub-induction hypothesis all negative positions in $j_1,\dots, j_q$ are reached after a finite number of steps; 
for each $j_\ell$ that is not an if-then-else rule, this immediately implies that the tokens move, in one further step, onto the positive positions of $j_\ell$; if $j_\ell$ is an if-then-else rule, we can conclude the same 
by applying the induction hypothesis, since the second and third premiss of $j_\ell$ are conclusions of sub-derivation of depth $<n$.
We can thus conclude that all negative positions of $i$ are finally reached. 
\end{itemize}
\end{itemize}
\end{proof}

\begin{proof}[Proof of Theorem \ref{thm:graph2}]
We suppose that the indexing of rules is the same for the definition of $\mathscr G_{|\pi|}$ and the colored type derivation. Under this hypothesis we will actually prove an equality $\mathscr G_{|\pi|}=\mathscr G$, rather than just an isomorphism $\mathscr G_{|\pi|}\simeq\mathscr G$.

By induction on the typing rules we prove the following claims:
\begin{itemize}
\item[i.] for any negative position $\sigma$ in the conclusion of $\pi$, $\sigma$ has color $i$ iff there exists a simple path from $\sigma$ to a negative position in $i$;

\item[ii.] for any positive position $\sigma$ in the conclusion of $\pi$, $\sigma$ has color $i$ iff there exists a simple path from a positive position in $i$ to $\sigma$;

\item[iii.] $\mathscr G_{|\pi|}=\mathscr G$.
\end{itemize}

We only consider the two main cases:

\begin{itemize}

\item the last rule of $\pi$ is an application: 
\[
  \infer{\Delta, \Gamma \vdash^{\mathscr G_1\cup\mathscr G_2\cup\mathscr M_{A,A'}} M~N : B}
  {\infer*{\Delta \vdash^{\mathscr G_1} M : A\multimap B}{\pi_1} &\infer*{\Gamma\vdash^{\mathscr G_2} N : A'}{\pi_2} & |A|=|A'|}
\]
\begin{itemize}
\item[i.] a negative position $\sigma$ in $B$ can only move upwards to the corresponding negative position $\sigma'$ in $\pi_1$; by IH $\sigma$ has color $i$ iff there exists a simple path in $\pi_1$ to a negative position in some synchronization point $i$ of $\pi_1$; since no simple path from $\sigma$ can enter $\pi_2$ we can conclude 
that $\sigma$ has color $i$ iff there exists a simple path in $\pi$ to a negative position in some synchronization point $i$ of $\pi$.  

\item[ii.] One can argue similarly to the case above.

\item[iii.]
By induction hypothesis we have $\mathscr G_1=\mathscr G_{|\pi_1|}$ and  $\mathscr G_2=\mathscr G_{|\pi_2|}$. Observe that a simple path in $\pi$ is either a simple path in $\pi_1$, a simple path in $\pi_2$, or a simple path $\sigma_0,\dots,\sigma,\tau,\dots,\tau_1$, where $\sigma_0,\dots,\sigma$ is a simple path from some synchronization point $i$ in $\pi_1$ to some position $\sigma$ in $A$ (which must have color $i$ by IH), 
 $\tau,\dots, \tau_1$ is a simple path from some synchronization point $j$ in $\pi_2$ to the corresponding position $\tau$ in $A'$ (which has color $j$ by IH); notice that one of the following holds: either $\sigma$ is positive in $A$ (hence negative in $A\multimap B$), and $\tau$ is positive in $A'$, or $\sigma$ is negative in $A$ (hence positive in $A\multimap B$) and $\tau$ is negative in $A'$.
This implies then that $i\to j\in \mathscr M_{A,A'}$ and, more generally that any edge $i\to j\in \mathscr G_{|\pi|}$ belongs to either $\mathscr G_1, \mathscr G_2$ or $\mathscr M_{A,A'}$, i.e.~that 
$\mathscr G_{|\pi|}
\subseteq \mathscr G_1\cup \mathscr G_2\cup\mathscr M_{A,A'}$.

Conversely, given $i\to j\in \mathscr M_{A,A'}$, there are two possibilities: 
either $i$ is the color of some positive position $\sigma$ in $A$ (hence negative in $A\multimap B$), and $j$ is the colore of some positive position $\tau$ in $A'$, or $i$ is the color of some negative position $\sigma$ in $A$ (hence positive in $A\multimap B$) and $j$ is the color of some negative position $\tau$ in $A'$.
By IH we deduce that, in either cases there exists a simple path $\sigma_0,\dots,\sigma,\tau,\dots,\tau_1$,
 connecting a synchronization point $i$ in $\pi_1$ to a synchronization point $j$ in $\pi_2$, and thus $i\to j\in \mathscr G_{|\pi|}$.
We conclude then $
\mathscr G_1\cup\mathscr G_2\cup\mathscr M_{A,A'}\subseteq
\mathscr G_{|\pi|}$.

 \end{itemize}

\item The last rule of $\pi$ is an if-then-else:
\[
  \infer[(i)]{\Delta, |\Gamma|^i \vdash^{\mathscr G_1\cup\mathscr G_2\cup\mathscr G_3\cup\{ i\}} \ite{M}{N}{P} : |A|^i}
  {
  \infer*{\Delta \vdash^{\mathscr G_1} M : \bit^j}{\pi_1} & \infer*{ \Gamma^{\mathscr G_2} \vdash N : A }{\pi_2} &\infer*{\Gamma' \vdash^{\mathscr G_3} P : A'}{\pi_3} & |\Gamma|=|\Gamma'|,|A|=|A'|}
  \]
Conditions i.~and ii.~are trivially verified for a position $\sigma$ in either $\Gamma$ or $B$, and follow from the induction hypothesis applied to $\pi_1$ for a position $\sigma$ in $\Delta$.

For condition iii.~we have, by IH, $\mathscr G_1=\mathscr G_{|\pi_1|},\mathscr G_2=\mathscr G_{|\pi_2|},\mathscr G_3=\mathscr G_{|\pi_3|}$. Observe that a simple path in $\pi$ can only be a simple path in $\pi_1,\pi_2$ or $\pi_3$ or a trivial path with no edge in a position inside $i$. We conclude then that 
$\mathscr G_{|\pi|}=\mathscr G_{|\pi_1|}\cup \mathscr G_{|\pi_2|}\cup \mathscr G_{|\pi_3|}\cup\{i\}=
\mathscr G_{1}\cup \mathscr G_{2}\cup \mathscr G_{3}\cup\{i\}$.
\end{itemize}
\end{proof}

\end{document}